\documentclass[11pt]{article}

\usepackage{cmap} % Load before fontenc
\usepackage{booktabs} % for tables

\usepackage{qstyle}

\usepackage{fullpage,xspace, endnotes, caption, subcaption,times,latexsym}

\providecommand{\tabularnewline}{\\}

\newcommand{\enote}[1]{\textcolor{cyan}{\small {\textbf{(Eric:} #1\textbf{)}}}}

\begin{document} %%%
	%%%%%%%%%%%%%%%%%%%%
	%%%%%%%%%%%%%%%%%%%%

	%%%%%%%%%%%%%%%%
	% Front Matter %
	%%%%%%%%%%%%%%%%
	\title{Group coset monogamy games and an application\\ to device-independent continuous-variable QKD}
	\author{
		Eric Culf\thanks{Department of Mathematics and Statistics, University of Ottawa, Canada. Email: \texttt{eculf019@uottawa.ca}}
		\and Thomas Vidick\thanks{Faculty of Mathematics and Computer Science, The Weizmann Institute of Science and Department of Computing and Mathematical Sciences, California Institute of Technology. Email:  \texttt{thomas.vidick@weizmann.ac.il}}
		\and Victor V. Albert\thanks{Joint Center for Quantum Information and Computer Science, NIST and University of Maryland, USA. Email: \texttt{vva@umd.edu}}
		}
	\date{\today}%{August 12, 2021}
	\maketitle

	\begin{abstract}
		We develop an extension of a recently introduced subspace coset state monogamy-of-entanglement game [Coladangelo, Liu, Liu, and Zhandry; Crypto'21] to general group coset states, which are uniform superpositions over elements of a subgroup to which has been applied a group-theoretic generalization of the quantum one-time pad. We give a general bound on the winning probability of a monogamy game constructed from subgroup coset states that applies to a wide range of finite and infinite groups. To study the infinite-group case, we use and further develop a measure-theoretic formalism that allows us to express continuous-variable measurements as operator-valued generalizations of probability measures.

		We apply the monogamy game bound to various physically relevant groups, yielding realizations of the game in continuous-variable modes as well as in rotational states of a polyatomic molecule. We obtain explicit strong bounds in the case of specific group-space and subgroup combinations. As an application, we provide the first proof of one sided-device independent security of a squeezed-state continuous-variable quantum key distribution protocol against general coherent attacks.
	\end{abstract}

	\newpage
	%\tableofcontents
	%\newpage

	\section{Introduction \& summary of results}

	Quantum entanglement between several parties can be considered as a shared resource. The principle of \emph{monogamy} of entanglement informally states that two parties
	cannot be maximally entangled with each other if they are also entangled
	with a third.  One way to understand this restriction is through the notion of a \emph{monogamy
	game}, in which two players, Bob and Charlie, are tasked with simultaneously
	determining features of a state sent to them by a third
	party Alice. Both Bob and Charlie have to extract their feature correctly in order to win the game. Winning is not always possible since the desired
	features are stored in the multipartite correlations between Alice
	and the duo of Bob and Charlie, and those correlations cannot be made
	available to both Bob and Charlie due to entanglement monogamy. This can be seen as a consequence of the no-cloning property of quantum information.

	Monogamy games have been developed to prove strong statements about the
	viability of various quantum cryptographic protocols. The initial example of a monogamy game was used to prove that the BB'84 quantum key
	distribution (QKD) protocol \cite{BB84} is secure in a one-sided device-independent model,
	\emph{i.e.}\ secure even if one does not make assumptions about the receiver's measurement device \cite{tomamichel2013monogamy,Primaatmaja2022}. More recently, a monogamy game based on an extension of BB'84 states called \emph{subspace coset states} was introduced \cite{coladangelo2021hidden}, with
applications to uncloneable decryption and copy-protection
	of pseudorandom functions.
	This game can be seen as a variant of the original $n$-qubit game, and in \cite{CV22} an exponentially decaying (in the number of qubits $n$) bound was shown on the maximum winning probability in the game.

	We generalize
	this monogamy game in several directions by recasting the original
	game in group-theoretic terms and studying the resulting formulation for a variety of groups. Our reformulation reveals an unexpected link between the monogamy game and  states studied in quantum error
	correction. We prove several monogamy bounds for a wide range of groups, notably non-abelian groups and continuous (topological) groups. As a motivating application of our results, we highlight a proof of one-sided device-independent security for a continuous-variable quantum-key distribution that is resistant to coherent attacks; to the best of our knowledge, this is the first such proof.

	\subsection*{Coset states, group theory, and quantum error correction}

	We recall the key ingredient behind the subspace coset monogamy game of \cite{coladangelo2021hidden} --- the coset state --- in group-theoretic terms; generalizations of such states will serve as the backbone of our new monogamy games.

	Coset states were originally defined on an $n$-qubit space, whose computational basis states are labeled by elements of the group $G=\Z_2^n$. The coset states are
	\begin{equation}
	\ket{H,s,s'}\,=\,\frac{1}{\sqrt{|H|}}\sum_{u\in H}(-1)^{u\cdot s'}\ket{u+s}\,.\label{eq:one}
	\end{equation}
	Here, $H$ is a set of binary strings that is closed under binary
	addition --- a subgroup of $G$; this set labels which canonical basis
	elements participate in the superposition of the ``base'' subspace state
	$\ket{H,0,0}$. The binary strings $s,s'$ can be set,
	respectively, by applying appropriate Pauli $X$ and $Z$ operators
	on this base state (a.k.a.\ Pauli-twirling or one-time padding).
	In the group-theoretic interpretation, the string $s$ is a coset representative of the quotient $G/H$ that partitions $G$ into a union of cosets
		$s+H$, while $(-1)^{u\cdot s'} \equiv \gamma_{s'}(u)$ is an irreducible representation
		of $H$ labelled by the string $s'$.

	From the perspective of quantum error correction, the coset states~\eqref{eq:one} are precisely the code or error words of a
	CSS stabilizer quantum error-correcting code. Following related work for continuous-variable (CV) spaces \cite{GKP01}, such codes have recently
	been generalized \cite{ACP20} to the setting of group-valued Hilbert spaces for a general group $G$ spanned by
	$\{|g\rangle\,|\,g\in G\}$, where each canonical basis state is in
	one-to-one correspondence with an element of a group. As in the
	$n$-qubit case above, we can construct coset states $\ket{aH^{\gamma}_{m,n}}$, where $H$ is a subgroup of $G$, $a$ is a coset representative of $G/H$, and $\gamma_{m,n}$ is a matrix element of an irreducible representation of $H$.

	A wide variety of interesting states can be represented as $H\subset G$ coset states, notably, eigenstates of position and momentum quadratures of a CV mode ($\R^m \subset \R^n$), CV GKP states ($\Z\subset\R$) \cite{GKP01}, and quantum superpositions of orientations of an asymmetric molecule or, more generally, any 3D rigid body ($H\subset SO(3)$) \cite{ACP20}. For each group type, there is a connection to a stabilizer-like error-correcting code on the corresponding group space. We list all the group types considered in the paper, along with their corresponding codes, in Table \ref{tab:qec}.

	For fixed $H$, the set of coset states forms a complete set of states for the Hilbert space. If the group is finite, they are orthonormal as basis states in the Kronecker sense, and if the group is continuous they are orthogonal as distributions in the Dirac sense --- what is generally known as a Zak basis (see \cite[Appx. F]{ACP20} for more context). For fixed $H$ and $\gamma$, the set of such states houses an induced representation \cite{Arovas} of the parent group $G$.

	Since coset states of infinite groups are not normalizable, approximate or ``damped'' versions \cite{GKP01,Menicucci2014,ACP20,designs} have to be constructed in order to utilize such states in the lab. Thus, to formally discuss preparation of coset states, we need to take into account the damping operation used. However, we can avoid this when we discuss measurement in a basis of coset states. There we may represent the measurement process using an operator-valued generalization of a probability measure, thereby avoiding problems of normalization and convergence.

	\begin{table}
	\centering
	\begin{tabular}{lllll}
	\toprule
 & Space & Group & Subgroup & Related error-correcting code\tabularnewline
	\midrule
	Ref. \cite{coladangelo2021hidden} & $n$ qubits & $\Z_{2}^{n}$ & $\Z_{2}^{m}$ & qubit CSS  \cite{PhysRevA.54.1098,PhysRevLett.77.793,Steane1996b}\tabularnewline
	Sec. \ref{sec:U1game} & planar rotor & $U(1)$ & $\Z_{n}$ & rotor GKP  \cite{GKP01,ACP20}\tabularnewline
	Secs. \ref{sec:Cgame}, \ref{sec:Rngame} & $n$ modes & $\R^{n}$ & $\R^{m}$ & analog CSS  \cite{braunstein1998quantum,lloyd98,Gu_2009,eczoo_analog_stabilizer}\tabularnewline
	Sec. \ref{sec:Rgame} & single mode & $\R$ & $\Z$ & GKP \cite{GKP01}\tabularnewline
	Sec. \ref{sec:SO3game} & rigid body & $SO(3)$ & point group & molecular \cite{ACP20}\tabularnewline
	Sec. \ref{sec:finite} & finite group & $G$ & $H$ & \multirow{3}{*}{group GKP \cite{ACP20,PhysRevX.10.041018,eczoo_group_gkp}}\tabularnewline
	Sec. \ref{sec:abelian-infinite} & abelian group & $G$ & $H$ & \tabularnewline
	Sec. \ref{sec:compact} & compact group & $G$ & $H$ & \tabularnewline
	\bottomrule
	\end{tabular}

	\caption{List of group spaces and relevant subgroups for the coset monogamy games considered in this manuscript. Coset states for each space form code and error words of quantum error-correcting codes, listed in the last column of the table.}
	\label{tab:qec}
	\end{table}

	\subsection*{Generalized monogamy-of-entanglement games}

	We study general coset monogamy games from two perspectives, corresponding to an entanglement-based and a state-sending version of the game, respectively.

	The state-sending version is closer to the original subspace coset game of \cite{coladangelo2021hidden}. In the original multi-qubit case utilizing the abelian coset states (\ref{eq:one}), Alice sends the state $\ket{H,s,s'}$ with randomly chosen $H$, $s$, and $s'$.
	Bob and Charlie can split the state in an arbitrary way (including the application of an arbitrary CPTP map to it) and then separate.
	Once they are separated, they are each given a description of $H$ and tasked with determining $s$ and $s'$, respectively, up to the choice of representatives.
	In our non-Abelian finite-group case, Alice prepares and sends a coset state $\ket{aH^\gamma_{m,n}}$, and Bob and Charlie attempt to guess $a$ and $\gamma_{m,n}$, respectively.

	In the entanglement-based version, Bob and Charlie prepare a tripartite shared entangled state. Once all parties are separated, Alice measures her system in a randomly-chosen coset basis to get outcomes $a,\gamma_{m,n}$. She then informs Bob and Charlie of her measurement basis, and they make guesses of $a$ and $\gamma_{m,n}$, respectively.

	The entanglement-based game can be studied more directly, so we focus on that one in all the cases we consider. This especially important for infinite groups, as Alice's measurement can be expressed using an operator measure, hence avoiding discussion of the non-normalisable coset states. As in the case of $n$-qubit games~\cite{BL20,coladangelo2021hidden}, it is possible to transform a strategy for the state-sending game into a strategy for the entanglement-based game, which leads to a bound on the former as well. We work this relationship out formally in the abelian case, but note that it holds in the same way in the non-abelian case. The result of this transformation naturally provides a state-sending coset monogamy game where Alice sends the damped version of the coset states.

	\subsection*{Device-independent continuous-variable QKD}

	Inspired by the one-sided device-independent quantum key distribution (QKD) security proof introduced in~\cite{tomamichel2013monogamy}, we analyze a QKD protocol using continuous-variable (CV) coset states for $G=\R^{n}$ and $H=\R^{n/2}$ --- conceptually the closest continuous generalisation to the original $\Z_2^{n/2}\subset \Z_2^n$ qubit protocol \cite{CV22}. The CV protocol considered reduces to a Gaussian one: the unnormalizable coset states are infinitely squeezed states, but their damped versions are practically-realizable finitely squeezed states.
	%, leading us to show security for a family of CVQKD protocols based on such states (cf. \cite{PhysRevA.63.022309}).

	We show that these squeezed-state protocols are one-sided device independent (one-sided DI) secure against coherent attacks in the finite-key regime, making them the first CV protocols with such a level of security.
	Previous one-sided DI proofs of security for CVQKD protocols were limited to memoryless attackers~\cite{furrer2012continuous,gehring2015implementation}; indeed, overcoming this limitation for the case of discrete-variable (qubit-based) QKD protocols was one of the main motivations for the introduction of the $n$-qubit coset monogamy game in~\cite{tomamichel2013monogamy}.

	Our analysis leads to an error tolerance which is comparable to the one obtained for DV protocols in~\cite{tomamichel2013monogamy}. While our protocol, employing squeezed states (cf. \cite{PhysRevA.63.022309}), remains more challenging than the coherent-state based Gaussian CV protocols \cite{GGDL19}, the security benefits of one-sided device independence may outweigh the experimental challenges.

	% We show that these protocols are one-sided device independent (one-sided DI) secure against coherent attacks in the finite-key regime. Previous one-sided DI proofs of security for CVQKD protocols were limited to memoryless attackers~\cite{furrer2012continuous,gehring2015implementation}; indeed, overcoming this limitation for the case of discrete-variable (qubit-based) QKD protocols was one of the main motivations for the introduction of the $n$-qubit coset monogamy game in~\cite{tomamichel2013monogamy}. In particular, our analysis leads to an error tolerance which is comparable to the one obtained for DV protocols in~\cite{tomamichel2013monogamy}. While our protocol, employing squeezed states (cf. \cite{PhysRevA.63.022309}), remains more challenging than the coherent-state based Gaussian CV protocols \cite{GGDL19}, the security benefits of one-sided device independence may outweigh the experimental challenges.

	While we consider only the $G=\R^n$ protocol in detail in this manuscript, we note that similar protocols should be possible for the other group spaces, and do not see an obstruction to proving analogous device-independent security for such protocols. In particular, our general formulation should allow for QKD protocols utilizing GKP states ($G=\R^n$ and $H\cong \Z^n$).
	Moreover, our formulation paves the way for analyzing more general subgroups of $G=\R^n$ that form degenerate lattices or products of lattices and planes, with the former corresponding to the recently developed GKP-stabilizer codes \cite{PhysRevLett.125.080503} that protect an entire logical mode against small fluctuations in all physical modes.

		\subsection*{Acknowledgements}

		EC acknowledges the support of an NSERC CGS M grant, and thanks Florence Grenapin and Jason Crann for interesting discussions on this topic.
		EC and VA acknowledge Alexander Barg for the suggestion to use algebraic-geometric codes for the QKD protocol.
		TV is supported by a grant from the Simons Foundation (828076, TV) and a research grant from the Center for New Scientists at the Weizmann Institute of Science.
		VVA acknowledges financial support from NSF QLCI grant OMA-2120757, and thanks Olga Albert and Ryhor Kandratsenia for providing daycare support throughout this work.
		Contributions to this work by NIST, an agency of the US government, are not subject to US copyright. Any mention of commercial products does not indicate endorsement by NIST.

		\subsection*{Outline}

		In Sections \ref{sec:U1game}-\ref{sec:compact} we develop monogamy games for the group spaces listed in Table \ref{tab:qec}, along with the mathematical formalism necessary to tackle other continuous and infinite groups. The discussion of the game in each section is meant to be stand-alone with only the proof of the winning probability bounds relying on the general results of Sections \ref{sec:abelian-infinite}-\ref{sec:compact}. The sections are intended to proceed in approximate order of mathematical difficulty.
		The squeezed-state device-independent QKD protocol is developed in Section \ref{subsec:qkd}, and Figure \ref{fig:qkds-favourite} plots the derived asymptotic error tolerance vs.\,key rate. In \cref{sec:POVM}, we work out the measure-theoretic formalism of integration with respect to an operator-valued measure. In \cref{sec:damping}, we formulate general damping operators and provide the relationship between such operators and maximally-entangled states that allows us to study state-sending versions of the games.

		Some of the sections that follow are technical.
		We recommend that readers unfamiliar with monogamy games first read about the qubit game \cite{coladangelo2021hidden} and then continue with the planar-rotor or CV two-mode generalizations in Secs. \ref{sec:U1game} and \ref{sec:Cgame}, respectively.
		Readers interested in qudit games based on non-abelian groups should consult \cref{sec:finite}.
		Readers interested in learning how to rigorously handle continuous-parameter measurements on noncompact infinite-dimensional spaces may skip to the locally compact abelian group games in \cref{sec:abelian-infinite} and associated mathematical details in \cref{sec:POVM}.
		Mathematically inclined readers interested in our general games may jump to the compact-group formalism in \cref{sec:compact}.
		Readers interested in the QKD protocol may go to either the two-mode CV warmup in \cref{sec:Cgame} or the $n$-mode CV games in \cref{sec:Rngame}, the latter also containing the QKD protocol.

	%%%%%%%%%%%%%%%%%%%%%%%%%%%%%%%%%%%%%%%%%%%%%%%%%%%%%%%%%%%%%
	\section{The coset monogamy game on $U(1)$}\label{sec:U1game}
	%%%%%%%%%%%%%%%%%%%%%%%%%%%%%%%%%%%%%%%%%%%%%%%%%%%%%%%%%%%%%

	We introduce the group-valued space of the planar rotor $G=U(1)$ and its associated coset states and monogamy game.

	\subsection{Planar rotor states}

	Systems confined to rotate in a two-dimensional plane may be described as a planar rotor. For such a system, the set of classical states can be represented by the group of rotations in the plane $G=U(1)$. There are various ways to work with this group, but we will consider it as $U(1)=\R/2\pi\Z$ and make use of the set of representatives $[0,2\pi)$. Since the space is continuous, the Hilbert space of quantum states is the space of square-integrable functions $L^2(U(1))$. The inner product on this space is provided by the Haar measure -- the unique normalized measure invariant under the action of the group -- which in the case of $U(1)$ takes the form
	\begin{align}
		\braket{\psi}{\phi}=\frac{1}{2\pi}\int_0^{2\pi}\overline{\psi(x)}\phi(x)dx\;.
	\end{align}

	The Fourier series provides the canonical orthonormal basis of $L^2(U(1))$: the basis of states $\ket{\ell}$ for $\ell\in\Z$ given by functions $\psi_\ell(x)=e^{i\ell x}$ for $x\in[0,2\pi)\cong U(1)$. This corresponds to the basis of angular momentum eigenstates of a $U(1)$ system. Dual to this basis are the position eigenstates $\ket{\theta}$ given by $\psi_\theta(x)=\delta_{U(1)}(\theta-x)=2\pi\delta(\theta-x)$ for $\theta\in[0,2\pi)$, which satisfy the generalised orthonormality condition $\braket{\theta}{\theta'}=2\pi\delta(\theta-\theta')$. These are however not states, since the Dirac delta is not a function, so they cannot be normalized. Accordingly, it takes some care to work with this kind of state. We will approach this in two different but complementary ways. First, we can consider the basis not as a set of physical states but as a measurement, which allows us to treat it in a measure-theoretic way. For any state $\ket{\psi}\in L^2(U(1))$, the probability of measuring a position in some set $E\subseteq U(1)$ is \begin{align}
		\frac{1}{2\pi}\int_E\abs*{\braket{\theta}{\psi}}^2d\theta=\braket{\psi}{\frac{1}{2\pi}\int_E\ketbra{\theta}d\theta}{\psi}.
	\end{align}
	In this way, we take the operator measure of a (Borel measurable) set in $U(1)$ as the operator $A^{U(1)}(E):=\frac{1}{2\pi}\int_E\ketbra{\theta}d\theta$, which can be seen a continuous-variable generalization of a projective measurement or an operator-valued generalization of a probability measure. Note also that these are well-defined operators on $L^2(U(1))$, acting as
	\begin{align}
		(A^{U(1)}(E)\ket{\psi})(x)=\chi_E(x)\psi(x)=\begin{cases}\psi(x)&x\in E\\0&x\notin E\end{cases}.
	\end{align}
	In the operator measure picture, the completeness of the basis is expressed by showing that the measure is a POVM --- $A^{U(1)}(U(1))=\Id$. The mathematical formalism of operator measures is worked out in \cref{sec:POVM}.

	The other way we approach unnormalizable states is by damping, that is we act on the state by an operator that makes it normalizable. A common way to do this is to replace the deltas by Gaussians:
	\begin{align}
		\ket{\theta}\mapsto\sqrt{\frac{a}{\pi}}\int_{\theta-\pi}^{\theta+\pi}e^{-a(x-\theta)^2}\ket{x}dx.
	\end{align}
	In the limit $a\rightarrow\infty$, this returns to the original delta function; so, for large $a$, this provides a very good approximation to the behaviour of the position eigenstates despite being normalizable. We formalise this in \cref{sec:damping} and work out how to pass from operator measures to damped states without going through an unnormalizable basis.

	The position and momentum bases, though disparate, are particular cases of the same construction, the \emph{coset state basis}. A coset state basis is a generally unnormalizable basis corresponding to a closed subgroup of $G$; the position basis corresponds to the subgroup of all elements $G$ and the momentum basis corresponds to the trivial subgroup $\{0\}$. The remaining closed subgroups are $\Z_n\leq U(1)$ for $n\in\N$, groups of rotations by multiples of $2\pi/n$. Fixing $n\in\N$, we define the $\Z_n$-\emph{subgroup state} as the uniform superposition over all elements of $\Z_n$,\footnote{We use a different normalization convention than previous work \cite[Eq. (124)]{ACP20} throughout the paper, resulting in rescaled Dirac-delta functions on relevant quotient spaces. Our convention translates to using a normalized Haar measure in the case of compact $G$, $\frac{1}{|\Z_n|}\sum_{x\in \Z_n}\to \int_G dg$, while in the previous convention, $\frac{1}{\sqrt{|\Z_n|}}\sum_{x\in \Z_n}\to \frac{1}{\sqrt{|G|}}\int_G dg$ with $|G|$ the group volume.}
	\begin{align}
		\ket{\Z_n}=\frac{1}{|\Z_n|}\sum_{x\in \Z_n}\ket{x}=\frac{1}{n}\sum_{k\in\Z_n}\ket{2\pi k/n}\;.
	\end{align}
	Note also that for the position and angular momentum bases, $\ket{\{0\}}=\ket{0}$ and $\ket{U(1)}=\frac{1}{2\pi}\int\ket{x}dx=\ket{\ell=0}$. To extend the subgroup state to a basis we orthogonalise in two ways: use superpositions with orthogonal supports in the basis, and introduce phases. The canonical way to move to an orthogonal support is to consider analogous superpositions over a \emph{coset} of $\Z_n$ rather than the group itself. The cosets are the equivalence classes $U(1)/\Z_n$, which as $U(1)$ is abelian form a group $U(1)/\Z_n\cong U(1)$. Thus, for $x+\Z_n\in U(1)/\Z_n$, the coset state $\ket{x+\Z_n}=\frac{1}{n}\sum_{y\in \Z_n}\ket{x+y}$. Also, we introduce phases given by the dual group $\hat{\Z}_n$ of $\Z_n$, the group of continuous group homomorphisms $\Z_n\rightarrow\set*{z\in\C}{|z|=1}$ under multiplication. For the usual representation as a quotient of $\Z$, the dual group of $\Z_n$ is $\hat{\Z}_n\cong\Z_n$ with action given by $\gamma_k(m)=e^{2\pi i\frac{km}{n}}$. In the same way, for $\Z_n$ seen as a subgroup of $U(1)$, the action $\gamma_k(x)=e^{ikx}$. Then, the subgroup coset states are defined by
	\begin{align}
		\ket{n,x,k}:=\ket{x+\Z_n^{\gamma_k}}=\frac{1}{n}\sum_{y\in \Z_n}\gamma_k(y)\ket{x+y}=\frac{1}{n}\sum_{m\in\Z_n}e^{2\pi i\frac{km}{n}}\ket{x+2\pi m/n}\;.
	\end{align}
	These are unnormalizable, but they are orthogonal in the sense that
	\begin{align}
	\begin{split}
		\braket{n,x,k}{n,x',k'}&=\frac{1}{n^2}\sum_{y,y'\in \Z_n}e^{i(k'y'-ky)}\delta_{U(1)}((x+y)-(x'+y'))\\
		&=\delta_{U(1)/\Z_n}(x-x'+\Z_n)\frac{1}{n}\sum_{y\in \Z_n}e^{i(k'(y+(x-x'))-ky)}\\
		&=\delta_{U(1)/\Z_n}(x-x'+\Z_n)\delta_{k,k'}\;,
	\end{split}
	\end{align}
	and complete in the sense that the position eigenstates are contained in their span. In general, the definition of the coset state basis depends on a choice of coset representatives, but since this only changes the states up to global phase, we do not need to consider it. This definition directly extends the coset states of a finite group. Again, to work with them more rigorously, we can consider the coset measure they induce. Now, as the basis is indexed by $U(1)/\Z_n\times \hat{\Z}_n\cong U(1)\times\Z_n$, the coset operator measure is an operator-valued measure on that set. For Borel measurable $E\subseteq U(1)/\Z_n\times\Z_n$,
	\begin{align}
		A^{\Z_n}(E)=n\int_E\ketbra{n,x,k}d(x+\Z_n,k)=\frac{n}{2\pi}\sum_{k=0}^{n-1}\int_{E_n}\ketbra{n,x,k}dx\;,
	\end{align}
	where we write $E=E_0\times\{0\}\cup\ldots\cup E_{n-1}\times\{n-1\}$. The additional coefficient $n$ is required because the dual measure on $\hat{\Z}_n$ is normalized so that $\mu(\{\gamma_0\})=1$ and not $\mu(\hat{\Z}_n)=1$. Again, this provides a well-defined operator:
	\begin{align}
		\braket{\phi}{A^{\Z_n}(E)}{\psi}=\frac{1}{2\pi n}\sum_{k=0}^{n-1}\int_{E_k}\sum_{y,y'\in \Z_n}e^{ik(y-y')}\bar{\phi}(x+y)\psi(x+y')dx\;.
	\end{align}
	We note again that this measure satisfies $A^{\Z_n}(U(1)/\Z_n\times\Z_n)=\Id$, which is equivalent to completeness of the basis. For a general (locally compact) abelian group, the coset measure is formally introduced in \cref{sec:abelian-coset-measure}.

	\subsection{Monogamy game and winning probability}

	We can use these planar-rotor coset states to play a monogamy game inspired by the strong monogamy game of \cite{coladangelo2021hidden}. Here, we describe the entanglement-based version of the game, which can be understood using the coset measure. %, and then pass to the more easily-realizable state-sending version of the game, which involves damped coset states. In either version,
	In this game, two cooperating players, Bob and Charlie, play against an honest referee, Alice. Let $p_1<\ldots<p_N$ be a set of distinct primes and let $0<\varepsilon<\frac{\pi}{p_N^2}$. The game proceeds as follows:

	\begin{enumerate}[1.]
		\item Bob and Charlie prepare a shared state $\rho_{ABC}$ but then are no longer allowed to communicate.

		\item Alice chooses $j=1,\ldots,N$ uniformly at random and measures her register in the basis $\set*{\ket{p_j,x,k}}$ to get measurements $x\in U(1)/\Z_{p_j}$, $k\in\Z_{p_j}$.

		\item Alice sends $p_j$ to Bob and Charlie. Bob answers with a guess $x_B$ for $x$ and Charlie answers with a guess $k_C$ for $k$.

		\item Bob and Charlie win if $|x-x_B|<\varepsilon$ in $U(1)/\Z_n$ and $k=k_C$.
	\end{enumerate}

	In this section, we do not formally introduce strategies for this game. Nevertheless, we note that Charlie makes a measurement with a finite number of outcomes, so his measurement may in general be expressed by a POVM. On the other hand, Bob has infinitely many (in fact continuously many) measurement outcomes. This means that Bob's measurement is modelled by an operator-valued measure. Also, Bob's winning condition is slackened compared to the game with finite information, as he needs to only guess within a neighborhood. This is because we cannot expect Bob to answer with infinite precision in a continuous space. First, this is physically infeasible as he would need to transmit infinitely many bits to Alice, and also the winning probability would always be $0$, as the space of correct answers would have measure $0$.

	General games of this form for abelian groups are introduced in \cref{sec:abelian-coset-game}; they are parametrized by the underlying group $G$, the collection of subgroups $\mc{S}$, as well as Bob and Charlie's neighborhoods of correct answers $E\subseteq G$ and $F\subseteq\hat{G}$. For the above game, $\ttt{G}_{N,\varepsilon}$, we have $G=U(1)$, $\mc{S}=\set*{U_{p_1},\ldots,U_{p_N}}$, $E=(-\varepsilon,\varepsilon)$, and $F=\{\gamma_0\}$. We make use of the general bound of \cref{thm:bound-abelian} to find an upper bound on the winning probability of $\ttt{G}_{N,\varepsilon}$:
	\begin{align}
		\mfk{w}(\ttt{G}_{N,\varepsilon})\,\leq\,\frac{1}{N}\sum_{j=1}^{N}\sup_{H\in\mc{S},g\in G}\sqrt{\frac{\abs*{H\cap (g+E+\pi_j(H))}~\abs*{F}}{|H|}}\;.
	\end{align}
	In this formula, $\pi_j:\mc{S}\rightarrow\mc{S}$ is a collection of orthogonal permutations --- bijections such that $\pi_j(H)=\pi_k(H)$ only if $j=k$. For any set, there always exists such a family, in our case the cycles $\pi_j(U_{p_k})=\pi_j(U_{p_{k+j}})$, which is the family we will use to get the bound on the winning probability of $\ttt{G}_{N,\varepsilon}$.

	\begin{theorem}
		Let $p_1<\ldots< p_N$ and $0<\varepsilon\leq\frac{\pi}{p_N^2}$. The winning probability of the $U(1)$ monogamy game satisfies
		\[\mfk{w}(\ttt{G}_{N,\varepsilon})\,\leq\,\frac{1}{N}+\frac{1}{\sqrt{p_1}}\;.\]
	\end{theorem}

	\begin{proof}
		First, for any $j=1,\ldots,N$ and $\gamma\in \hat{G}$, $\abs*{F}=\abs*{\{\gamma_0\}}=1$. Also, for any $j\neq k$ and $2\pi r\in U(1)$,
		\begin{align}
		\begin{split}
			\abs{\Z_{p_j}\cap(2\pi r+E+\Z_{p_k})}&=\abs*{\set*{n\in\Z_{p_j}}{\exists\;m\in\N.\abs*{\tfrac{n}{p_j}-r-\tfrac{m}{p_k}}<\frac{\varepsilon}{2\pi}}}\\
			&\leq\abs*{\set*{[n]\in\Z_{p_j}}{\exists\;m\in\N.\abs*{np_k-rp_jp_k-mp_j}<\tfrac{1}{2}}}\leq 1,
		\end{split}
		\end{align}
		as the interval $(rp_jp_k-\frac{1}{2},rp_jp_k+\frac{1}{2})$ contains at most one integer, and $np_k$ is equal to this integer modulo $p_j$ for at most one value of $n$ by inversion modulo $p_j$. Thus,  as $\pi_0$ is the identity,
		\begin{align}
			\mfk{w}(\ttt{G}_{N,\varepsilon})\leq\frac{1}{N}+\frac{1}{N}\sum_{i=2}^N\sup_{k}\frac{1}{\sqrt{p_k}}\leq\frac{1}{N}+\frac{1}{\sqrt{p_1}}.
		\end{align}
	\end{proof}

  %%%%%%%%%%%%%%%%%%%%%%%%%%%%%%%%%%%%%%%%%%%%%%%%%%%%%%%%%%
	\section{The coset monogamy game on $\C$}\label{sec:Cgame}
  %%%%%%%%%%%%%%%%%%%%%%%%%%%%%%%%%%%%%%%%%%%%%%%%%%%%%%%%%%

	We consider the group state space $G=\C$. This can represent the space of two independent oscillators, or the wavefront of a laser. The subgroups we consider are copies of $\R$, corresponding to rotated single modes embedded in the parent two-mode space $\C\cong \R^2$. As opposed to the planar rotor of the previous section, the spaces of \textit{both} the coset representative and the irreducible representation are continuous and non-compact, requiring appropriate regularization.

	\subsection{Coset states}

	The Hilbert space on this group is the space of square-integrable functions $L^2(\C)$, to which we can associate an unnormalizable basis $\set*{\ket{z}}{z\in\C}$. This basis satisfies orthogonality $\braket{z}{z'}=\delta_{\C}(z-z')=\delta(z_r-z'_r)\delta(z_i-z'_i)$, where $z_r$ and $z_i$ are the real and imaginary parts of $z$. The Haar integral is the usual Lebesgue integral over $\C\cong\R^2$, and the dual $\hat{G}\cong\C$ acts as $\gamma_w(z)=e^{2\pi i\latRe(\bar{w}z)}$ for $w\in\C$.

	We consider subgroups corresponding to lines in the plane, of the form $H=\zeta\R$ for $\abs{\zeta}=1$. %Note that there are many other subgroups of interest, such as lattices.
	Then, via the restriction of the isomorphism above, $\hat{H}\cong\bar{\zeta}\R$. Also, we have $\C/\zeta\R\cong\R$ via $z+\zeta\R\mapsto\latIm(\bar{\zeta}z)$. As such, the coset states can be indexed by $(x,y)\in\R^2$, corresponding to coset $i\zeta x+\zeta\R$ and character $\gamma_{\bar{\zeta}y}$, and take the form
	\begin{align}\label{eq:damped-two-mode}
		\ket{\zeta,x,y}:=\ket{i\zeta x+\zeta\R^{\gamma_{\bar{\zeta}y}}}=\int_{\R} e^{2\pi iyr}\ket{\zeta(r+ix)}dr.
	\end{align}

	To work with normalized damped states, we can use Gaussian damping again. Fix $b>a>0$ and, writing $\tilde{a}=\frac{ab}{b-a}$, define $\Delta_{a,b}:L^2(\C)\rightarrow L^2(\C)$ as
	\begin{align}\label{eq:damping}
		\Delta_{a,b}\ket{z}=e^{-\tilde{a}|z|^2}\int_\C e^{-b|w-z|^2}\ket{w}dw.
	\end{align}
	Note that in order to have this work effectively, it must be a product of two Gaussians because there are two infinities that need to be damped: "position," since $\C$ and $\R$ are not compact, and "momentum," since the duals are not compact, giving rise to delta functions. In the notation of \cref{sec:damping}, this can be used to define a damping sequence $\parens*{\frac{\Delta_{a_n,b_n}^\dag}{\norm{\Delta_{a_n,b_n}}}}$ using sequences where $b_n\rightarrow\infty$ and $a_n\rightarrow 0$. The normalized damped states then have the form
	\begin{align}
		\ket{\zeta,x,y|_{a,b}}:=\frac{c\Delta_{a,b}\ket{\zeta,x,y}}{\norm{\Delta_{a,b}\ket{\zeta,x,y}}}=\sqrt{\frac{2\sqrt{ab}}{\pi}}\int_{\C}e^{-aw_r^2-b(w_i-x)^2-2\pi i\parens*{1-\frac{a}{b}}w_ry}\ket{\zeta w}dw,
	\end{align}
	where $c$ is the usual complex conjugate $(c\ket{\psi})(z)=\overline{\psi(z)}$. Seeing $\C\cong\R^2$ as the space of two oscillators, these states may be see as two-mode squeezed states, as illustrated in \cref{fig:c-state}. In the state-sending picture, the norms of the damped states give rise to the probability distribution parametrizing the choice of state,
	\begin{align}
		\pi^{\zeta}_{a,b}(x,y)=\frac{\norm*{\Delta_{a,b}\ket{\zeta,x,y}}^2}{\norm*{\Delta_{a,b}}_2^2}=2\sqrt{\tfrac{a}{b}}e^{-2\frac{ab}{b-a}x^2-2\pi^2\frac{b-a}{b^2}y^2}=2\sqrt{\tfrac{\tilde{a}}{\tilde{a}+b}}e^{-2\tilde{a}x^2-\frac{2\pi^2}{\tilde{a}+b}y^2},
	\end{align}
	consisting of independent normal distributions in $x$ and $y$ with means both $0$ and variances $\frac{b-a}{4ab}$ and $\frac{b^2}{4\pi^2(b-a)}$, respectively.

	\begin{figure}
		\centering
		\begin{tikzpicture}
			\node at (0,0) {\includegraphics[scale=0.75]{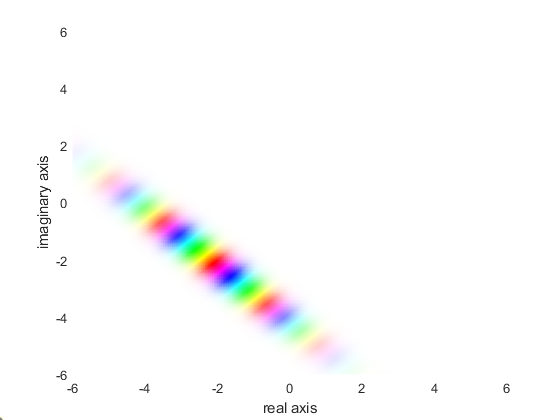}};
			\draw[opacity=0.5] (0.2,3.5) -- (0.2,-3.3);
			\draw[opacity=0.5] (4.5,0.15) -- (-4,0.15);
			\draw[-|] (0.2,0.15) -- (-1.2,-1.1) node[pos=0.5,below right]{$x$};
			\draw[|-|] (-2.2,-1) -- (-1.3,-1.85) node[pos=0.5,below left=-0.1cm]{$\frac{1}{(1-a/b)y}$};
			\draw[-Latex] (0.2,0.65) arc (90:223:0.5) node[pos=0.5,above left]{$\theta$};
			\draw[Latex-Latex] (2.2,1) node[below right=-0.2cm]{$\sim\frac{1}{a}$} -- (0.8,2.35);
			\draw[Latex-Latex] (1.2,1.375) -- (1.8,1.975) node[above right=-0.2cm]{$\sim\frac{1}{b}$};
		\end{tikzpicture}

		\caption{An example of the coset state $\ket{\zeta,x,y|_{a,b}}$ with $a=0.1$, $b=4$, $\zeta=e^{i\theta}=e^{i\frac{3\pi}{4}}$, $x=3$, and $y=0.5$.}
		\label{fig:c-state}
	\end{figure}
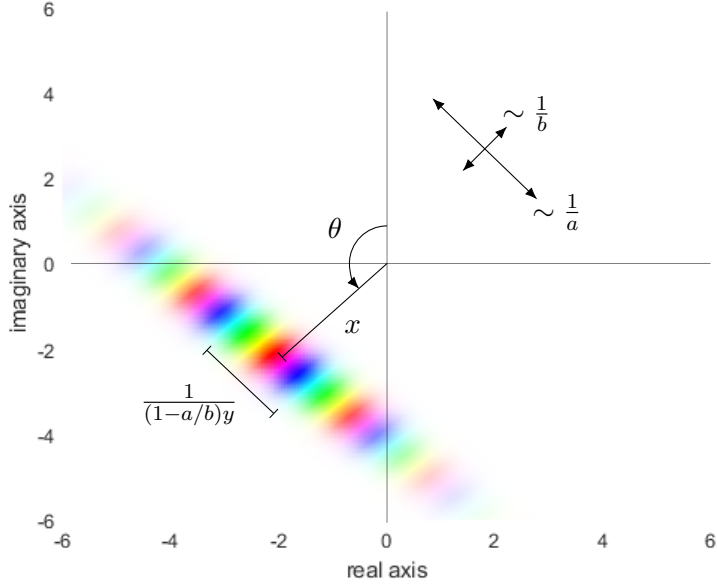

	\subsection{Monogamy game analysis}

	We can construct a monogamy-of-entangelement game based on the coset states above. The game is played by two cooperating players, Bob and Charlie, against an honest referee, Alice.  For fixed $\delta,\varepsilon>0$ and a finite collection $\zeta_1,\ldots,\zeta_n$ such that $|\zeta_i|=1$, the gameplay proceeds as follows.

	\begin{enumerate}[1.]
		\item Bob and Charlie prepare a shared state $\rho_{ABC}$ but then are no longer allowed to communicate.

		\item Alice chooses $i$ uniformly at random and measures her register in basis $\set*{\ket{\zeta_i,x,y}}$ to get measurements $x,y\in\R$.

		\item Alice sends $\zeta_i$ to Bob and Charlie. Bob answers with a guess $x_B$ for $x$ and Charlie answers with a guess $y_C$ for $y$.

		\item Bob and Charlie win if $|x-x_B|<\delta$ and $|y-y_C|<\varepsilon$.
	\end{enumerate}

	We bound the winning probability of this game using \cref{thm:bound-abelian}.

	\begin{theorem}\label{thm:c-bound}
		Fix $n\in\N$ divisible by $4$, and take the set of subgroups $\mc{S}=\set*{e^{2\pi i\frac{k}{n}}\R}{k=0,\ldots,n-1}$, and the sets describing the precision $E=\set*{z\in\C}{|z|<\delta}$ and $F=\set*{\gamma_z}{|z|<\varepsilon}$. Let the abelian coset monogamy game $\ttt{G}_{n,\delta,\varepsilon}=(G,\mc{S},E,F)$. Then, the winning probability satisfies
		\begin{align}
			\mfk{w}(\ttt{G}_{n,\delta,\varepsilon})\,\leq\,\frac{2}{n}+4\parens*{1+\frac{1}{n}}\sqrt{\delta\varepsilon}\;.
		\end{align}
	\end{theorem}

	First, we need to bound the overlap measures.
	\begin{lemma}\label{lem:c-overlaps}
		Let $\zeta,\xi\in\C$ such that $|\zeta|=|\xi|=1$, and let $z\in\C$. Then, the measures
		\begin{itemize}
			\item $\mu_{\widehat{\zeta\R}}(F)=2\varepsilon$
			\item $\mu_{\zeta\R}(\zeta\R\cap(z+E+\xi\R))\leq\frac{2\delta}{|\latIm(\bar{\zeta}\xi)|}$
		\end{itemize}
	\end{lemma}

	\begin{proof}
		For the first point, we use the isometric isomorphisms $\widehat{\zeta\R}\cong\bar{\zeta}\R\cong\R$ to get
		\begin{align}
			\mu_{\widehat{\zeta\R}}(F)=\mu_{\bar{\zeta}\R}(\set*{\bar{\zeta}\latRe(\zeta z)}{|z|<\varepsilon})=\mu_{\R}((-\varepsilon,\varepsilon))=2\varepsilon.
		\end{align}
		For the second, we begin similarly and get $\mu_{\zeta\R}(\zeta\R\cap(z+E+\xi\R))=\mu_\R(\R\cap(\bar{\zeta}z+E+\bar{\zeta}\xi\R))$. Now we note that this set is in fact an interval, so its measure is again its length. If $\bar{\zeta}\xi$ is real, then the two lines are parallel, giving measure $0$ or $\infty=\frac{2\delta}{|\latIm(\bar{\zeta}\xi)|}$. Else, writing $\bar{\zeta}\xi=e^{i\theta}$, the overlap is always the hypotenuse of a right triangle with side length $2\delta$ and angle $\theta$ (see \cref{fig:trig}), giving measure $\frac{2\delta}{|\sin2\theta|}=\frac{2\delta}{|\latIm(\bar{\zeta}\xi)|}$.
	\end{proof}

	\begin{figure}[h!]\label{fig:trig}
		\centering
		\begin{tikzpicture}
			\fill[lightgray] (-1.5,-2) -- (0.5,2) -- (1.5,2) -- (-0.5,-2) -- cycle;
			\draw[thick] (-5,0) -- (5,0) node[pos=0.85,above]{$\R$};
			\draw[thick] (-1.5,-2) -- (0.5,2)  node[left,pos=0.8]{$E+\bar{\zeta}\xi\R$};
			\draw[thick] (-0.5,-2) -- (1.5,2);

			\draw[thick] (0.5,0) -- (-0.3,0.4) node[above, pos=0.35]{$2\delta$};
			\draw (-0.2,0.35) -- (-0.25,0.25) -- (-0.35,0.3);

			\draw[very thick,red] (-0.5,0) -- (0.5,0);
			\draw (0.8,0) arc (0:60:0.3) node[right, pos=0.6]{$\theta$};
		\end{tikzpicture}
		\caption{The geometry of the overlap of \cref{lem:c-overlaps}. The overlap is given in red.}
	\end{figure}
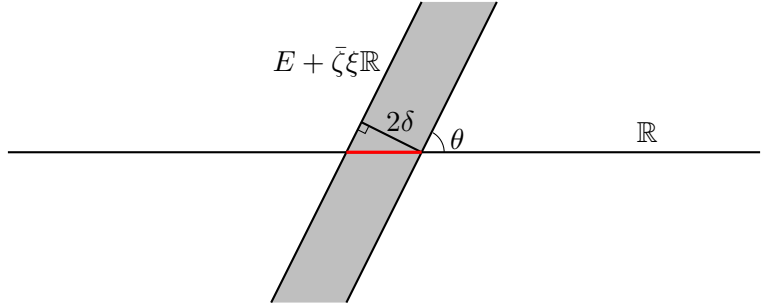

	Now, we can prove the main result.

	\begin{proof}[Proof of \cref{thm:c-bound}]
		To make use of \cref{thm:bound-abelian}, we choose a set of orthogonal permutations of $\mc{S}$. For $j=0,\ldots,n-1$, let $\pi_j(e^{2\pi i\frac{k}{n}}\R)=e^{2\pi i\frac{k+j}{n}}\R$. Then, using the theorem, we get the bound
		\begin{align}
		\begin{split}
			\mfk{w}(\ttt{G}_{n,\delta,\varepsilon})&\leq\expec{j}\sup_{H\in\mc{S},g\in G}\min\set*{1,\sqrt{\mu_H\parens[\big]{H\cap gE\pi_j(H)}\mu_{\hat{H}}\parens[\big]{F}}}\\
			&\leq\frac{2}{n}+\frac{1}{n}\sum_{j\neq0,\tfrac{n}{2}}\sup_{k;z,w\in\C}\sqrt{\mu_{e^{2\pi ik/n}\R}(\zeta\R\cap(z+E+e^{2\pi i(j+k)/n}\R))\mu_{\widehat{e^{2\pi ik/n}\R}}(F)}.
		\end{split}
		\end{align}
		Finally, using \cref{lem:c-overlaps} and some simple bounds, we get the wanted bound
		\begin{align}
		\begin{split}
			\mfk{w}(\ttt{G})&\leq\frac{2}{n}+\frac{1}{n}\sum_{j\neq0,\frac{n}{2}}\sup_{k,z,w}\sqrt{\frac{2\delta}{|\latIm(e^{-2\pi i\frac{k}{n}}e^{2\pi i\frac{k+j}{n}})|}2\varepsilon}\\
			&=\frac{2}{n}+2\frac{2\sqrt{\delta\varepsilon}}{n}\frac{1}{\sqrt{\sin\frac{\pi}{2}}}+4\frac{2\varepsilon}{n}\sum_{j=1}^{\frac{n}{4}-1}\frac{1}{\sqrt{\sin(2\pi\frac{j}{n})}}\\
			&\leq\frac{2}{n}+\frac{4\sqrt{\delta\varepsilon}}{n}+8\sqrt{\delta\varepsilon}\int_{\frac{1}{n}}^{\frac{1}{4}}\frac{1}{\sqrt{\sin(2\pi x)}}dx\leq\frac{2}{n}+\frac{4\sqrt{\delta\varepsilon}}{n}+4\sqrt{\delta\varepsilon}\int_{0}^{\frac{1}{4}}\frac{1}{\sqrt{x}}dx\\
			&=\frac{2}{n}+\frac{4\sqrt{\delta\varepsilon}}{n}+4\sqrt{\delta\varepsilon}.
		\end{split}
		\end{align}
	\end{proof}

	We close this section by presenting the state-sending version of the game. This allows for applications where Alice sends a damped version of the coset states rather than making a measurement, which would generally require the shared state to be entangled. The game proceeds as follows.

	\begin{enumerate}[1.]
		\item Alice chooses $i$ uniformly at random and samples $(x,y)$ according to the distribution $\pi^{\zeta_i}_{a,b}$. She prepares the state $\ket{\zeta_i,x,y|_{a,b}}$ and sends it to Bob and Charlie.

		\item Bob and Charlie attempt to split the state using an arbitrary channel $\Phi$, and then are no longer allowed to communicate.

		\item Alice sends $\zeta_i$ to Bob and Charlie. Bob answers with a guess $x_B$ for $x$ and Charlie answers with a guess $y_C$ for $y$.

		\item Bob and Charlie win if $|x-x_B|<\delta$ and $|y-y_C|<\varepsilon$.
	\end{enumerate}

	Due to \cref{thm:state-sending-game}, the winning probability of this game is also bounded as $\mfk{w}(\ttt{G}_{n,\delta,\varepsilon,a,b})\leq\frac{2}{n}+4\parens*{1+\frac{1}{n}}\sqrt{\delta\varepsilon}$.
\section{The coset monogamy game on $\R^n$}\label{sec:Rngame}
%%%%%%%%%%%%%%%%%%%%%%%%%%%%%%%%%%%%%%%%%%%%%%%%%%%%%%%%%%%%%

	We introduce the continuous-variable space of $n$ modes ($G=\R^n$) and its associated continuous-subgroup coset states and monogamy game. This section is in essence a generalization of the special case of $n=2$ from the previous section.

	\subsection{Optical quadrature coset states}

	With $G=\R^n$, the position states are the multimode quadratures $\ket{q}$ for $q\in\R^n$, where the inner products ${\langle x|q\rangle=\prod_{i=1}^n\delta(q_i-x_i)}$, for $x=\sum_ix_ie_i$ the expansion of $x\in\R^n$ in the canonical orthonormal basis $\{e_1,\ldots,e_n\}$.  Noting that the dual $\hat{\R}^n\cong\R^n$ with action given by the dot product $\gamma_x(y)=e^{2\pi ix\cdot y}$, the momentum states are $\ket{p}=\int_{\R^n}e^{2\pi i p\cdot q}\ket{q}dq$.

	We consider subgroups corresponding to linear subspaces in $\R^n$. Intuitively, this is the continuous case closest to the original finite case of subspaces of a finite vector space. Let $P\subseteq\R^n$ be a subspace. The quotient $\R^n/P\cong P^\perp$, the usual orthogonal subspace of the standard inner product; and the dual $\hat{P}\cong\R^n/P^\perp\cong P$ with the action inherited from the dual of $\R^n$. Due to the normalization of the dual action, the Haar measure on all of these is simply the usual Lebesgue measure. This gives, for $q\in P^\perp$ and $p\in P$, coset states
	\begin{align}
		\ket{P_{q,p}}=\ket{q+P^{\gamma_p}}=\int_Pe^{2\pi i p\cdot x}\ket{x+q}d_Px.
	\end{align}
	It is important to note that for register subspaces, \emph{i.e.}\ $P=\spn_\R\set*{e_i}{i\in I}$ for some subset $I\subseteq[n]$, the coset states expand as quadrature modes
	\begin{align}\label{eq:reg-cosets}
		\ket{P_{q,p}}=\bigotimes_{i=1}^n\begin{cases}\ket{q=q_i}&i\notin I\\\ket{p=p_i}&i\in I\end{cases}.
	\end{align}
	This naturally extends the expression of finite subspace coset states of register subspaces as conjugate-coding (BB84) states \cite{vidick2021classical}.

	Making contact with error correction, the subspace $P$ represents the code subspace of an analog stabilizer code, encoding $\mathrm{dim}(P)$ logical modes into $n$ physical modes. Register-subspace coset states $\ket{P_{0,p}}$ (\ref{eq:reg-cosets}) provide a basis of momentum states for this subspace and are eigenvalue-zero eigenstates of the position quadrature operators of the modes outside of $I$, which are called nullifiers in this context (see \cite[Appx. E]{Vuillot2018}\cite{eczoo_analog_stabilizer}).

	To deal with measurement in this basis rigorously, we also work with the coset operator measure. This is the operator measure on $\scr{B}(P^\perp)\otimes\scr{B}(P)\cong\scr{B}(\R^n)$,
	\begin{align}
		A^P(E)=\int_E\ketbra{P_{q,p}}d_{P^\perp\times P}(q,p),
	\end{align}
	or more formally,
	\begin{align}
		\braket{\phi}{A^P(E)}{\psi}=\int_E\overline{(\mc{F}_P\ket{\phi\circ q})(p)}(\mc{F}_P\ket{\psi\circ q})(p)d_{P^\perp\times P}(q,p),
	\end{align}
	where $\ket{\phi},\ket{\psi}\in L^2(\R^n)$, and $\mc{F}_P:L^2(\R^n)\rightarrow L^2(\R^n)$ is the Fourier transform with respect to $P$, defined on $\psi$ continuous with compact support as $(\mc{F}_P\ket{\psi})(p)=\int_Pe^{-2\pi ip\cdot x}\psi(x)d_Px$, and extended by continuity.

	A simple and natural damping operator on $\R^n$ is simply the $n$-fold tensor product of the damping operators on $\R$, $\Delta_{a,b}^{\otimes n}$. On each mode, this operator sends quadrature eigenstates to squeezed coherent states:
	\begin{align}
		&\ket{q}\mapsto e^{-aq^2}\int e^{-b(q-x)^2}\ket{q=x}dx=\parens*{\tfrac{\pi}{2b}}^{1/4}e^{-aq^2}\ket{b,q,0}\\
		&\ket{p}\mapsto \frac{\pi}{\sqrt{ab}}e^{-\frac{\pi^2}{a+b}p^2}\int e^{-\pi^2\frac{a+b}{ab}\parens*{y-\tfrac{1}{1+\frac{a}{b}}p}^2}\ket{p=y}dy=\parens*{\tfrac{\pi^3}{2ab(a+b)}}^{1/4}e^{-\frac{\pi^2}{a+b}p^2} \ket{\tfrac{ab}{a+b},0,\tfrac{bp}{a+b}},
	\end{align}
	where the general squeezed state
	\begin{equation}\label{eq:squeezed}
		\ket{a,x_0,p_0}=\parens*{\tfrac{2a}{\pi}}^{1/4}\int e^{-a(x-x_0)^2+2\pi i p_0x}\ket{x}dx~.
	\end{equation}
	In particular, for $a=\frac{\pi^2 b}{b^2-\pi^2}$ there is equal squeezing on both quadratures. With this operator, the damped coset state of a register subspace $P$ is
	\begin{align}\label{eq:squeezed-subspace}
		\ket{P_{q,p}|_{a,b}}=\frac{c\Delta_{a,b}^{\otimes n}\ket{P_{q,p}}}{\norm{\Delta_{a,b}^{\otimes n}\ket{P_{q,p}}}}=\bigotimes_{i=1}^n\begin{cases}\ket{b,q_i,0}&,\;i\notin I\\\ket{\tfrac{ab}{a+b},0,-\frac{bp_i}{a+b}}&,\;i\in I\end{cases}.
	\end{align}
	The distribution of damped states is
	\begin{align}\label{eq:distribution}
		\pi^P_{a,b}(q,p)=\frac{\norm{\Delta_{a,b}^{\otimes n}\ket{P_{q,p}}}^2}{\norm{\Delta_{a,b}}_2^{2n}}=\parens*{\tfrac{2a}{\pi}}^{(n-\dim P)/2}\parens*{\tfrac{2\pi}{a+b}}^{\dim P/2}e^{-2a\norm{q}^2_2-\tfrac{2\pi^2}{a+b}\norm{p}^2_2}.
	\end{align}

	\subsection{Monogamy game analysis}

	We need to use a slightly more involved analysis to study monogamy games constructed from multimode coset states, as the usual overlap lemma needs to be strengthened. This is due to the fact that the overlap of two linear subspaces may contain the neighborhood of a line which has infinite measure. However, we are able to remain in the context of the entanglement-based game throughout the analysis.

	This game proceeds as follows.
	\begin{enumerate}[1.]
		\item Bob and Charlie prepare a shared state $\rho_{ABC}$ but then are no longer allowed to communicate.

		\item Alice chooses a register subspace $P$ of dimension $n/2$ uniformly at random and measures her register in basis $\set*{\ket{P_{q,p}}}$ to get outcomes $q,p$.

		\item Alice sends a description of $P$ to Bob and Charlie. Bob answers with a guess $q_B$ for $q$ and Charlie answers with a guess $p_C$ for $p$.

		\item Bob and Charlie win if $\norm{q-q_B}_\infty<\delta$ and $\norm{p-p_C}_\infty<\varepsilon$.
	\end{enumerate}

	Thus, this corresponds to the abelian coset measure monogamy game~$\ttt{G}_{n,\delta,\varepsilon}$ with collection of subspaces $\mc{S}=\set*{\spn_\R\set{e_i}{i\in I}}{I\subseteq[n],|I|=n/2}$, and Bob and Charlie's error neighbourhoods $E=(-\delta,\delta)^n$ and $F=(-\varepsilon,\varepsilon)^n$. We are able to find a similar bound on the winning probability of this game as the bound in~\cite{CV22}.

	\begin{theorem}\label{thm:Rn-monogamy}
		The winning probability of the quadrature monogamy game satisfies
		\begin{align}
			\mfk{w}(\ttt{G}_{n,\delta,\varepsilon})\leq\frac{1}{\binom{n}{n/2}}\sum_{k=0}^{n/2}\binom{n/2}{k}^2(2\sqrt{\delta\varepsilon})^{k}\leq\sqrt{e}\parens*{\tfrac{1}{2}+\sqrt{\delta\varepsilon}}^{n/2}\;.
		\end{align}
	\end{theorem}

	Theorem~\ref{thm:Rn-monogamy} also provides a bound for the state-sending version of the game, $\ttt{G}_{n,\delta,\varepsilon,a,b}$, in which Alice sends squeezed states (\ref{eq:squeezed}).

	To show the claimed bound we use the following strengthening of the overlap bound \cref{lem:abelian-overlap}.

	\begin{lemma}\label{lem:Rn-overlap}
		Let $P=\spn_\R\set{e_i}{i\in I}$ and $Q=\spn_\R\set{e_i}{i\in J}$ be register subspaces,
		$p\in P$, and $q\in Q^\perp$. Then,
		\begin{align}
			\norm{A^P(P^\perp\times(F\cap P+p))A^Q((E\cap Q^\perp+q)\times Q)}\leq(2\sqrt{\delta\varepsilon})^{n/2-|I\cap J|}.
		\end{align}
	\end{lemma}

	\begin{proof}
		First, we note that as in \cref{lem:abelian-overlap},
		\begin{align}
		\begin{split}
			\braket{\psi}{A^Q((E\cap Q^\perp+q)\times Q)}{\psi}&=\int_{E\cap Q^\perp+q}\int_Q\abs*{(\mc{F}_P\ket{\psi\circ q'})(p')}^2 dp'dq'\\
			&=\int_{E\cap Q^\perp+q}\int_Q\abs*{\psi(p'+q')}^2 dp'dq'=\int_{q+E+Q}\abs*{\psi(q')}^2dq',
		\end{split}
		\end{align}
		that is $A^Q((E\cap Q^\perp+q)\times Q)=\Pi_{q+E+Q}$, the projector onto $q+E+Q$. Now, fixing $\ket{\psi}\in L^2(\R^n)$ continuous with compact support,
		\begin{align}
		\begin{split}
			\norm{A^P(P^\perp\times(F\cap P+p))\Pi_{q+E+Q}\ket{\psi}}^2&=\int_{P^\perp}\int_{F\cap P+p}\abs*{(\mc{F}_P(\Pi_{q+E+Q}\ket{\psi}\circ q'))(p')}^2d_Pp'd_{P^\perp}q'\\
			&=\int_{P^\perp}\int_{F\cap P+p}\abs*{\int_{P\cap (q-q'+E+Q)}e^{-2\pi i p'\cdot x}\psi(x+q')d_Px}^2d_Pp'd_{P^\perp}q'.
		\end{split}
		\end{align}
		To simplify this, we first study the set $P\cap (q-q'+E+Q)$. By definition,
		\begin{align}
			P\cap (q-q'+E+Q)=\set{x\in P}{\exists y\in Q~\text{s.t.}~\;\norm{x-y+q'-q}_\infty<\varepsilon}.
		\end{align}
		Thus, $x=\sum_{i\in I}x_ie_i\in P\cap (q-q'+E+Q)$ if and only if there exists $y=\sum_{j\in J}y_je_j$ such that $|x_i-y_i|<\varepsilon$ for all $i\in I\cap J$, $|x_i-q_i|<\varepsilon$ for all $i\in I\cap J^c$, $|y_i-q_i'|<\varepsilon$ for all $i\in I^c\cap J$, and $|q_i'-q_i|<\varepsilon$ for all $i\in I^c\cap J^c$. Since we can choose $y=\sum_{i\in I\cap J}x_ie_i+\sum_{i\in I^c\cap J}q_i'e_i$, the set simplifies to
		\begin{align}
		\begin{split}
			P\cap (q-q'+E+Q)&=\begin{cases}\set*{x\in P}{|x_i-q_i|<\varepsilon~\forall i\in I\cap J^c}&\text{ if }|q_i'-q_i|<\varepsilon\forall i\in I^c\cap J^c\\\varnothing&\text{ else}\end{cases}\\
			&=\begin{cases}q|_P+E\cap P+P\cap Q&\text{ if }q'\in q|_{P^\perp}+E\cap Q+P^\perp\cap Q\\\varnothing&\text{ else}\end{cases}.
		\end{split}
		\end{align}
		Hence, writing $Q'=q|_{P^\perp}+E\cap Q+P^\perp\cap Q$,
		\begin{align}
		\begin{split}
			&\norm{A^P(P^\perp\times(F\cap P+p))\Pi_{q+E+Q}\ket{\psi}}^2=\int_{Q'}\int_{F\cap P+p}\abs*{\int_{q|_{P}+E+P\cap Q}e^{-2\pi i p'\cdot x}\psi(x+q')d_Px}^2d_Pp'd_{P^\perp}q'\\
			&=\int_{Q'}\int_{F\cap P+p}\abs*{\int_{q|_{P\cap Q^\perp}+E\cap (P\cap Q^\perp)}\int_{P\cap Q}e^{-2\pi i p'\cdot (x+y)}\psi(x+y+q')d_{P\cap Q}xd_{P\cap Q^\perp}y}^2d_Pp'd_{P^\perp}q'\\
			&=\int_{Q'}\int_{F\cap P+p}\abs*{\int_{q|_{P\cap Q^\perp}+E\cap (P\cap Q^\perp)}e^{-2\pi i p'\cdot y}(\mc{F}_{P\cap Q}\ket{\psi\circ(y+q')})(p')d_{P\cap Q^\perp}y}^2d_Pp'd_{P^\perp}q'\\
			&\leq\int_{Q'}\int_{F\cap P+p}\mu_{P\cap Q^\perp}(q|_{P\cap Q^\perp}+E\cap (P\cap Q^\perp))\int_{P\cap Q^\perp}|(\mc{F}_{P\cap Q}\ket{\psi\circ(y+q')})(p')|^2d_{P\cap Q^\perp}yd_Pp'd_{P^\perp}q'\\
			&\leq\mu_{P\cap Q^\perp}(E\cap (P\cap Q^\perp))\int_{Q'}\int_{F\cap P\cap Q^\perp+p|_{Q^\perp}}\int_{P\cap Q}\int_{P\cap Q^\perp}|\psi(y+q'+p'')|^2d_{P\cap Q^\perp}yd_{P\cap Q}p''d_{P\cap Q^\perp}p'd_{P^\perp}q'\\
			&\leq\mu_{P\cap Q^\perp}(E\cap (P\cap Q^\perp))\mu_{P\cap Q^\perp}(F\cap (P\cap Q^\perp))\norm{\ket{\psi}}^2.
		\end{split}
		\end{align}
		Therefore, we get the bound
		\begin{align}
			\norm{A^P(P^\perp\times(F\cap P+p))A^Q((E\cap Q^\perp+q)\times Q)}\leq\sqrt{\mu_{P\cap Q^\perp}(E\cap (P\cap Q^\perp))\mu_{P\cap Q^\perp}(F\cap (P\cap Q^\perp))}.
		\end{align}
		The final bound is found by noting that
		\begin{align}
		\begin{split}
			\mu_{P\cap Q^\perp}(E\cap (P\cap Q^\perp))&=\mu_{\R^{\dim P\cap Q^\perp}}((-\delta,\delta)^{\dim P\cap Q^\perp})\\
			&=(2\delta)^{\dim P\cap Q^\perp}=(2\delta)^{n/2-|I\cap J|},
		\end{split}
		\end{align}
		and identically $\mu_{P\cap Q^\perp}(F\cap (P\cap Q^\perp))=(2\varepsilon)^{n/2-|I\cap J|}$.
	\end{proof}

	\begin{proof}[Proof of \cref{thm:Rn-monogamy}]
		The proof procedes exactly as \cref{thm:bound-abelian}. As in the finite case, there are $\binom{n}{n/2}$ register subspaces, and we can use the set of orthogonal permutations of \cite{CV22}. Using those, there are $\binom{n/2}{k}^2$ permutations such that $|I\cap\pi_i(I)|=n/2-k$, and therefore
		\begin{align}
			\mfk{w}(\ttt{G}_{n,\delta,\varepsilon})\leq\frac{1}{\binom{n}{n/2}}\sum_i(2\sqrt{\delta\varepsilon})^{n/2-|I\cap \pi_i(I)|}=\frac{1}{\binom{n}{n/2}}\sum_{k=0}^{n/2}\binom{n/2}{k}^2(2\sqrt{\delta\varepsilon})^{k}.
		\end{align}
		Finally, using the bound $\frac{1}{\binom{n}{n/2}}\sum_{k=0}^{n/2}\binom{n/2}{k}^2x^k\leq\frac{\sqrt{e}}{2^{n/2}}\parens*{1+x}^{n/2}$ gives the final result.
	\end{proof}

	Before describing the related QKD protocol, we modify the above game by first accounting for failures in some number of modes, and then by discretizing the continuous measurement values.

	\subsubsection{Mode failure}
	We can use the bound from Theorem~\ref{thm:Rn-monogamy} to work out a version of the game that accounts for measurement failures on a small number of the modes. With an additional error parameter $\gamma$, the game $\ttt{G}_{n,\delta,\varepsilon,\gamma,a,b}$ proceeds as follows.

	\begin{enumerate}[1.]
		\item Alice chooses a register subspace $P=\spn_\R\set*{e_i}{i\in\mc{I}}$ of dimension $n/2$ uniformly at random and samples $(q,p)$ according to the distribution $\pi^{P}_{a,b}$. She prepares the squeezed state $\ket{P_{q,p}|_{a,b}}$ (\ref{eq:squeezed-subspace}) and sends it to Bob and Charlie.

		\item Bob and Charlie attempt to split the state using an arbitrary channel $\Phi$, and then are no longer allowed to communicate.

		\item Alice sends $P$ to Bob and Charlie. Bob answers with a guess $q_B$ for $q$ and Charlie answers with a guess $p_C$ for $p$.

		\item Bob and Charlie win if $\abs{q_i-(q_B)_i}\geq\delta$ for at most $\gamma n/2$ values of $i\in\mc{I}^c$, and $\norm{p-p_C}_\infty<\varepsilon$.
	\end{enumerate}

	\begin{corollary}\label{cor:sq-approx-monogamy}
		For $n$ even and $\gamma$ such that $\gamma n/2$ is an integer, the winning probability of $\ttt{G}_{n,\delta,\varepsilon,\gamma}$ satisfies
		\begin{align}
			\mfk{w}(\ttt{G}_{n,\delta,\varepsilon,\gamma})\,\leq\, 2^{\squ*{(1-\gamma)\lg\parens*{\tfrac{1}{2}+\sqrt{\delta\varepsilon}}+h(\gamma)+\tfrac{1}{(\ln 2)n}}\tfrac{n}{2}}\;.
		\end{align}
	\end{corollary}

	Here $\lg$ is the base-two logarithm and  $h(\gamma)=-\gamma\lg\gamma-(1-\gamma)\lg(1-\gamma)$ is the binary entropy function.

	\begin{proof}
		The winning probability of a strategy $\ttt{S}$ can be expressed as
		\begin{align}
			\mfk{w}_{\ttt{G}_{n,\delta,\varepsilon,\gamma}}(\ttt{S})=\Pr\squ[\Bigg]{\biglor_{\substack{I\subseteq\mc{I}^c\\|I|=(1-\gamma)n/2}}(|Q_i-(Q_B)_i|<\delta\forall i\in I)\land\norm{P-P_C}_\infty<\varepsilon}.
		\end{align}
		Using a union bound,
		\begin{align}
		\begin{split}
			\mfk{w}_{\ttt{G}_{n,\delta,\varepsilon,\gamma}}(\ttt{S})&\leq\expec{|\mc{I}|=n/2}\sum_{\substack{I\subseteq\mc{I}^c\\|I|=(1-\gamma)n/2}}\Pr\squ{(|Q_i-(Q_B)_i|<\delta\forall i\in I)\land\norm{P-P_C}_\infty<\varepsilon}\\
			&\leq\sum_{\substack{I\subseteq[n/2]\\|I|=(1-\gamma)n/2}}\Pr\squ{|Q_{\mc{I}^c_i}-(Q_B)_{\mc{I}^c_i}|<\delta\land\abs{P_{\mc{I}_i}-(P_C)_{\mc{I}_i}}<\varepsilon\forall i\in I}.
		\end{split}
		\end{align}
		Now, we note that since each term in the sum depends only on the elements in $I$, it is the winning probability of a strategy for the game $\ttt{G}_{(1-\gamma)n,\delta,\varepsilon}$ played on $I$. Thus, as there are $\binom{n/2}{\gamma n/2}$ sets of cardinality $(1-\gamma)n/2$, we get that $\mfk{w}(\ttt{G}_{n,\delta,\varepsilon,\gamma})\leq\binom{n/2}{\gamma n/2}\mfk{w}(\ttt{G}_{(1-\gamma)n,\delta,\varepsilon})$. Finally, bounding $\binom{n/2}{\gamma n/2}\leq 2^{\tfrac{n}{2}h(\gamma)}$ and $\mfk{w}(\ttt{G}_{(1-\gamma)n,\delta,\varepsilon})\leq\sqrt{e}\parens*{\tfrac{1}{2}+\sqrt{\delta\varepsilon}}^{(1-\gamma)n/2}$ by \cref{thm:Rn-monogamy} gives the result.
	\end{proof}

	\subsubsection{Quadrature measurement outcome discretization}
	For applications such as QKD it is advantageous to work with discretized versions of the outputs. Though working with error neighborhoods is more natural in the context of a monogamy game, this is not particularly amenable to discretization. Rather, it is better to partition the space into disjoint bins, as is commonly done in this context. For $m\in\Z$ and $\delta>0$, we take the bin of index $m$ and width $\delta$ as the interval $B_\delta(m)=[(m-\tfrac{1}{2})\delta,(m+\tfrac{1}{2})\delta)$. Over all integer indices, the bins of width $\delta$ are a disjoint cover of $\R$. Similarly, we can take bins in $\R^n$ indexed by the integer vectors.

	We work out a binned version of the monogamy game, $\ttt{G}^{\text{binned}}_{n,\delta,\varepsilon,\gamma,a,b}$, which proceeds as follows.

	\begin{enumerate}[1.]
		\item Alice chooses a register subspace $P=\spn_\R\set*{e_i}{i\in\mc{I}}$ of dimension $n/2$ uniformly at random and samples $(q,p)$ according to the distribution $\pi^{P}_{a,b}$. She prepares the squeezed state $\ket{P_{q,p}|_{a,b}}$ (\ref{eq:squeezed-subspace}) and sends it to Bob and Charlie.

		\item Bob and Charlie attempt to split the state using an arbitrary channel $\Phi$, and then are no longer allowed to communicate.

		\item Alice sends $P$ to Bob and Charlie. Bob answers with a bin index $k\in\Z^n$ and Charlie answers with a bin index $m\in\Z^n$.

		\item Bob and Charlie win if $q_i\notin B_\delta(k_i)$ for at most $\gamma n/2$ values of $i\in\mc{I}^c$, and $p_i\in B_\varepsilon(m_i)$ for all $i\in\mc{I}$.
	\end{enumerate}

	There is a simple transformation that allows this to be reduced to the game $\ttt{G}_{n,\delta,\varepsilon,\gamma,a,b}$.

	\begin{corollary}\label{cor:sq-bin-approx-monogamy}
		The winning probability
		\begin{align}
			\mfk{w}(\ttt{G}^{\text{binned}}_{n,\delta,\varepsilon,\gamma,a,b})\leq\mfk{w}(\ttt{G}_{n,\delta,\varepsilon,\gamma,a,b})\leq2^{\squ*{(1-\gamma)\lg\parens*{\tfrac{1}{2}+\sqrt{\delta\varepsilon}}+h(\gamma)+\tfrac{1}{(\ln 2)n}}\tfrac{n}{2}}.
		\end{align}
	\end{corollary}

	\begin{proof}
		Fix a strategy $\ttt{S}$ for $\ttt{G}^{\text{binned}}_{n,\delta,\varepsilon,\gamma,a,b}$. We construct a strategy $\ttt{S}'$ for $\ttt{G}_{n,\delta,\varepsilon,\gamma,a,b}$ that wins with at least $\mfk{w}_{\ttt{G}^{\text{binned}}_{n,\delta,\varepsilon,\gamma,a,b}}(\ttt{S})$. To construct $\ttt{S}'$, we keep the channel $\Phi$, but change Bob and Charlie's measurements: Bob measures a bin index $k$ and outputs $q_B=\delta k$, and similarly Charlie measures $m$ and outputs $p_C=\varepsilon m$. If they win at the binned game, then $q_i\in[(k_i-1/2)\delta,(k_i+1/2)\delta)$ for no less than $(1-\gamma)n/2$ values of $i\in\mc{I}^c$ and $p_i\in[(m-1/2)\varepsilon,(m_i+1/2)\varepsilon)$ for all $i\in\mc{I}$. As the bin of width $\delta$ is contained in the ball of radius $\delta$, this implies that $|q_i-(q_B)_i|<\delta$ for no less than $(1-\gamma)n/2$ values of $i\in\mc{I}^c$ and $|p_i-(p_B)_i|<\varepsilon$ for all $i\in\mc{I}$. Thus, the players win at $\ttt{G}_{n,\delta,\varepsilon,\gamma,a,b}$. Using this implication, $\mfk{w}_{\ttt{G}^{\text{binned}}_{n,\delta,\varepsilon,\gamma,a,b}}(\ttt{S})\leq\mfk{w}_{\ttt{G}_{n,\delta,\varepsilon,\gamma,a,b}}(\ttt{S}')\leq\mfk{w}(\ttt{G}_{n,\delta,\varepsilon,\gamma,a,b})$, giving the desired result.
	\end{proof}

	\subsection{Application: squeezed-state one-sided device-independent QKD}\label{subsec:qkd}

	As an application of our monogamy bound for coset states over $G=\R^n$ we give a proof of security for continuous-variable quantum key distribution (CVQKD) with squeezed states in the one-sided device-independent (one-sided DI) model. What this means is that security holds even in the case where the receiver's (Bob) quantum measurement device is untrusted. To the best of our knowledge, this is the first one-sided DI proof of security for CVQKD that is secure against the most general class of attacks, coherent attacks.

The idea of one-sided DI for QKD was first introduced in~\cite{tomamichel2011uncertainty}, where the authors show one-sided DI security of the BB'84 prepare-and-measure protocol for qubits based on the use of an entropic uncertainty relation. As later pointed out in \cite{tomamichel2013monogamy}, the proof of security from~\cite{tomamichel2011uncertainty} only holds under the assumption that Bob's measurement device is memoryless, and indeed removing this assumption is one of the motivations for the monogamy-of-entanglement game studied in \cite{tomamichel2013monogamy}. The approach to one-sided DI security via uncertainty relations was later extended to squeezed-state CV protocols in~\cite{furrer2012continuous}, and has been experimentally demonstrated~\cite{gehring2015implementation}.

Our extension of the monogamy game from~\cite{tomamichel2011uncertainty} to subgroups of $\R^n$ enables us to obtain the first one-sided DI proof of security for CVQKD. While for Gaussian CVQKD protocols there has recently been a full security proof against coherent attacks~\cite{leverrier2017security}, there is currently no one-sided DI security proof, even against collective attacks, for such protocols. As we will see, our analysis leads to an error tolerance which is comparable to the one obtained for DV protocols in~\cite{tomamichel2013monogamy}. While our protocol, employing squeezed states, remains more challenging than coherent-state CV protocols, the important benefits of one-sided device independence may outweigh the experimental challenges.

	\subsubsection{Preliminaries}
	\label{sec:cvqkd-prelim}

	First, we recall, following \cite{Ren05,MR22arxiv}, the security definition of QKD.

	\begin{definition}
	  A \emph{one-sided device-independent QKD protocol} is an interaction between Alice, who is trusted, and Bob, who has an untrusted quantum device, and on which an attacker Eve may eavesdrop. The interaction produces a state $\rho_{FK\hat{K}E}$ where  $F=\Z_2$ holds a flag set to $1$ if the protocol accepts and $0$ otherwise, $K=\Z_2^\ell$ holds Alice's output, $\hat{K}=\Z_2^\ell$ holds Bob's output, and $E$ is Eve's side information. The protocol is
	  \begin{itemize}
	      \item \emph{$\varepsilon_1$-correct} if $\Pr\squ{K\neq\hat{K}\land F=1}\leq\varepsilon_1$.

	      \item \emph{$\varepsilon_2$-secret} if $\norm{\rho_{KE\land(F=1)}-\mu_K\otimes\rho_{E\land(F=1)}}_{\Tr}\leq\varepsilon_2$.

	      \item \emph{$(\Phi,\varepsilon_3)$-complete} if, when Eve acts as the channel $\Phi$ and Bob's device works as intended, then $\Pr\squ*{F=0}\leq\varepsilon_3$.
	  \end{itemize}
	\end{definition}

	Note that we generally use lowercase letters to refer to the variables on the register whose name is the corresponding uppercase letter. To achieve privacy amplification, we make use of hash functions.

	\begin{definition}[Universal\textsubscript{2} hash functions \cite{tomamichel2017largely}]
		Let $\mc{F}$ be a family of functions $X\rightarrow Y$. $\mc{F}$ is \emph{universal\textsubscript{2}} if, for all $x,x'\in X$, $\Pr\squ*{F(x)=F(x')}=\frac{1}{|Y|}$, where $F$ is the uniform random variable on $\mc{F}$.
	\end{definition}

	Such a family of functions always exists if $|X|$ and $|Y|$ are powers of $2$. We fix $\mc{F}$ a universal\textsubscript{2} hash family of functions $\Z_2^{nn_N/2}\rightarrow\Z_2^\ell$. We use the following important property.

	\begin{lemma}[Universal hash lemma \cite{Ren05}]
		Let $\rho_{FXE}=\mu_{F}\otimes\rho_{XE}$ be a (sub-normalized) quantum state, where $X=\Z_2^n$ and $\mc{F}$ is a universal\textsubscript{2} family of functions $\Z_2^n\rightarrow\Z_2^\ell$. Then,
		\begin{align}
			\norm{\rho_{F(X)FE}-\mu_{Z}\otimes\rho_{FE}}_{\Tr}\leq 2^{-\tfrac{1}{2}(H_{\min}(X|E)_\rho-(\ell-2))}.
		\end{align}
	\end{lemma}

	\begin{comment}
	Although we use continuous variables in the QKD, the finally outcome is as usual a bit string. In order to make the transition between real numbers and bit strings, we make use of Gray codes. This is a binary representation of the non-negative integers such that the Hamming distance $d(\ttt{Gray}_N(n+1),\ttt{Gray}_N(n))=1$, $\ttt{Gray}_N:\{0,\ldots,2^N-1\}\rightarrow\Z_2^N$ is the $N$-bit Gray code representation.

	\begin{lemma}
		Let $M,N\in\N$ and $x,y\in\R$ such that $|x|,|y|\leq 2^M-1$. Set $m,n\in\Z$ such that $|x-m/2^N|$ and $|y-n/2^N|$ are minimal. Then, if $|x-y|<2^{-N}$ then $d(\ttt{Gray}_{M+N+1}(m+2^{M+N}),\ttt{Gray}_{M+N+1}(n+2^{M+N}))\leq 1$.
	\end{lemma}

	\begin{proof}
		First, note that $m,$ are uniquely defined such that $|x-m/2^N|,|y-n/2^N|\leq 1/2^{N+1}$. Also, this means that $-2^{M+N}+1\leq m,n\leq 2^{M+N}-1$, giving that the Gray code is well defined on $m+2^{M+N}$ and $n+2^{M+N}$. With this, if $|x-y|<1/2^{N}$, $|m/2^N-n/2^N|<1/2^{N-1}$, or $|m-n|<2$. As they are integers, that means $|m-n|\leq 1$, so $m+2^{M+N}$ and $n+2^{M+N}$ are successive positive integers, so $d(\ttt{Gray}_{M+N+1}(m+2^{M+N}),\ttt{Gray}_{M+N+1}(n+2^{M+N}))\leq 1$.
	\end{proof}
	\end{comment}

	Now, we describe the parameters of the protocol. Fix $n\in\N$ even, $\delta,\varepsilon>0$, and $M=2^{n_M-1}$ and $N=2^{n_N-1}$ for $n_M,n_N\in\N$.
	Note that $n_M$ bits are needed to represent an integer $-M\leq m<M$ and similarly for $N$. To translate between continuous values and binned ones, for any $x\in\R$, we write $x^\varepsilon$ for the integer representing the index of the bin of width $\varepsilon$, that is $x^{\varepsilon}=\floor*{\tfrac{x}{\varepsilon}+\tfrac{1}{2}}$. And, if $-N\varepsilon\leq x_i<N\varepsilon$ for each $i\in[n]$, we write $x^{\varepsilon,N}$ for the representation as a binary string of length $n_N$. In order to preserve the metric properties of integers, we use Gray codes for the binary representation. This is a binary representation of the integers in an interval such that the Hamming distance $\varDelta((a+1)^N,a^N)=1$, which implies $\varDelta(a^N,b^N)\leq |a-b|$. For vectors, the bin indices are the corresponding integer vectors and their binary representations are the concatenation of the binary representations of the vector components. For the position and momentum measurements, only $n/2$ of the components are important in defining the coset representative, so we see the representations of their bins as bit strings as the corresponding strings of length $nn_N/2$.

	We require a classical error-correcting code family that will be used to correct the measurements of the $n/2$ position or momentum modes with $n_N$ bits per mode. Let $C\subseteq\Z_2^{nn_N/2}$ denote an infinite family of asymptotically good $[nn_N/2,nn_N/2-s,d]$ binary linear error-correcting codes with syndrome function $\mathrm{syn}:\Z_2^{nn_N/2}\rightarrow\Z_2^s$, and with $s$ such that $s=\tfrac{nn_N}{2}h(\gamma)$ asymptotically for error parameter $\gamma$, where $h$ is the binary entropy function. In other words, we need a code family that achieves the Gilbert-Varshamov (GV) bound. Explicit algebraic-geometric code constructions of binary linear codes achieving (or even beating) the GV bound are given in \cite{VNT07} (see also \cite{eczoo_ag}). However, to give an idea what is possible, it suffices for our purposes to use the GV bound.

	\begin{table}
	\begin{tabular}{c p{0.85\textwidth}}
	\toprule
	Parameter & Description\tabularnewline
	\midrule
	$P$ & $n/2$-dimensional subspace of $\R^{n}$ picked randomly by
	Alice\tabularnewline
	${\cal I},{\cal I}^{c}$ & subset of $\R^{n}$-basis vectors used to define momentum
	quadrature register subspace $P$ and its complementary set\tabularnewline
	$q,p$ & $n/2$ position and momentum values parameterizing ideal quadrature
	eigenstates picked by Alice\tabularnewline
	$a,b$ & damping parameters parameterizing normalizable quadrature eigenstates\tabularnewline
	$\delta,\varepsilon$ & Bin widths of the respective discretizations of continuous position
	and momentum quadratures\tabularnewline
	$M,N$ & Number of bins of the respective discretizations\tabularnewline
	\midrule
	$I$ & subset of ${\cal I}^{c}$ position quadratures picked for parameter
	estimation\tabularnewline
	$\theta$ & fraction of quadratures defined by ${\cal I}^{c}$ that are picked
	for parameter estimation; $|I|=\theta n/2$\tabularnewline
	$q_{i}^{\delta},p_{i}^{\varepsilon}$ & Integer bin index of the $i$th position and momentum quadrature with
	$i\in{\cal I}^{c},{\cal I}$, respectively\tabularnewline
	$q_{i}^{\delta,M},p_{i}^{\varepsilon,N}$ & Binary representation of the integer bin index of the $i$th position
	and momentum quadrature\tabularnewline
	$n_{M},n_{N}$ & Length of the binary representation of the integer position and momentum
	bin index\tabularnewline
	$\hat{q},\hat{p}$ & Bob's measured values of position and momentum quadratures\tabularnewline
	$x,y$ & additive Gaussian white-noise channel (\ref{eq:displacement}) parameters\tabularnewline
	\midrule
	$s=n-k$ & Length of the syndrome of a classical binary linear $[n,k,d]$ code
	used for privacy amplification.\tabularnewline
	$\bar{p}$ & Bob's error-corrected values of the momentum quadratures, corresponding to his raw key\tabularnewline
	\midrule
	$\eta$ & Fraction of bits of position-quadrature binary strings used for information
	reconciliation\tabularnewline
	$J$ & Subset of size $\eta n_{N}n/2$ determining $n_{N}n/2$ bit locations
	of position-quadrature binary strings used for information reconciliation\tabularnewline
	\midrule
	$f$ & function sampled from a family ${\cal F}$ of universal$_{2}$
	hash functions\tabularnewline
	$k,\hat{k}$ & Alice's and Bob's final keys\tabularnewline
	$\ell$ & Length of Alice's and Bob's keys\tabularnewline
	\bottomrule
	\end{tabular}

	\caption{Table of parameters used in the CVQKD \cref{prot:sq-cont-qkd}. Horizontal lines divide
	rows associated with the five steps of the protocol.}
	\end{table}

		%%%%%%%%%%%%%%%%%%%%%%%%%%%%%%%%%%%%%%%%%
		\subsubsection{Protocol and completeness}
		%%%%%%%%%%%%%%%%%%%%%%%%%%%%%%%%%%%%%%%%%
		We use the following protocol, based on standard tools: the preparation of squeezed states by Alice and the use of homodyne detection for Bob.

		\begin{mdframed}
		\begin{protocol}[Squeezed-state device-independent continuous-variable QKD]\label{prot:sq-cont-qkd}\hphantom{}
		\begin{description}
			\item[State preparation] Alice samples a register subspace $P=\spn_\R\set*{e_i}{i\in\mc{I}}\subseteq\R^n$ of dimension $n/2$ uniformly at random. She samples $(q,p)$ according to $\pi^{P}_{a,b}$ (\ref{eq:distribution}) until $-M\delta\leq q_i< M\delta$ and $-N\varepsilon\leq p_i< N\varepsilon$ for all $i$, discarding and resampling as many times as necessary. Then, she prepares $\ket{P_{q,p}|_{a,b}}$ (\ref{eq:squeezed-subspace}) and sends it to Bob.

			\item[Parameter estimation] Bob acknowledges receipt of the state, then Alice sends $P$ and $q^{\delta,M}_I$, where $I\subseteq\mc{I}^c$ is a subset of size $\theta n/2$. Bob measures the modes of his received states with homodyne detection to get guesses $(\hat{q},\hat{p})$, measuring in position for $i\notin\mc{I}$ and momentum for $i\in\mc{I}$.
			If any of the terms are outside the expected ranges, he replaces them with $0$. If $\hat{q}^\delta_i\neq q^\delta_i$ for more than $\gamma\theta n/2$ values of $i\in I$, Bob aborts.

			\item[Error correction] Alice sends $\mathrm{syn}(p^{\varepsilon,N})$. Bob corrects $\parens*{-(1+a/b)\hat{p}}^{\varepsilon,N}$ using this to get $\bar{p}^{\varepsilon,N}$.

			\item[Information reconciliation] Alice chooses a random subset $J\subseteq[n_Nn/2]$ of size $\eta n_Nn/2$ and sends $J,p^{\varepsilon,N}_J$ to Bob; if $\bar{p}^{\varepsilon,N}_J\neq p^{\varepsilon,N}_J$, Bob aborts.

			\item[Privacy amplification] Alice chooses $f\in\mc{F}$ uniformly random and sends it to Bob. Alice computes the key $k=f(p^{\varepsilon,N})$ and Bob computes $\hat{k}=f(\bar{p}^{\varepsilon,N})$.
		\end{description}
		\end{protocol}
		\end{mdframed}

		\begin{proposition}
			\cref{prot:sq-cont-qkd} is $\parens*{1-\frac{2d}{n_Nn}}^{\eta n_Nn/2}$-correct.
		\end{proposition}

		As long as $1/\eta\in o(n)$, this provides a protocol with strong correctness for large enough $n$.

		\begin{proof}
				First, we have that
				\begin{align*}
				\begin{split}
					\Pr\squ{K\neq\hat{K}\land F=1}&=\Pr\squ*{F(P^{\varepsilon,N})\neq F(\bar{P}^{\varepsilon,N})\land\abs{\set{i\in I}{Q_i^\delta\neq\hat{Q}_i^\delta}}\leq\gamma\theta n/2\land\bar{P}^{\varepsilon,N}_J=P^{\varepsilon,N}_J}\\
					&\leq\Pr\squ{P^{\varepsilon,N}\neq\bar{P}^{\varepsilon,N}\land\bar{P}^{\varepsilon,N}_J=P^{\varepsilon,N}_J}
				\end{split}
				\end{align*}
				Note that $P$ here is the register/random variable corresponding to the momentum value, not subspace. By the error-correcting code, if $(-(1+a/b)\hat{p})^{\varepsilon,N}$ is corrected to something different from $p^{\varepsilon,N}$, we must have that the Hamming distance $\varDelta(\bar{p}^{\varepsilon,N},p^{\varepsilon,N}
				)\geq d$. Hence,
				\begin{align}
				\begin{split}
					\Pr\squ{K\neq\hat{K}\land F=1}&\leq\Pr\squ{\varDelta(P^{\varepsilon,N},\bar{P}^{\varepsilon,N})\geq d\land\bar{P}^{\varepsilon,N}_J=P^{\varepsilon,N}_J}\\
					&\leq\frac{\binom{n_Nn/2-d}{\eta n_Nn/2}}{\binom{n_Nn/2}{\eta n_Nn/2}}\leq\parens*{1-\frac{2d}{n_Nn}}^{\eta n_Nn/2}.
				\end{split}
				\end{align}
			\end{proof}

			We first verify completeness in the protocol when there is no noise on the quantum communication channel.

		\begin{lemma}\label{lem:inte-great}
			Let $\alpha>0$ and let $\pi:\R\rightarrow[0,\infty)$ be a probability density function that is decreasing in the sense that if $|x|\geq|y|$, then $\pi(x)\leq\pi(y)$. Then,
			\begin{align}
				\iint\abs*{\floor*{x+\tfrac{1}{2}}-\floor*{y+\tfrac{1}{2}}}e^{-\alpha(x-y)^2}\pi(x)dxdy\leq\frac{6}{\alpha}
			\end{align}
			%and
			%\begin{align}
			%	\iint_{\floor*{x+\tfrac{1}{2}}\neq\floor*{y+\tfrac{1}{2}}}e^{-\alpha(x-y)^2}\pi(x)dxdy\leq\frac{4}{\alpha}e^{-\frac{\alpha}{4}}.
			%\end{align}
		\end{lemma}

		\begin{proof}
			Naively, we can try to approximate $\abs*{\floor*{x+\tfrac{1}{2}}-\floor*{y+\tfrac{1}{2}}}$ by $|x-y|$. However, this is not an upper bound. Even if we were to take a scalar multiple of $|x-y|$ as the candidate upper bound, there are values of $x$ and $y$ where the bound does not hold. To remedy that, we use the property that $\pi(x)$ is decreasing. First, note that $\abs*{\floor*{x+\tfrac{1}{2}}-\floor*{y+\tfrac{1}{2}}}\leq 3|x-y|$ except for on an infinite sequence of regions contained in $\abs*{|x-n-1/2|+|y-n-1/2|}\leq 1/3$ for $n\in\Z$. To cover those regions as well, we upper bound them by the integrals on $\abs*{|x-n-1/3|+|y-n-1/3|}\leq 1/3$ and $\abs*{|x-n+1/3|+|y-n+1/3|}\leq 1/3$, which is a bound as $\pi(x)$ decreases with $|x|$. As such,
			\begin{align}
				\iint\abs*{\floor*{x+\tfrac{1}{2}}-\floor*{y+\tfrac{1}{2}}}e^{-\alpha(x-y)^2}\pi(x)dxdy\leq\iint6|x-y|e^{-\alpha(x-y)^2}\pi(x)dxdy.
			\end{align}
			To finish,
			\begin{align}
				\iint6|x-y|e^{-\alpha(x-y)^2}\pi(x)dxdy=12\int_{-\infty}^\infty\int_0^\infty ue^{-\alpha u^2}du\pi(x)dx=\frac{6}{\alpha}.
			\end{align}
		\end{proof}

		\begin{theorem}\label{thm:sq-completeness} Suppose that the error parameter $\gamma>\frac{6}{\sqrt{2\pi b}\delta}$, and the error-correcting code distance is ${d>\frac{n}{2}\frac{3\sqrt{a(1+a/b)}}{\pi^{3/2}\varepsilon}}$. Then, \cref{prot:sq-cont-qkd} is $(\mathrm{id},p)$-complete,
			where
			\begin{align}
				p=\exp\parens*{-\parens*{\gamma-\frac{6}{\sqrt{2\pi b}\delta}}^2}^{\theta n}+\exp\parens*{-\parens*{\tfrac{2d}{nn_N}-\tfrac{3\sqrt{a(1+a/b)}}{n_N\pi^{3/2}\varepsilon}}}^n\;.
			\end{align}
		\end{theorem}

		If the conditions in the theorem are satisfied, then we get exponentially good completeness in the case of no errors.

		\begin{proof}
			The probability of aborting is given by
			\begin{align}
			\begin{split}
				\Pr\squ*{F=0}&=\Pr\squ*{\abs{\set{i\in I}{Q_i^\delta\neq\hat{Q}_i^\delta}}>\gamma\theta n/2\lor P_J^{\varepsilon,N}\neq \bar{P}_J^{\varepsilon,N}}\\
				&\leq\Pr\squ*{\abs{\set{i\in I}{Q_i^\delta\neq\hat{Q}_i^\delta}}>\gamma\theta n/2}+\Pr\squ*{P^{\varepsilon,N}\neq \bar{P}^{\varepsilon,N}}
			\end{split}
			\end{align}
			For each $i$, let $\Gamma_i$ be the random variable that is $0$ if $Q_i^\delta=\hat{Q}_i^\delta$ and $1$ otherwise. Writing $\pi(q)=e^{-aq^2}/\int_{-M\delta}^{M\delta}e^{-at^2}dt$ for the distribution of position quadratures, the expectation value
			\begin{align}
			\begin{split}
				\mbb{E}\Gamma_i&\leq\mbb{E}|Q_i^\delta-\hat{Q}_i^\delta|=\mbb{E}\abs*{\floor*{\tfrac{Q_i}{\delta}+\tfrac{1}{2}}-\floor*{\tfrac{\hat{Q}_i}{\delta}+\tfrac{1}{2}}}\\
				&=\iint\abs*{\floor*{\tfrac{q}{\delta}+\tfrac{1}{2}}-\floor*{\tfrac{\hat{q}}{\delta}+\tfrac{1}{2}}}\abs*{\braket{\hat{q}}{b,q,0}}^2\pi(q)dqd\hat{q}\\
				&=\sqrt{\tfrac{2b}{\pi}}\iint\abs*{\floor*{\tfrac{q}{\delta}+\tfrac{1}{2}}-\floor*{\tfrac{\hat{q}}{\delta}+\tfrac{1}{2}}}e^{-2b(q-\hat{q})^2}\pi(q)dqd\hat{q}\\
				&=\sqrt{\tfrac{2b}{\pi}}\delta\iint\abs*{\floor*{x+\tfrac{1}{2}}-\floor*{y+\tfrac{1}{2}}}e^{-2b\delta^2(x-y)^2}\delta\pi(\delta x)dxdy,
			\end{split}
			\end{align}
			using the change of variables $x=q/\delta$, $y=\hat{q}/\delta$. %\enote{This approximation is the problem. If we can improve it, we can lower the necessary squeezing.}
			So, we can apply \cref{lem:inte-great} and get $\mbb{E}\Gamma_i\leq\frac{6}{\sqrt{2\pi b}\delta}$. As the random variables are independent, we can invoke Hoeffding's inequality in the case that $\gamma>\frac{6}{\sqrt{2\pi b}\delta}$ and get
			\begin{align}
			\begin{split}
				\Pr\squ*{\abs{\set{i\in I}{Q_i^\delta\neq\hat{Q}_i^\delta}}>\gamma\theta \tfrac{n}{2}}&=\Pr\squ[\Big]{\sum_i\Gamma_i-\theta n/2\mbb{E}\Gamma_i>(\gamma-\mbb{E}\Gamma_i)\theta\tfrac{n}{2}}\leq e^{-\tfrac{2\parens*{(\gamma-\mbb{E}\Gamma_i)\theta n/2}^2}{\theta n/2}}\\
				&=e^{-(\gamma-\mbb{E}\Gamma_i)^2\theta n}\leq\exp\parens*{-\parens*{\gamma-\frac{6}{\sqrt{2\pi b}\delta}}^2}^{\theta n}.
			\end{split}
			\end{align}

			Now, we bound the other term. Due to the error-correcting code, the event $P^{\varepsilon,N}\neq \bar{P}^{\varepsilon,N}$ implies $\varDelta(P^{\varepsilon,N},(-(1+\tfrac{a}{b})\hat{P})^{\varepsilon,N})\geq d$. Let the random variable $\Delta_i=\varDelta(P_i^{\varepsilon,N},(-(1+\tfrac{a}{b})\hat{P}_i)^{\varepsilon,N})$. This is a random variable supported on the interval $[0,n_N]$. Then, if $\tfrac{2d}{n}>\mbb{E}\Delta_i$, Hoeffding's inequality gives that $\Pr\squ*{\sum_i\Delta_i\geq d-\tfrac{n}{2}\mbb{E}\Delta_i}\leq e^{-\tfrac{n}{n_N^2}(2d/n-\mbb{E}\Delta_i)^2}$. It remains to compute $\mbb{E}\Delta_i$. We have, writing $\pi(p)=e^{-\tfrac{2\pi^2}{a+b}p^2}/\int_{-N\varepsilon}^{N\varepsilon}e^{-\tfrac{2\pi^2}{a+b}t^2}dt$ for the distribution of the momentum quadratures,
			\begin{align}
			\begin{split}
				\mbb{E}\Delta_i&=\mbb{E}\varDelta((P_i)^{\varepsilon,N},(-(1+\tfrac{a}{b})\hat{P}_i)^{\varepsilon,N})\\
				&\leq\mbb{E}\abs*{(P_i)^{\varepsilon}-(-(1+\tfrac{a}{b})\hat{P}_i)^{\varepsilon}}=\mbb{E}\abs*{\floor*{\tfrac{P_i}{\varepsilon}+\tfrac{1}{2}}-\floor*{\tfrac{-(1+a/b)\hat{P}_i}{\varepsilon}+\tfrac{1}{2}}}\\
				&=\iint \abs*{\floor*{\tfrac{p}{\varepsilon}+\tfrac{1}{2}}-\floor*{\tfrac{-(1+a/b)\hat{p}}{\varepsilon}+\tfrac{1}{2}}}\abs*{\braket{\hat{p}}{\tfrac{ab}{a+b},0,-\tfrac{bp}{a+b}}}^2\pi(p)dpd\hat{p}\\
				&=\sqrt{2\pi\parens*{\tfrac{1}{a}+\tfrac{1}{b}}}\iint \abs*{\floor*{\tfrac{p}{\varepsilon}+\tfrac{1}{2}}-\floor*{\tfrac{-(1+a/b)\hat{p}}{\varepsilon}+\tfrac{1}{2}}}e^{-2\pi^2\parens*{\tfrac{1}{a}+\tfrac{1}{b}}\parens*{\hat{p}+\tfrac{1}{1+a/b}p}^2}\pi(p)dpd\hat{p}\\
				&=\sqrt{\tfrac{2\pi}{a(1+a/n)}}\varepsilon\iint\abs*{\floor*{x+\tfrac{1}{2}}-\floor*{y+\tfrac{1}{2}}}\varepsilon\pi(\varepsilon x)e^{-\tfrac{2\pi^2\varepsilon^2}{a(1+a/b)}(x-y)^2}dxdy,
			\end{split}
			\end{align}
			using a change of variables $x=p/\varepsilon$, $y=-\parens*{1+a/b}\hat{p}/\varepsilon$. Then, as $x\mapsto\varepsilon\pi(\varepsilon x)$ is a probability density function that satisfies the hypothesis of \cref{lem:inte-great}, we can apply that and get that $\mbb{E}\Delta_i\leq\frac{3\sqrt{a(1+a/b)}}{\pi^{3/2}\varepsilon}$.
		\end{proof}

		Next, we study completeness of the protocol with respect to a simple noise model: iid classical Gaussian noise (\textit{a.k.a.}\ additive Gaussian white noise \cite{Banaszek:20,PhysRevLett.125.080503}) on the modes. On one mode, this channel acts as
		\begin{align}\label{eq:displacement}
			\Phi_{x,y}(\rho)=\frac{1}{\pi xy}\iint D(\xi,\phi)\rho D(\xi,\phi)^\dag e^{-(\xi/x)^2-(\phi/y)^2}d\xi d\phi,
		\end{align}
		where the displacement operator $D(\xi,\phi)\ket{q}=e^{\pi i\phi\xi}e^{2\pi i\phi q}\ket{q+\xi}$. On squeezed states (\ref{eq:squeezed}), the displacement operator acts simply as $D(\xi,\phi)\ket{a,q_0,p_0}=e^{-\pi i(\phi+2p_0)}\ket{a,q_0+\xi,p_0+\phi}$ on the position modes. We consider $n$-mode channels of the form $\Phi_{x,y}^{\otimes n}$.

		\begin{corollary}
			Suppose that $\gamma>\frac{6\sqrt{1+2bx^2}}{\sqrt{2\pi b}\delta}$ and $d>\frac{n}{2}\frac{3\sqrt{a(1+a/b)}}{\pi^{3/2}\varepsilon}\sqrt{1+2\pi^2\parens*{1/a+1/b}y^2}$. Then, \cref{prot:sq-cont-qkd} is $(\Phi_{x,y}^{\otimes n},p')$-complete, where $p'$ is
			\begin{align}
				\exp\parens*{-\parens*{\gamma-\tfrac{6\sqrt{1+2bx^2}}{\sqrt{2\pi b}\delta}}^2}^{\theta n}\hspace{-0.3cm}+\exp\parens*{-\parens*{\tfrac{2d}{nn_N}-\tfrac{3\sqrt{a(1+a/b)}}{n_N\pi^{3/2}\varepsilon}\sqrt{1+2\pi^2\parens*{1/a+1/b}y^2}}^2}^n
			\end{align}
		\end{corollary}

		\begin{proof}
			Following the same method as \cref{thm:sq-completeness}, we can bound the probability of aborting as $\Pr\squ*{F=1}\leq e^{-(\gamma-\mbb{E}\Gamma_i)^2\theta n}+e^{-(2d/(n_Nn)-\mbb{E}\Delta_i/n_N)^2n}$, where the random variables $\Gamma_i$ and $\Delta_i$ are defined as above. It remains to bound those. First,
			\begin{align}
			\begin{split}
				\mbb{E}\Gamma_i&\leq\frac{1}{\pi xy}\iint\iint\abs*{\floor*{\tfrac{q}{\delta}+\tfrac{1}{2}}-\floor*{\tfrac{\hat{q}}{\delta}+\tfrac{1}{2}}}\abs*{\braket{\hat{q}}{D(\xi,\phi)}{b,q,0}}^2\pi(q)dqd\hat{q}e^{-(\xi/x)^2-(\phi/y)^2}d\xi d\phi\\
				&=\frac{1}{\pi xy}\iint\iint\abs*{\floor*{\tfrac{q}{\delta}+\tfrac{1}{2}}-\floor*{\tfrac{\hat{q}}{\delta}+\tfrac{1}{2}}}\abs*{\braket{\hat{q}}{b,q+\xi,\phi}}^2\pi(q)dqd\hat{q}e^{-(\xi/x)^2-(\phi/y)^2}d\xi d\phi\\
				&=\frac{1}{\pi xy}\iint\iint\abs*{\floor*{\tfrac{q}{\delta}+\tfrac{1}{2}}-\floor*{\tfrac{\hat{q}}{\delta}+\tfrac{1}{2}}}\sqrt{\frac{2b}{\pi}}e^{-2b(\hat{q}-q-\xi)^2}\pi(q)dqd\hat{q}e^{-(\xi/x)^2-(\phi/y)^2}d\xi d\phi\\
				&=\frac{1}{\pi xy}\sqrt{\frac{2b}{\pi}}\delta\iint\iint\abs*{\floor*{w+\tfrac{1}{2}}-\floor*{z+\tfrac{1}{2}}}e^{-2b\delta^2(z-w-\xi/\delta)^2}\delta\pi(\delta w)dwdze^{-(\xi/x)^2-(\phi/y)^2}d\xi d\phi\\
				&\leq\frac{6}{\pi xy}\sqrt{\frac{2b}{\pi}}\delta\iint\iint|w-z|e^{-2b\delta^2(z-w-\xi/\delta)^2}\delta\pi(\delta w)dwdze^{-(\xi/x)^2-(\phi/y)^2}d\xi d\phi\\
				&=\frac{6\sqrt{2b}}{\pi x}\delta\int\int|u|e^{-2b\delta^2(u-\xi/\delta)^2}due^{-(\xi/x)^2}d\xi=\frac{6\sqrt{2b}}{\pi x}\delta\sqrt{\frac{\pi}{2b+1/x^2}}\int|u|e^{-2b\delta^2\squ*{1-\tfrac{2b}{2b+1/x^2}}u^2}du\\
				&=\frac{6\sqrt{1+2bx^2}}{\sqrt{2\pi b}\delta}
			\end{split}
			\end{align}
			Proceeding similarly for $\Delta_i$, we get
			\begin{align}
			\begin{split}
				\mbb{E}\Delta_i&\leq\iint\iint \abs*{\floor*{\tfrac{p}{\varepsilon}+\tfrac{1}{2}}-\floor*{\tfrac{-(1+a/b)\hat{p}}{\varepsilon}+\tfrac{1}{2}}}\abs*{\braket{\hat{p}}{D(\xi,\phi)}{\tfrac{ab}{a+b},0,-\tfrac{bp}{a+b}}}^2\pi(p)dpd\hat{p}\tfrac{1}{\pi xy}e^{-(\xi/x)^2-(\phi/y)^2}d\xi d\phi\\
				&=\tfrac{1}{ xy}\sqrt{\tfrac{a+b}{\pi ab}}\iint\iint \abs*{\floor*{\tfrac{p}{\varepsilon}+\tfrac{1}{2}}-\floor*{\tfrac{-(1+a/b)\hat{p}}{\varepsilon}+\tfrac{1}{2}}}e^{-\tfrac{2\pi^2(a+b)}{ab}\parens*{\hat{p}+\tfrac{1}{1+a/b}p-\phi}^2-(\xi/x)^2-(\phi/y)^2}\pi(p)dpd\hat{p}d\xi d\phi\\
				&=\tfrac{1}{y}\sqrt{\tfrac{1}{a(1+a/b)}}\varepsilon\int\iint \abs*{\floor*{w+\tfrac{1}{2}}-\floor*{z+\tfrac{1}{2}}}e^{-\tfrac{2\pi^2(a+b)}{ab}\parens*{-\frac{\varepsilon }{1+a/b}z+\tfrac{\varepsilon}{1+a/b}w-\phi}^2-(\phi/y)^2}\varepsilon\pi(\varepsilon w)dwdzd\phi\\
				&\leq\tfrac{6}{y}\sqrt{\tfrac{1}{a(1+a/b)}}\varepsilon\int\iint \abs*{w-z}e^{-\tfrac{2\pi^2}{a(1+a/b)}\varepsilon^2\parens*{w-z-\frac{1+a/b}{\varepsilon}\phi}^2-(\phi/y)^2}\varepsilon\pi(\varepsilon w)dwdzd\phi\\
				&=\tfrac{6}{y}\sqrt{\tfrac{1}{a(1+a/b)}}\varepsilon\int\int \abs*{u}e^{-\tfrac{2\pi^2}{a(1+a/b)}\varepsilon^2\parens*{u-\frac{1+a/b}{\varepsilon}\phi}^2-(\phi/y)^2}dud\phi\\
				&=\frac{3\sqrt{a(1+a/b)}}{\pi^{3/2}\varepsilon}\sqrt{1+2\pi^2\parens*{1/a+1/b}y^2}
			\end{split}
			\end{align}
		\end{proof}

		\subsubsection{Secrecy}

		We now show secrecy of the protocol.

		\begin{theorem}\label{thm:qkd-security}
			Let $\tau>0$ and let $\delta,\varepsilon,M,N>0$ be defined as in Section~\ref{sec:cvqkd-prelim}. Then, \cref{prot:sq-cont-qkd} is $\varepsilon'$-secret, where
			\begin{align}
			\begin{split}
				\varepsilon'&=2^{\tfrac{n}{4}\squ*{(1-\gamma-\tau)\lg\parens*{\tfrac{1}{2}+\sqrt{\delta\varepsilon}}+h(\gamma+\tau)+\theta n_M+\tfrac{2 s}{n}+\eta n_N-\lg\parens*{1-\tfrac{e^{-2aM^2\delta^2}}{\sqrt{2\pi aM^2\delta^2}}}-\lg\parens*{1-\sqrt{\tfrac{a+b}{\pi^3N^2\varepsilon^2}}e^{-\frac{2\pi^2}{a+b}N^2\varepsilon^2}}+\tfrac{2(\ell-2)}{n}+\tfrac{1}{(\ln 2)n}}}\\
				&\quad+4e^{-\tau^2\theta n}
			\end{split}
			\end{align}
		\end{theorem}

		The key rate is the number of secure key bits returned by the protocol per quantum signal transmitted, \emph{i.e.}\ it is the ratio $r=\frac{\ell}{n}$.
		To get the asymptotic performance of the protocol, we seek to express the key rate as a function of the error tolerance for reasonable choices of parameters, in the limit $n\rightarrow\infty$. We need that the secrecy parameter $\varepsilon'\rightarrow 0$.
		To do so, we can choose $\tau$, $\theta$, and $\eta$ so that $\tau$, $e^{-\tau^2\theta n}$, $\theta n_M$, $\eta n_N$, $\tfrac{-4}{n}$, and $\tfrac{1}{(\ln 2)n}$ arbitrarily small, \emph{e.g.} by taking and $\theta=\tau=\eta=n^{-1/4}$.
		The error-correcting code asymptotically achieves the Gilbert-Varshamov bound, so $s=\tfrac{nn_N}{2}h(\gamma)$ \cite{tomamichel2013monogamy}.
		Thus, the condition for key generation reduces to
		\begin{align}
			-(1-\gamma)\lg\parens*{\tfrac{1}{2}+\sqrt{\delta\varepsilon}}-h(\gamma)-\tfrac{2s}{n}+\lg\parens*{1-\tfrac{e^{-2aM^2\delta^2}}{\sqrt{2\pi aM^2\delta^2}}}+\lg\parens*{1-\sqrt{\tfrac{a+b}{\pi^3N^2\varepsilon^2}}e^{-\frac{2\pi^2}{a+b}N^2\varepsilon^2}}>2r.
		\end{align}
		Now, we can introduce values of the constants. Generally, the damping coefficients $a,b$ are functions of a squeezing parameter $\Delta$: $a=\tfrac{\Delta^2}{2}$, $b=\tfrac{1}{2\Delta^2}$. We choose $\Delta=0.001$. Note that only squeezing with parameter $\Delta=\sqrt{0.1}$ has been achieved \cite{MWV22,FRMDH22} experimentally, so this is choice is for the moment unrealistic. We choose the remaining parameters so that completeness is also achieved: $n_N=n_M=16$, $\delta=4$, and $\varepsilon=1/64$. Under these conditions, we may plot the asymptotic error tolerance as a function of the rate based on the relation
		\begin{align}
			(0.4150)-(0.4150)\gamma-17h(\gamma)\,>\,2r\;.
		\end{align}
		This gives an asymptotic error tolerance $\gamma\approx0.24\%$ and an optimal rate $r\approx0.21$, or if completeness is satisfied with respect to the identity channel then $r\approx0.07$. Note that by increasing the Gaussian error, completeness remains satisfied up to error parameters $x\approx0.0027$ and $y\approx0.0002$. The plot of this relation is given in \cref{fig:qkds-favourite}.

		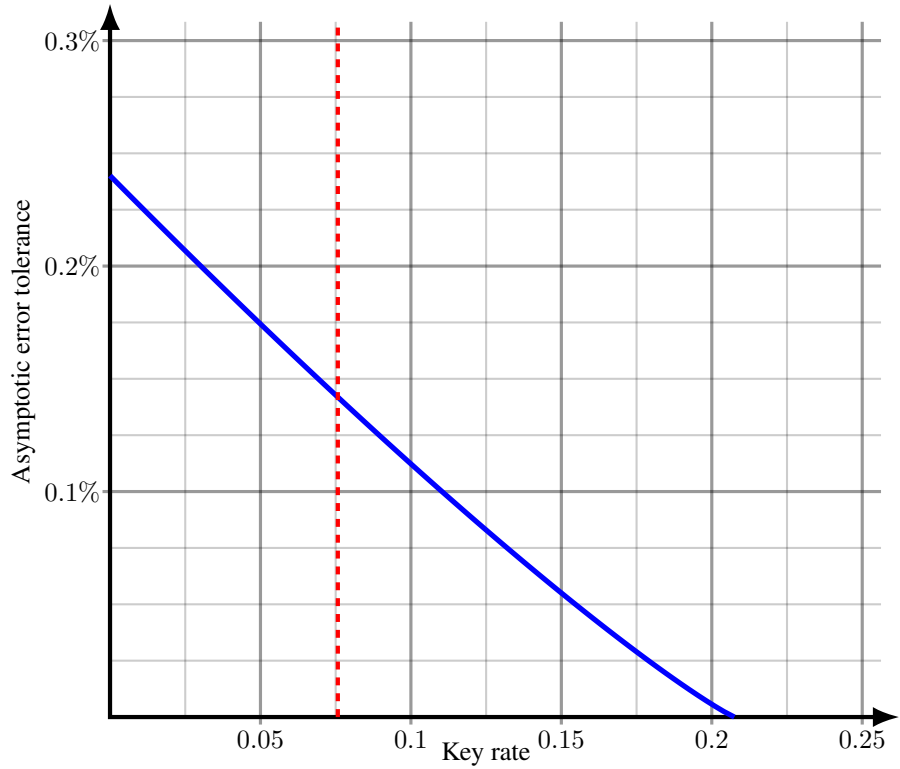
\begin{figure}
			\centering
			\begin{tikzpicture}[scale=0.5]

				%\draw[white, opacity=0] (-1.5, -1.5) rectangle (21.5,19.5);

				\draw[very thick, opacity=0.4] (-0.2,6) -- (20.5,6);
				\node at (-1,6){\small$0.1\%$};
				\draw[very thick, opacity=0.4] (-0.2,12) -- (20.5,12);
				\node at (-1,12){\small$0.2\%$};
				\draw[very thick, opacity=0.4] (-0.2,18) -- (20.5,18);
				\node at (-1,18){\small$0.3\%$};

				\draw[very thick, opacity=0.4] (4,-0.2) -- (4,18.5);
				\node at (4,-0.6){\small$0.05$};
				\draw[very thick, opacity=0.4] (8,-0.2) -- (8,18.5);
				\node at (8,-0.6){\small$0.1$};
				\draw[very thick, opacity=0.4] (12,-0.2) -- (12,18.5);
				\node at (12,-0.6){\small$0.15$};
				\draw[very thick, opacity=0.4] (16,-0.2) -- (16,18.5);
				\node at (16,-0.6){\small$0.2$};
				\draw[very thick, opacity=0.4] (20,-0.2) -- (20,18.5);
				\node at (20,-0.6){\small$0.25$};

				\draw[thick, opacity=0.2] (0,1.5) -- (20.5,1.5);
				\draw[thick, opacity=0.2] (0,3) -- (20.5,3);
				\draw[thick, opacity=0.2] (0,4.5) -- (20.5,4.5);

				\draw[thick, opacity=0.2] (0,7.5) -- (20.5,7.5);
				\draw[thick, opacity=0.2] (0,9) -- (20.5,9);
				\draw[thick, opacity=0.2] (0,10.5) -- (20.5,10.5);

				\draw[thick, opacity=0.2] (0,13.5) -- (20.5,13.5);
				\draw[thick, opacity=0.2] (0,15) -- (20.5,15);
				\draw[thick, opacity=0.2] (0,16.5) -- (20.5,16.5);

				\draw[thick, opacity=0.2] (2,0) -- (2,18.5);
				\draw[thick, opacity=0.2] (6,0) -- (6,18.5);
				\draw[thick, opacity=0.2] (10,0) -- (10,18.5);
				\draw[thick, opacity=0.2] (14,0) -- (14,18.5);
				\draw[thick, opacity=0.2] (18,0) -- (18,18.5);

				\draw[ultra thick, Latex-Latex] (21,0) -- (0,0) -- (0,19);

				\node at (10,-1) {{\small Key rate}};
				\node[rotate=90] at (-2.3,10) {{\small Asymptotic error tolerance}};

				\draw[line width=0.07cm,blue] (-0.0003471976782644295,14.406600000000001) --(0.14363747812098498,14.261078787878791) --(0.28786552749876454,14.115557575757578) --(0.4323394535718924,13.970036363636364) --(0.577061811608448,13.824515151515154) --(0.7220352106748631,13.678993939393942) --(0.8672623153531439,13.533472727272729) --(1.012745847531383,13.387951515151519) --(1.1584885882732578,13.242430303030305) --(1.304493379768569,13.096909090909092) --(1.450763127371057,12.951387878787878) --(1.5973008017278145,12.805866666666669) --(1.7441094410049442,12.660345454545455) --(1.8911921532166724,12.514824242424243) --(2.038552118661655,12.369303030303033) --(2.1861925924745726,12.22378181818182) --(2.3341169072991788,12.078260606060606) --(2.482328476089661,11.932739393939396) --(2.630830795049801,11.787218181818183) --(2.7796274467161086,11.641696969696971) --(2.9287221031958004,11.49617575757576) --(3.0781185295686324,11.350654545454548) --(3.227820587462743,11.205133333333336) --(3.377832238817354,11.059612121212123) --(3.528157549842403,10.91409090909091) --(3.678800695189994,10.7685696969697) --(3.829765962351319,10.623048484848486) --(3.981057756294009,10.477527272727274) --(4.1326806043582796,10.332006060606062) --(4.284639161427818,10.18648484848485) --(4.436938215396769,10.040963636363639) --(4.589582692953717,9.895442424242425) --(4.74257766570536,9.749921212121214) --(4.895928356667039,9.604400000000002) --(5.049640147145489,9.458878787878788) --(5.203718584045851,9.313357575757577) --(5.358169387635453,9.167836363636365) --(5.512998459800041,9.022315151515153) --(5.668211892833944,8.876793939393941) --(5.823815978805622,8.73127272727273) --(5.979817219548269,8.585751515151516) --(6.136222337327939,8.440230303030305) --(6.293038286247008,8.294709090909091) --(6.450272264449181,8.14918787878788) --(6.607931727195404,8.003666666666668) --(6.766024400891571,7.858145454545455) --(6.924558298155573,7.712624242424243) --(7.083541734021304,7.5671030303030316) --(7.242983343390571,7.421581818181819) --(7.402892099853102,7.276060606060607) --(7.563277336013202,7.1305393939393955) --(7.724148765476154,6.985018181818182) --(7.885516506667585,6.839496969696971) --(8.047391108680184,6.6939757575757595) --(8.209783579369937,6.548454545454546) --(8.372705415949637,6.402933333333334) --(8.536168638365368,6.257412121212122) --(8.700185825779101,6.11189090909091) --(8.864770156526824,5.966369696969698) --(9.029935451978025,5.820848484848486) --(9.195696224782136,5.675327272727274) --(9.362067732066508,5.529806060606061) --(9.529066034237221,5.38428484848485) --(9.696708060139903,5.238763636363637) --(9.86501167946614,5.093242424242425) --(10.033995783439133,4.947721212121213) --(10.203680374998807,4.802200000000001) --(10.3740866699264,4.656678787878788) --(10.545237210618726,4.511157575757577) --(10.717155994556148,4.365636363636365) --(10.889868619913363,4.220115151515152) --(11.063402451273703,4.07459393939394) --(11.23778680904046,3.9290727272727275) --(11.413053186938576,3.7835515151515158) --(11.58923550301775,3.6380303030303036) --(11.766370390864624,3.492509090909091) --(11.944497539416115,3.3469878787878797) --(12.123660091956763,3.201466666666667) --(12.303905117775631,3.055945454545455) --(12.485284173819753,2.910424242424243) --(12.667853978891209,2.7649030303030306) --(12.851677230079703,2.6193818181818185) --(13.036823601047992,2.4738606060606063) --(13.223370975818918,2.328339393939394) --(13.411406991911988,2.1828181818181824) --(13.601030996341878,2.03729696969697) --(13.792356562580343,1.8917757575757579) --(13.98551478532658,1.7462545454545455) --(14.180658679074774,1.6007333333333336) --(14.377969185669564,1.4552121212121214) --(14.577663601965982,1.3096909090909092) --(14.780007785547166,1.164169696969697) --(14.985334529107604,1.018648484848485) --(15.194072579527987,0.8731272727272728) --(15.406795358298822,0.7276060606060607) --(15.624309659550244,0.5820848484848485) --(15.847836480608503,0.4365636363636364) --(16.079448867484935,0.29104242424242427) --(16.3235117411567,0.14552121212121213) --(16.60056,0.0);

				\draw[red, dashed, ultra thick] (6.057,0) -- (6.057,18.5);

			\end{tikzpicture}
			\caption{Asymptotic error tolerance as a function of key rate for our choice of parameters. The blue curve is the error tolerance and red line is the optimal rate where the completeness is preserved. This curve does not correspond to an optimal outcome, and is rather intended to illustrate the fact that there is a non-trivial choice of parameters.}
			\label{fig:qkds-favourite}
		\end{figure}

		\begin{proof}[Proof of \cref{thm:qkd-security}]
			During the protocol, Alice, Bob and Eve construct the following state. First, Alice prepares
			\begin{align}
			\begin{split}
				&\sigma_{SQ^{\delta,M}P^{\varepsilon,N}\R^n}=\frac{1}{\binom{n}{n/2}}\sum_{P}\int\ketbra{P}\otimes\ketbra{q^{\delta,M}}\otimes\ketbra{p^{\varepsilon,N}}\otimes\ketbra{P_{q,p}|_{a,b}}\pi^{P}_{a,b}(q,p)d(q,p).
			\end{split}
			\end{align}
			Restricted to the case $\norm{q}_\infty\leq M\delta$ and $\norm{p}_\infty\leq N\varepsilon$ (which is what Alice's choice of truncation does up to a measure $0$ set), the state becomes
			\begin{align}
			\begin{split}
				&\rho_{SQ^{\delta,M}P^{\varepsilon,N}\R^n}=\sigma_{SQ^{\delta,M}P^{\varepsilon,N}\R^n|(\norm{Q}_\infty\leq M\delta\land \norm{P}_\infty\leq N\varepsilon)}\\
				&\quad=\frac{C}{\binom{n}{n/2}}\sum_{P}\int_{\norm{q}_\infty\leq M\delta,\norm{p}_\infty\leq N\varepsilon}\ketbra{P}\otimes\ketbra{q^{\delta,M}}\otimes\ketbra{p^{\varepsilon,N}}\otimes\ketbra{P_{q,p}|_{a,b}}\pi^{P}_{a,b}(q,p)d(q,p),
			\end{split}
			\end{align}
			where $C=\parens*{\frac{1}{\binom{n}{n/2}}\sum_P\int_{\norm{q}_\infty\leq M\delta,\norm{p}_\infty\leq N\varepsilon}\pi^P_{a,b}(q,p)d(q,p)}^{-1}$ is the normalization coefficient. Then, sending the coset state to Bob, Eve acts on it by some splitting operation $\Phi$ to get
			\begin{align}
			\begin{split}
				&\rho_{SQ^{\delta,M}P^{\varepsilon,N}BE}\\
				&\quad=\frac{C}{\binom{n}{n/2}}\sum_{P}\int_{\norm{q}_\infty\leq M\delta,\norm{p}_\infty\leq N\varepsilon}\ketbra{P}\otimes\ketbra{q^{\delta,M}}\otimes\ketbra{p^{\varepsilon,N}}\otimes\Phi(\ketbra{P_{q,p}|_{a,b}})\pi^{P}_{a,b}(q,p)d(q,p),
			\end{split}
			\end{align}
			Alice shares $P$ publicly. Bob then makes guesses $\hat{q}^\delta$ and $\hat{p}^\varepsilon$, extending the state to $\rho_{SQ^{\delta,M}P^{\varepsilon,N}\hat{Q}^\delta\hat{P}^\varepsilon E}$. Suppose that, at this point, Eve attempts to make a guess $p_C^\varepsilon$ for $p^\varepsilon$. Let $\tau>0$ and write
			\begin{align}
				p_0=2^{\squ*{(1-\gamma-\tau)\lg\parens*{\tfrac{1}{2}+\sqrt{\delta\varepsilon}}+h(\gamma+\tau)+\tfrac{1}{(\ln 2)n}}\tfrac{n}{2}}.
			\end{align}
			 We know from the binned monogamy relation \cref{cor:sq-bin-approx-monogamy} that
			\begin{align}
				\Pr\squ*{\abs{\set{i\in[n]}{Q_i^\delta\neq\hat{Q}_i^\delta}}\leq(\gamma+\tau)n/2\land P^\varepsilon=P_C^\varepsilon}_{\sigma}\leq p_0.
			\end{align}
			Due to the bounds we have imposed, this gives that
			\begin{align}
				\Pr\squ*{\abs{\set{i\in[n]}{Q_i^\delta=\hat{Q}_i^\delta}}\leq(\gamma+\tau)n/2\land P^{\varepsilon,N}=P_C^{\varepsilon,N}}_{\rho}\leq Cp_0.
			\end{align}
			Fix the subset $I\subseteq[n]$ that Bob checks and let $\Omega$ be the event \begin{align}
				\tfrac{2}{n}\abs{\set{i\in[n]}{Q_i^\delta\neq\hat{Q}_i^\delta}}\leq\tfrac{2}{\theta n}\abs{\set{i\in I}{Q_i^\delta\neq\hat{Q}_i^\delta}}+\tau.
			\end{align}
			Using Hoeffding's inequality, we can see that $\Pr\squ*{\lnot\Omega}_\rho\leq e^{-\tau^2\theta n}$, giving
			\begin{align}
				\norm*{\rho-\rho_{|\Omega}}_{\Tr}=\norm*{\parens*{1-\tfrac{1}{\Pr[\Omega]_\rho}}\rho_{\land\Omega}+\rho_{\land\lnot\Omega}}=\parens*{\tfrac{1}{\Pr[\Omega]_\rho}-1}\Pr[\Omega]_{\rho}+\Pr[\lnot\Omega]_\rho=2\Pr[\lnot\Omega]_\rho\leq 2e^{\tau^2\theta n}.
			\end{align}
			On $\rho_{\land\Omega}$, we get
			\begin{align}
				\Pr\squ*{\abs{\set{i\in I}{Q_i^\delta\neq\hat{Q}_i^\delta}}\leq\gamma\theta n/2\land P^{\varepsilon,N}=P_C^{\varepsilon,N}}_{\rho_{\land\Omega}}\leq Cp_0,
			\end{align}
			Now, let $\Omega_0$ be the event $\abs{\set{i\in I}{Q_i^\delta\neq\hat{Q}_i^\delta}}\leq\gamma\theta n/2$. Fixing the measurement Bob makes to guess $q$, write $p_1=\Pr\squ*{\Omega_0}_{\rho_{\land\Omega}}$ so $\frac{Cp_0}{p_1}\geq\Pr\squ*{P^{\varepsilon,N}=P_C^{\varepsilon,N}\Big|\Omega_0}_{\rho_{\land\Omega}}$. Therefore, we have the min-entropy relation
			\begin{align}
				H_{\min}(P^{\varepsilon,N}|SE)_{\rho_{|\Omega_0\land\Omega}}\geq-\lg\tfrac{p_0}{p_1}=-\lg Cp_0+\lg p_1.
			\end{align}
			Until the privacy amplification, Eve also receives $(q_I)^{\delta,M}$ ($\theta n_Mn/2$ bits), $\mathrm{syn}(p^{\varepsilon,N})$ ($s$ bits), $p^{\varepsilon,N}_J$ ($\eta n_Nn/2$ bits). Giving Eve's register $E'=SI(Q_I)^{\delta,M}\mathrm{syn}(P^{\varepsilon,N})JP^{\varepsilon,N}_JE$ and hence by chain rule,
			\begin{align}
			\begin{split}
				H_{\min}(P^{\varepsilon,N}|E')_{\rho_{|\Omega_0\land\Omega}}&\geq -\lg Cp_0+\lg p_1-\squ*{\theta n_M+\tfrac{2s}{n}+\eta n_N}\tfrac{n}{2}.
			\end{split}
			\end{align}
			If $P^{\varepsilon,N}_J\neq \bar{P}^{\varepsilon,N}_J$, Bob aborts, but as $\rho_{|\Omega_0\land(\Omega\land P^{\varepsilon,N}_J=\bar{P}_J^{\varepsilon,N})}\leq\rho_{|\Omega_0\land\Omega}$, this does not change the entropy relation. Let $\Omega_1$ be the event $P^{\varepsilon,N}_J=\bar{P}_J^{\varepsilon,N}$.
			Now, using the universal hash lemma,
			\begin{align}
				\norm{\rho_{e(Q^{\varepsilon,N},R)RE'|\Omega_0\land(\Omega_1\land\Omega)}-\mu_{K}\otimes\mu_R\otimes\rho_{E'|\Omega_0\land(\Omega_1\land\Omega)}}_{\Tr}\leq 2^{-\tfrac{1}{2}\parens*{-\lg Cp_0+\lg p_1-\squ*{\theta n_M+\tfrac{2s}{n}+\eta n_N}\tfrac{n}{2}-(\ell-2)}}.
			\end{align}
			Combining this with the other case and noting that Eve's final register is $E''=RE'$ gives
			\begin{align}
			\begin{split}
				\norm{\rho_{KE''\land(F=1\land\Omega)}-\mu_R\otimes\rho_{E''\land(F=1\land\Omega)}}_{\Tr}&\leq 2^{-\tfrac{1}{2}\parens*{-\lg Cp_0-\squ*{\theta n_M+\tfrac{2s}{n}+\eta n_N}\tfrac{n}{2}-(\ell-2)}}\\
				&=\sqrt{C}2^{\tfrac{n}{4}\squ*{(1-\gamma-\tau)\lg\parens*{\tfrac{1}{2}+\sqrt{\delta\varepsilon}}+h(\gamma+\tau)+\theta n_M+\tfrac{2 s}{n}+\eta n_N+\tfrac{2(\ell-2)}{n}+\tfrac{1}{(\ln 2)n}}}
			\end{split}
			\end{align}
			Now, we can remove the event $\Omega$ using the trace norm bound: $\norm{\rho_{KE''\land(F=1)}-\mu_R\otimes\rho_{E''\land(F=1)}}_{\Tr}\leq \sqrt{C}2^{\tfrac{n}{4}\squ*{(1-\gamma-\tau)\lg\parens*{\tfrac{1}{2}+\sqrt{\delta\varepsilon}}+h(\gamma+\tau)+\theta M+\tfrac{2 s}{n}+\eta N+\tfrac{2(\ell-2)}{n}+\tfrac{1}{(\ln 2)n}}}+4e^{-\tau^2\theta n}$. Finally, it remains to bound $C$. First, we have that
			\begin{align}
			\begin{split}
				\frac{1}{C}&=\Pr[\norm{Q}_\infty\leq M\delta\land\norm{P}_\infty\leq N\varepsilon]=\prod_{i\in\mc{I}^c}\Pr[|Q_i|\leq M\delta]\prod_{i\in\mc{I}}\Pr[|P_i|\leq N\varepsilon]\\
				&=\parens*{1-2\sqrt{\tfrac{2a}{\pi}}\int_{M\delta}^\infty e^{-2aq^2}dq}^{n/2}\parens*{1-2\sqrt{\tfrac{a+b}{2\pi}}\int_{N\varepsilon}^\infty e^{-\frac{2\pi^2}{a+b}p^2}dp}^{n/2}\\
				&\geq \parens*{1-2\sqrt{\tfrac{2a}{\pi}}\int_{M\delta}^\infty \frac{q}{M\delta}e^{-2aq^2}dq}^{n/2}\parens*{1-2\sqrt{\tfrac{a+b}{2\pi}}\int_{N\varepsilon}^\infty \frac{p}{N\varepsilon}e^{-\frac{2\pi^2}{a+b}p^2}dp}^{n/2}\\
				&=\parens*{1-\tfrac{e^{-2aM^2\delta^2}}{\sqrt{2\pi aM^2\delta^2}}}^{n/2}\parens*{1-\sqrt{\tfrac{a+b}{\pi^3N^2\varepsilon^2}}e^{-\frac{2\pi^2}{a+b}N^2\varepsilon^2}}^{n/2}.
			\end{split}
			\end{align}
			Thus, we may bound
			\begin{align}
				\sqrt{C}=2^{\tfrac{1}{2}\lg C}\leq2^{-\tfrac{n}{4}\squ*{\lg\parens*{1-\tfrac{e^{-2aM^2\delta^2}}{\sqrt{2\pi aM^2\delta^2}}}+\lg\parens*{1-\sqrt{\tfrac{a+b}{\pi^3N^2\varepsilon^2}}e^{-\frac{2\pi^2}{a+b}N^2\varepsilon^2}}}},
			\end{align}
			giving the result.
		\end{proof}

	\section{The coset monogamy game on $\R$}\label{sec:Rgame}

	As in the previous two sections, we study oscillators, concentrating this time on a single mode, that corresponds to the real group $\R$. The only proper non-trivial closed subgroups in this case are discrete and infinite, corresponding to the integers $\Z$. This yields a connection to the GKP code \cite{GKP01}, whose coset states correspond to lattices in the CV phase space.

	\subsection{GKP states and coset states}

	Quantum states of the oscillator are normalized elements of the square-integrable functions of the real line $L^2(\R)$. Nevertheless, as before, it often useful to consider ``bases'' of this space composed of non-normalizable states. The two canonical choices are the position states $q(x)=\delta(q-x)$ with normalization $\braket{q}{q'}=\delta(q-q')$ and the momentum states $p(x)=e^{2\pi i px}$ with $\braket{p}{p'}=\delta(p-p')$. The position states correspond to points in $\R$ and the momentum states correspond to elements of the dual $\hat{\R}\cong\R$. The isomorphism takes the natural form $x\rightarrow\gamma_x$ for $x\in\R$, where $\gamma_x(y)=e^{2\pi i xy}$.

	An important family of states on the oscillator are the GKP code states \cite{GKP01}. They take the form
	\begin{align}\label{eq:gkp-states}
		\ket{\alpha,d,k}\propto\sum_{n=-\infty}^\infty\ket{q=(k+dn)\alpha}=\sum_{n=-\infty}^\infty e^{2\pi i\frac{k}{d\alpha}n}\ket{p=\tfrac{n}{d\alpha}},
	\end{align}
	for some $\alpha\in\R$, $d\in\N$ and $k=0,\ldots,d-1$. Again, this is non-normalizable: in fact these states are infinite both in position and in momentum, as they are an equal superposition of infinitely many states in either basis. The standard way to handle this is to turn this into a normalizable state in $L^2(\R)$ by damping the Dirac deltas $\ket{q}$, replacing them with Gaussians $q_{a,b}(x)=e^{-aq^2}e^{-b(q-x)^2}$. As in the case of $G=\R^n$, let $\Delta_{a,b}$ be the operator that effects this transformation.

	In the same way as subspace states are generalised to subspace coset states \cite{coladangelo2021hidden}, GKP states can be generalised to subgroup coset states. The relevant subgroup is $\alpha\Z\leq\R$, as the GKP state is a countable superposition of position states. Then, the cosets $\R/\alpha\Z\cong U(1)$ can be indexed by $[0,\alpha)$. Also, the dual group of $\hat{\alpha\Z}\cong\hat{\R}/(\alpha\Z)^\perp\cong U(1)$, as $(\alpha\Z)^\perp\cong\tfrac{1}{\alpha}\Z$ under the isomorphism $\hat{\R}\cong\R$, so the characters are indexed by $[0,1/\alpha)$. Thus, for $x\in[0,\alpha)$ and $y\in[0,1/\alpha)$, the subspace coset states take the form \begin{align}
		\ket{\alpha,x,y}=\ket{x+\alpha\Z^{\gamma_y}}=\sum_{n\in\Z}e^{2\pi i y\alpha n}\ket{q=x+\alpha n}.
	\end{align}
	These satisfy orthogonality as well, in the form $\braket{\alpha,x,y}{\alpha,x',y'}=\delta(x-x')\delta(y-y')$, and can be transformed into normalizable states via the same damping operation. However, it is useful to again consider this basis as an operator measure, which will allow us to rigorously interact with these unnormalizable objects. As the basis is indexed by $[0,\alpha)\times[0,1/\alpha)$, the measurable sets are the Borel sets of this space $\scr{B}([0,\alpha)\times[0,1/\alpha))$, and the operator measure of a set $E\in\scr{B}([0,\alpha)\times[0,1/\alpha))$ is the projector onto the "span" of the states $\ket{\alpha,x,y}$ such that $(x,y)\in E$:
	\begin{align}
		A^\alpha(E)=\int_E\ketbra{\alpha,x,y}d(x,y).
	\end{align}
	This is in fact a well-defined bounded operator. We go through the rigorous definition of this operator measure and the proof of this for a general abelian group in \cref{sec:abelian-infinite}.

	Alternately, we can again work with damped states $\ket{\alpha,x,y|_{a,b}}:=\frac{c\Delta_{a,b}\ket{\alpha,x,y}}{\norm{\Delta_{a,b}\ket{\alpha,x,y}}}$, distributed according to $\pi^{\alpha}_{a,b}(x,y)=\frac{\norm{\Delta_{a,b}\ket{\alpha,x,y}}^2}{\norm{\Delta_{a,b}}_2^2}$.

	\subsection{Monogamy game analysis}

	The definition and analysis of the GKP state monogamy game is a bit more involved, as the general overlap bound we work out is not directly useful in analysing these states. First, we demonstrate where the difficulty arises, and then adapt the game and bound to make it work.

	Naively, the basic GKP state monogamy game should take the following form. Fix $\alpha_1,...,\alpha_N>0$ real numbers. The GKP state monogamy game, played between a referee Alice and two cooperating players Bob and Charlie, proceeds as follows.

	\begin{enumerate}[1.]
		\item Bob and Charlie prepare a shared state $\rho_{ABC}$ but then are no longer allowed to communicate.

		\item Alice chooses $i=1,\ldots,N$ uniformly at random and measures her register in basis $\set*{\ket{\alpha_i,x,y}}$ to get measurements $x,y$.

		\item Alice sends $i$ to Bob and Charlie. Bob answers with a guess $x_B$ for $x$ and Charlie answers with a guess $y_C$ for $y$.

		\item Bob and Charlie win if $|x-x_B|<\delta$ and $|y-y_C|<\varepsilon$ (modulo $\alpha_i$ and $1/\alpha_i$, respectively).
	\end{enumerate}

	Then, we might make use of the bound of \cref{thm:bound-abelian} to upper bound the game. Crucially, this bound relies on the overlap $c(\alpha,\beta):=\sup_{x\in\R}\sqrt{\mu_{\alpha\Z}(\alpha\Z\cap(x+(-\delta,\delta)+\beta\Z))\mu_{\hat{\alpha\Z}}((-\varepsilon,\varepsilon))}$. However, for any $\delta>0$ the set $\alpha\Z\cap((-\delta,\delta)+\beta\Z)$ is infinite. The cardinality, and hence the measure, of this set is equal to the number of $n\in\Z$ such that $\abs*{\frac{\alpha}{\beta}-\frac{m}{n}}<\frac{\delta}{\beta n}$: for rational $\frac{\alpha}{\beta}$, there are infinitely many $n$ such that for some $m$, $\frac{\alpha}{\beta}=\frac{m}{n}$; and for irrational $\frac{\alpha}{\beta}$, Dirichlet's approximation theorem gives that there are infinitely many fractions $\frac{m}{n}$ such that $\abs*{\frac{\alpha}{\beta}-\frac{m}{n}}<\frac{1}{n^2}$, and thus since there must be infinitely many $n\geq\frac{\beta}{\delta}$ such that this is satisfied, $\abs*{\frac{\alpha}{\beta}-\frac{m}{n}}<\frac{1}{n^2}\leq\frac{\delta}{\beta n}$ for these $n$. Hence, for any reasonable instantiation of the game, the bound is trivial.

	Nevertheless, it is possible to get a nontrivial bound by working more directly with the state-sending version of the game. Most importantly, the damped GKP states only have significant support in a finite interval. Then, by translating the strategy for a state-sending version of the game to the original version using \cref{thm:state-sending-game}, we need only work with states that have significant support but in a particular finite interval. By projecting onto this interval, the overlap becomes finite and manageable, with a small perturbation of the winning probability. First, we present the state-sending game, then formalise the preceding argument.

	\begin{enumerate}[1.]
		\item Alice chooses $i$ uniformly at random and samples $(x,y)$ according to the distribution $\pi^{\alpha_i}_{a,b}$. She prepares the state $\ket{\alpha_i,x,y|_{a,b}}$ (\ref{eq:gkp-states}) and sends it to Bob and Charlie.

		\item Bob and Charlie attempt to split the state using an arbitrary channel $\Phi$, and then are no longer allowed to communicate.

		\item Alice sends $\alpha_i$ to Bob and Charlie. Bob answers with a guess $x_B$ for $x$ and Charlie answers with a guess $y_C$ for $y$.

		\item Bob and Charlie win if $|x-x_B|<\delta$ and $|y-y_C|<\varepsilon$.
	\end{enumerate}

	\begin{theorem}\label{thm:GKP-monogamy}
		Let $\alpha_1,\ldots,\alpha_N$ be an ascending sequence of prime numbers. For any $M>0$, the winning probability of the GKP state-sending monogamy game
		\begin{align}
			\mfk{w}(\ttt{G}_{GKP})\leq\frac{1}{N}+2\sqrt{\parens*{\alpha_N+\frac{2M}{\alpha_1}}\varepsilon}+\sqrt{\frac{1}{M}\sqrt{\frac{2}{\pi a}}}e^{-aM^2}
		\end{align}
	\end{theorem}

	This gives a good bound on the winning probability if $N\gg 1$, $\alpha_i\sim\sqrt{M}$, $\varepsilon^2\ll 1/M\ll a^{1/2}$. First, we prove a variant of the overlap lemma involving the projection onto the interval $[-M,M]$.

	\begin{lemma}\label{lem:GKP-overlap}
		Let $\alpha,\beta\in\R$, and $\varepsilon,\delta,M>0$. Let $E=(-\delta,\delta)$, $F=(-\varepsilon,\varepsilon)$ (which we might consider modulo $\alpha$, $\beta$, $1/\alpha$, or $1/\beta$ implicitly depending on the context), and $\Pi\in\mc{P}(L^2(\R))$ be the projector $(\Pi\ket{\psi})(x)=\psi(x)$ if $|x|\leq M$ and $0$ otherwise. Then, for any $r\in[0,1/\alpha)$, $s\in[0,\beta)$,
		\begin{align}
			\norm{\Pi A^{\alpha}([0,\alpha)\times(r+F))\Pi A^{\beta}((s+E)\times[0,1/\beta))\Pi}\leq\sup_{x\in\R}\sqrt{2\alpha\varepsilon|(x+\alpha\Z)\cap(s+E+\beta\Z)\cap[-M,M]|}.
		\end{align}
		Also, if $\alpha,\beta$ are integers and $2\varepsilon\leq1$, the bound simplifies to
		\begin{align}
			\norm{\Pi A^{\alpha}([0,\alpha)\times(r+F))\Pi A^{\beta}((s+E)\times[0,1/\beta))\Pi}\leq\sqrt{4\alpha\parens*{1+\frac{2M}{\mathrm{lcm}(\alpha,\beta)}}\varepsilon}.
		\end{align}
	\end{lemma}

	\begin{proof}
		We follow the method of \cref{lem:abelian-overlap}, and additionally keep track of the projector $\Pi$. First, we have that $A^{\beta}((s+E)\times[0,1/\beta))=\Pi_{s+E+\beta\Z}$, so
		\begin{align}
		\begin{split}
			&\norm{\Pi A^{\alpha}([0,\alpha)\times(r+F))\Pi A^{\beta}((s+E)\times[0,1/\beta))\Pi}\\
			&\qquad=\norm{\Pi A^{\alpha}([0,\alpha)\times(r+F))\Pi_{(s+E+\beta\Z)\cap[-M,M]}}\\
			&\qquad\leq\norm{A^{\alpha}([0,\alpha)\times(r+F))\Pi_{(s+E+\beta\Z)\cap[-M,M]}}
		\end{split}
		\end{align}
		Thus, for any $\ket{\psi}\in L^2(\R)$,
		\begin{align}
		\begin{split}
			&\norm{A^{\alpha}([0,\alpha)\times(r+F))\Pi_{(s+E+\beta\Z)\cap[-M,M]}\ket{\psi}}^2\\
			&\qquad=\int_{[0,\alpha)\times(r+F)}\abs*{\int_{\alpha\Z\cap(s-x+E+\beta\Z)\cap[-M-x,M-x]}\psi(x+h)\overline{\gamma_y(h)}d_{\alpha\Z}h}^2d_{\R/\alpha\Z\times\hat{\alpha\Z}}(x,y)\\
			&\qquad=\int_0^\alpha\alpha\int_{-\varepsilon}^{\varepsilon}\abs*{\sum_{n\in\Z}\chi_{(s-x+E+\beta\Z)\cap[-M-x,M-x]}(\alpha n)\psi(x+\alpha n))e^{-2\pi i\alpha n(r+y)}}^2dydx\\
			&\qquad\leq\alpha\int_0^\alpha\int_{-\varepsilon}^{\varepsilon}\sum_{n\in\Z}|\psi(x+\alpha n))|^2\sup_x|\alpha\Z\cap(s-x+E+\beta\Z)\cap[-M-x,M-x]|dydx\\
			&\qquad=\norm{\ket{\psi}}^22\alpha\varepsilon\sup_x|(x+\alpha\Z)\cap(s+E+\beta\Z)\cap[-M,M]|.
		\end{split}
		\end{align}
		Hence, we get the bound $\sup_x\sqrt{2\alpha\varepsilon|(x+\alpha\Z)\cap(s+E+\beta\Z)\cap[-M,M]|}$ on the overlap.

		It remains to bound $|(x+\alpha\Z)\cap(s+E+\beta\Z)\cap[-M,M]|$ in the case where $\alpha,\beta$ are integers. First, note that
		\begin{align}
		\begin{split}
			|(x+\alpha\Z)\cap(s+E+\beta\Z)\cap[-M,M]|&=\abs*{\set{n\in\Z}{|x+\alpha n|\leq M;\;\exists m\in\Z.\;|(x+\alpha n)-(s+\beta m)|<\delta}}\\
			&\leq 1+\abs*{\set{n\in\Z}{|n|\leq \tfrac{2M}{\alpha};\;\exists m\in\Z.\;|\alpha n-\beta m|<2\delta}}\\
			&\leq 2+2\abs*{\set{n\in\N}{n\leq \tfrac{2M}{\alpha};\;\exists m\in\Z.\;|\alpha n-\beta m|<2\delta}},
		\end{split}
		\end{align}
		where the inequality following from noting that the difference of any two elements of the left hand set is an element of the right hand set. Now, using the hypothesis, $|\alpha n-\beta m|<2\delta$ if and only if $\alpha n-\beta m=0$, that is if $\frac{m}{n}=\frac{\alpha}{\beta}$. Thus, the elements of the set take the form $n=\frac{k\beta}{\mathrm{gcd}(\alpha,\beta)}$, so
		\begin{align}
		\begin{split}
			|(x+\alpha\Z)\cap(s+E+\beta\Z)\cap[-M,M]|&\leq 2+2\abs*{\set{k\in\N}{k\leq \tfrac{2M\mathrm{gcd}(\alpha,\beta)}{\alpha\beta}}}\leq2+\frac{4M}{\mathrm{lcm}(\alpha,\beta)}.
		\end{split}
		\end{align}
	\end{proof}

	Now, we can proceed to the proof of the theorem.

	\begin{proof}[Proof of \cref{thm:GKP-monogamy}]
		Fixing a strategy for the state-sending game, the construction of \cref{thm:state-sending-game} gives a reexpression as a strategy for the original monogamy game:
		\begin{align}
			\mfk{w}_{\ttt{G}_{GKP}}(\ttt{S})=\expec{i}\Tr\squ[\big]{(A^{\alpha_i}\otimes B^{\alpha_i\Z}\otimes C^{\alpha_i\Z})(E_{\alpha_i\Z})(\Id\otimes\Phi)(\ketbra{\Psi_{a,b}})},
		\end{align}
		where $\ket{\Psi_{a,b}}$ is the maximally-entangled state corresponding to the damping operator $\Delta_{a,b}$. Now, using the projector $\Pi=\Pi_{[-M,M]}$ and writing $P_i=(\mathrm{id}\otimes\Phi^\dag)\parens*{(A^{\alpha_i}\otimes B^{\alpha_i\Z}\otimes C^{\alpha_i\Z})(E_{\alpha_i\Z})}$,
		\begin{align}
		\begin{split}
			\mfk{w}_{\ttt{G}_{GKP}}&(\ttt{S})=\expec{i}\parens*{\braket{\Psi_{a,b}}{\Pi P_i\Pi}{\Psi_{a,b}}+\braket{\Psi_{a,b}}{(\Id-\Pi)P_i\Pi}{\Psi_{a,b}}+\braket{\Psi_{a,b}}{P_i(\Id-\Pi)}{\Psi_{a,b}}}\\
			&\leq\expec{i}\parens*{\braket{\Psi_{a,b}}{\Pi P_i\Pi}{\Psi_{a,b}}+\norm*{(\Id-\Pi)\ket{\Psi_{a,b}}}\norm*{P_i\Pi\ket{\Psi_{a,b}}}+\norm*{P_i\ket{\Psi_{a,b}}}\norm*{(\Id-\Pi)\ket{\Psi_{a,b}}}}\\
			&\leq\expec{i}\Tr\squ[\big]{(\Pi A^{\alpha_i}\Pi\otimes B^{\alpha_i\Z}\otimes C^{\alpha_i\Z})(E_{\alpha_i\Z})(\Id\otimes\Phi)(\ketbra{\Psi_{a,b}})}+2\norm*{(\Id-\Pi)\ket{\Psi_{a,b}}}.
		\end{split}
		\end{align}
		First, we bound the second term. Using \cref{sec:damping}, the singular-value decomposition $\Delta_{a,b}=\sum_is_i\ketbra{\phi_i}{\chi_i}$ gives $\ket{\Psi_{a,b}}=\frac{1}{\norm{\Delta_{a,b}}_2}\sum_is_{n,i}\ket{\phi_i}\otimes c\ket{\chi_i}$. Hence,
		\begin{align}
			\norm{(\Id-\Pi)\ket{\Psi_{a,b}}}^2=\braket{\Psi_{a,b}}{\Id-\Pi}{\Psi_{a,b}}=\frac{1}{\norm{\Delta_{a,b}}_2^2}\sum_is_i^2\braket{\phi_i}{\Id-\Pi}{\phi_i}=\frac{\Tr\squ{\Delta_{a,b}^\dag(\Id-\Pi)\Delta_{a,b}}}{\Tr\squ{\Delta_{a,b}^\dag\Delta_{a,b}}}.
		\end{align}
		So, we have $\Tr\squ{\Delta_{a,b}^\dag\Delta_{a,b}}=\int_{-\infty}^\infty\norm{\Delta_{a,b}\ket{x}}^2dx=\int_{-\infty}^\infty\sqrt{\frac{\pi}{2b}}e^{-2\bar{a}x^2}dx=\frac{\pi}{2\sqrt{\bar{a}b}},$ and similarly
		\begin{align}
		\begin{split}
			\Tr\squ{\Delta_{a,b}^\dag(\Id-\Pi)\Delta_{a,b}}&=\int_{-\infty}^\infty\norm{(\Id-\Pi)\Delta_{a,b}\ket{x}}^2dx=2\int_{-\infty}^\infty\int_{M}^\infty e^{-2\tilde{a}x^2-2b(x-y)^2}dydx\\
			&=2\sqrt{\frac{\pi}{2(\tilde{a}+b)}}\int_M^\infty e^{-2\frac{\tilde{a}b}{\tilde{a}+b}y^2}dy\leq 2\sqrt{\frac{\pi}{2(\tilde{a}+b)}}\int_M^\infty \frac{y}{M}e^{-2\frac{\tilde{a}b}{\tilde{a}+b}y^2}dy\\
			&=\frac{\sqrt{\pi(\tilde{a}+b)/2}}{4\tilde{a}bM}e^{-2\frac{\tilde{a}b}{\tilde{a}+b}M^2},
		\end{split}
		\end{align}
		which gives $2\norm{(\Id-\Pi)\ket{\Psi_{a,b}}}\leq\sqrt{\frac{1}{M}\sqrt{\frac{2}{\pi a}}}e^{-aM^2}$.

		For the first term, we follow the template of \cref{thm:bound-abelian}, replacing the use of \cref{lem:abelian-overlap} with \cref{lem:GKP-overlap}. We take the trivial set of orthogonal permutations $\pi_i(\alpha_j)=\alpha_{i+j}$, where the addition is modulo $N$. Then,
		\begin{align}
		\begin{split}
			&\expec{i}\Tr\squ[\big]{(\Pi A^{\alpha_i}\Pi\otimes B^{\alpha_i\Z}\otimes C^{\alpha_i\Z})(E_{\alpha_i\Z})(\Id\otimes\Phi)(\ketbra{\Psi_{a,b}})}\\
			&\qquad\leq\expec{i}\sup_{j=1,\ldots,N;r,s\in\R}\norm{\Pi A^{\alpha_i}([0,\alpha_i)\times(r+F))\Pi\Pi A^{\alpha_{i+j}}((s+E)\times[0,1/\alpha_{i+j}))\Pi}\\
			&\qquad\leq\frac{1}{N}+\frac{1}{N}\sum_{i\neq 0}\sup_{j}\sqrt{4\parens*{\alpha_i+\frac{2M}{\alpha_{i+j}}}\varepsilon}\leq\frac{1}{N}+2\sqrt{\parens*{\alpha_N+\frac{2M}{\alpha_1}}\varepsilon},
		\end{split}
		\end{align}
		giving the wanted result.
	\end{proof}

\section{The coset monogamy game on $SO(3)$}\label{sec:SO3game}

	We move on to study the case $G=SO(3)$ --- a compact Lie group whose elements label the distinct orientations of a rigid body. This case is akin to the planar rotor (see \cref{sec:U1game}) in that the position states are labelled by elements of a compact space, meaning that the dual irreducible representation indices are discrete. However, there is an extra complication of the group being non-abelian, meaning that irreducible representations become matrix-valued.

	\subsection{Rigid rotor states}
	We consider a molecule with no symmetries as a rigid body in 3D. Then, the configuration space of rotational states of the molecule corresponds to the rotation group $G=SO(3)$. As before, the associated Hilbert space $L^2(SO(3))$ can be spanned by the unnormalizable position eigenstates $\ket{R}$, $R\in SO(3)$. The dual basis consists of the angular momentum eigenstates, given by the Wigner $D$-matrices extending the spherical harmonic basis for $S^2\cong SO(3)/U(1)$. For $\ell\geq 0$, $-\ell\leq m,n\leq\ell$,
	\begin{align}
		\ket{^\ell_{m,n}}=\sqrt{\tfrac{2\ell+1}{8\pi^2}}\int D^\ell_{m,n}(R)\ket{R}dR.
	\end{align}
	These form a bone fide orthonormal basis.

	As before, we can consider coset states of an arbitrary closed subgroup $H\leq G$. Since there are multiple types of interesting subgroups, corresponding to the proper point groups, we will not for the moment specialise to a particular group structure. As for finite non-abelian groups, the characters of the continuous representations no longer span $L^2(H)$; however the matrix elements of the representations do. Write $\tsf{IB}(H)=\set{\gamma_{m,n}}$ for the set of matrix elements of a full collection of inequivalent continuous finite-dimensional irreducible representations of $H$. By the Peter-Weyl theorem, these matrix elements are complete and orthogonal in the sense that
	\begin{align}
		\int_H\overline{\gamma_{m,n}(h)}\gamma'_{m',n'}(h)dh=d_\gamma\delta_{\gamma,\gamma'}\delta_{m,m'}\delta_{n,n'},
	\end{align}
	where $d_\gamma$ is the dimension of the representation $\gamma$. Note that, unlike the abelian case, if $G$ were not compact, $H$ may not be compact, and hence these properties may not hold. On the other hand, $H$ is not necessarily normal, so $G/H$ will not be a group. Since we cannot appeal to the dual group structure as in the abelian case, we fix a set of coset representatives $\tsf{CS}(H)$. We can find simple fundamental domains in all the relevant cases \cite{ACP20}.

	The coset states, indexed by $R\in\tsf{CS}(H)$ and $\gamma_{m,n}\in\tsf{IB}(H)$, generalise as
	\begin{align}
		\ket{RH^\gamma_{m,n}}=\sqrt{d_\gamma}\int_H\gamma_{m,n}(h)\ket{Rh}dh.
	\end{align}
	However, other than the $G=H$ case corresponding to the angular momentum basis, these are not normalizable states: we have $\braket{RH^\gamma_{m,n}}{R'H^{\gamma'}_{m',n'}}=\delta_{G/H}(R^{-1}R'H)\delta_{\gamma_{m,n},\gamma'_{m',n'}}$. Hence, as in the case of infinite abelian groups, we make use of the coset measure. Here, we can endow the coset representatives $\tsf{CS}(H)$ with a continuous measure induced by the quotient measure on the symmetric space $SO(3)/H$; and $\tsf{IB}(H)$ can simply be given a discrete structure via the counting measure. For a measurable set $E=\bigcup_{\gamma_{m,n}}E_{\gamma_{m,n}}\times\{\gamma_{m,n}\}\subseteq\tsf{CS}(H)\times\tsf{IB}(H)$, the operator measure is
	\begin{align}
		A^H(E)=\sum_{\gamma_{m,n}}\int_{E_{\gamma_{m,n}}}\ketbra{RH^\gamma_{m,n}}dR.
	\end{align}

	\subsection{Monogamy game analysis}

	To construct a coset monogamy game for $SO(3)$, we first need to choose a suitable collection of subgroups. One way to do this is to consider subgroups isomorphic to $U(1)$. The canonical example of this is the subgroup of rotations around the $z$ axis, $H_0=\set*{Z(\theta)}{\theta\in[0,2\pi)}$, where
	\begin{align}
		Z(\theta)=\begin{bmatrix}\cos\theta&-\sin\theta&0\\\sin\theta&\cos\theta&0\\0&0&1\end{bmatrix}.
	\end{align}
	We can choose the collection of subgroups to be rotations of this group around the $x$ axis, $H_\theta=X(\theta)H_0X(-\theta)$, where
	\begin{align}
		X(\theta)=\begin{bmatrix}1&0&0\\0&\cos\theta&-\sin\theta\\0&\sin\theta&\cos\theta\end{bmatrix}.
	\end{align}
	Since these groups are all copies of the abelian group $U(1)$, we have that cosets form the symmetric space isomorphic to the sphere, $SO(3)/H_\theta\cong S^2$, and the irreducible representations are simply the characters, indexed by $n\in\Z$ and acting as $\gamma_n(X(\theta)Z(\varphi)X(\theta))=e^{i n\varphi}$. Using Euler angles $\alpha,\beta,\gamma$ to parameterize arbitrary 3D rotations, we can write
	$R=RX(\theta)X(-\theta)=Z(\alpha)X(\beta)Z(\gamma)X(-\theta)$ for $R\in SO(3)$.
	This yields the coset $RH_\theta=\set*{Z(\alpha)X(\beta)Z(\varphi)X(-\theta)}{\varphi\in[0,2\pi)}$, parametrized by the two angles $\alpha,\beta$. This gives coset states of the form
	\begin{align}
		\ket{\theta,\alpha,\beta,n}=\ket{Z(\alpha)X(\beta-\theta)H_\theta^{\gamma_n}}=\frac{1}{2\pi}\int_0^{2\pi}e^{in\varphi}\ket{Z(\alpha)X(\beta)Z(\varphi)X(-\theta)}d\varphi.
	\end{align}
	For some $N\in\N$, fix the set of subgroups $\mc{S}_N=\set*{H_{\theta_k}}{k=0,\ldots,N-1}$, where $\theta_k=\frac{2\pi k}{N}$. It remains to fix the neighborhood of unity corresponding to the correct answers: we take it to be the ball of radius $\varepsilon>0$ in an appropriate metric $E_{\varepsilon}=\set*{R\in SO(3)}{d(R,I)<\varepsilon}$. There are a variety of options for this metric; however, since $SO(3)$ is compact, they are all equivalent. We choose a metric that arises naturally from the quaternion representation of rotations. There is a surjective group homomorphism from the unit quaternions to $SO(3)$ that acts as $\cos\theta+i\sin\theta\mapsto Z(2\theta)$, $\cos\theta+k\sin\theta\mapsto X(2\theta)$. Using the Euclidean norm on the quaternions, we take the metric to be the distance between the preimages, \emph{i.e.}\ for $R,S\in SO(3)$, if $p\mapsto R$ and $q\mapsto S$, $d(R,S):=\min\{\norm{p-q}_2,\norm{p+q}_2\}$. A nice property of this metric is that it is translation-invariant by construction.

	The game then proceeds as usual.

	\begin{enumerate}[1.]
		\item Bob and Charlie prepare a shared state $\rho_{ABC}$ but then are no longer allowed to communicate.

		\item Alice chooses $k=1,\ldots,N$ uniformly at random and measures her register in basis $\set*{\ket{\theta_k,\alpha,\beta,n}}$ to get measurements $\alpha,\beta, n$.

		\item Alice sends $k$ to Bob and Charlie. Bob answers with a guess $(\alpha_B,\beta_B)$ for $(\alpha,\beta)$ and Charlie answers with a guess $n_C$ for $n$.

		\item Bob and Charlie win if there exists $\varphi$ such that $d\parens*{Z(\alpha_B)X(\beta_B),Z(\alpha)X(\beta)Z(\varphi)}<\varepsilon$ and $n=n_C$
	\end{enumerate}

	We bound the winning probability of this game using \cref{thm:compact-coset-bound}.

	\begin{theorem}\label{thm:so3-monogamy}
		Let $N\in\N$ be even and $0<\varepsilon<2\sin\tfrac{\pi}{2N}$. Let the compact coset monogamy game $\ttt{G}_{N,\varepsilon}=(SO(3),\mc{S}_N,E_{\varepsilon})$. Then, the winning probability
		\begin{align}
			\mfk{w}(\ttt{G}_{N,\varepsilon})\leq\frac{2}{N}+2\sqrt{\pi\varepsilon}.
		\end{align}
	\end{theorem}

	First, we compute the overlaps.

	\begin{lemma}\label{lem:so3-overlap}
		For $\theta\in[0,2\pi)$ and $0<\varepsilon<1$,
		\begin{align}
			\sup_{R\in SO(3)}\mu_{H_0}(H_0\cap RE_\varepsilon H_\theta)=\begin{cases}1&\text{if }|\cos\theta|>\cos\eta\\\frac{2}{\pi}\arcsin\sqrt{1-\sqrt{1-\frac{\sin^2\eta}{\sin^2\theta}}}&\text{else}\end{cases},
		\end{align}
		where $\cos\tfrac{\eta}{2}=1-\frac{\varepsilon^2}{2}$
	\end{lemma}

	\begin{proof}
		First, we can assume that $\sin\theta>0$: if $\sin\theta=0$, we have $H_\theta=H_0$ and $|\cos\theta|=1$, so we have that the measure is $1$; if $\sin\theta<0$, we can conjugate by $Z(\pi)$, which does not change $H_0$ and sends $H_\theta\mapsto H_{-\theta}$. To get a description of $H_0\cap RE_\varepsilon H_\theta$, we note that
		\begin{align}
			RE_\varepsilon H_\theta=\{g\in SO(3)|d(g,RX(\theta)Z(\varphi)X(-\theta))<\varepsilon\text{ for some }\varphi\},
		\end{align}
		so  $H_0\cap RE_\varepsilon H_\theta$ is the set of elements of $H_0$ that are less than $\varepsilon$ away from some element of  $RH_\theta$. First, we find the distance between arbitrary elements of $H_0$ and $RH_\theta$. There exist $\alpha,\beta,\gamma\in[0,2\pi)$ such that $R=Z(\alpha)X(\beta)Z(\gamma)X(-\theta)$; therefore as every element of $H_\theta$ is of the form $X(\theta)Z(\chi-\gamma)X(-\theta)$ for some $\chi\in[0,2\pi)$, $Z(\alpha)X(\beta)Z(\chi)X(-\theta)$ parametrizes all the elements of $RH_\theta$. Similarly, every element of $H_0$ may be expressed $Z(\varphi+\alpha)$ for some $\varphi\in[0,2\pi)$, so the distance between those elements is
		\begin{align}
		\begin{split}
			d\parens*{Z(\varphi+\alpha),Z(\alpha)X(\beta)Z(\chi)X(-\theta)}&=d\parens*{Z(\varphi)X(\theta),X(\beta)Z(\chi)}\\
			&\hspace{-1cm}=\min_{\pm}\sqrt{2(1\pm\cos\tfrac{\chi}{2}\cos\tfrac{\varphi}{2}\cos\tfrac{\theta-\beta}{2}\pm\sin\tfrac{\chi}{2}\sin\tfrac{\varphi}{2}\cos\tfrac{\theta+\beta}{2})}.
		\end{split}
		\end{align}
		To find how close $g=Z(\varphi+\alpha)$ is to any element of $RH_\theta$, we need to minimize with respect to $\chi$, yielding
		\begin{align}
			d(g,RH_\theta)=\sqrt{2\parens*{1-\sqrt{\cos^2\tfrac{\varphi}{2}\cos^2\tfrac{\theta-\beta}{2}+\sin^2\tfrac{\varphi}{2}\cos^2\tfrac{\theta+\beta}{2}}}}~.
		\end{align}
		Since $g\in RE_\varepsilon H_\theta$ iff $d(g,RH_\theta)<\varepsilon$, which is iff
		$\cos(\theta-\beta)-\cos(\eta)>\sin^2\tfrac{\varphi}{2}\sin(\beta)\sin(\theta).$ Integrating with normalization $\frac{1}{2\pi}$ gives that the set of $\varphi\in[0,2\pi)$ that satisfy this has measure
		\begin{align}
			\frac{2}{\pi}\arcsin\sqrt{\frac{\cos(\theta-\beta)-\cos\eta}{\sin\theta\sin\beta}}.
		\end{align}
		To maximise the measure with respect to $g$, it remains to maximise this measure with respect to $\beta$. If $|\cos\theta|>\cos\eta$, then taking $\cos\beta=\frac{\cos\eta}{\cos\theta}$ gives that the set is measure $1$. Else, value of $\beta$ that does this is $\beta=\arccos\parens*{\frac{\cos\theta}{\cos\eta}}$. Using this value, the measure becomes $\frac{2}{\pi}\arcsin\sqrt{1-\sqrt{1-\frac{\sin^2\eta}{\sin^2\theta}}}$, giving the result.
	\end{proof}

	Now, we can pass to the proof of the theorem.

	\begin{proof}[Proof of \cref{thm:so3-monogamy}]
		We make use of the bound of \cref{thm:compact-coset-bound}. To do so, we need to choose a collection of orthogonal permutations: we make use of the trivial option $\pi_i(H_{\theta_k})=H_{\theta_{k+i}}$. Then, it is direct to note by conjugation that $\mu_{H_{\theta_k}}(H_{\theta_k}\cap RE_\varepsilon H_{\theta_{k+i}})=\mu_{H_{0}}(H_{0}\cap RE_\varepsilon H_{\theta_{i}})$, giving the bound
		\begin{align}
			\mfk{w}(\ttt{G}_{N,\varepsilon})\leq\expec{i}\sup_{R\in SO(3)}\sqrt{\mu_{H_{0}}(H_{0}\cap RE_\varepsilon H_{\theta_{i}})}=\frac{1}{N}\sum_{i=0}^{N-1}\sup_{R\in SO(3)}\sqrt{\mu_{H_{0}}(H_{0}\cap RE_\varepsilon H_{\theta_{i}})}.
		\end{align}
		Now, using \cref{lem:so3-overlap},
		\begin{align}
			\mfk{w}(\ttt{G}_{N,\varepsilon})\leq\frac{2}{N}+\frac{2}{N}\sum_{i=1}^{N/2-1}\sqrt{\frac{2}{\pi}\arcsin\sqrt{1-\sqrt{1-\frac{\sin^2\eta}{\sin^2\theta_i}}}}.
		\end{align}
		It remains to simplify this upper bound. First, using $\arcsin x\leq\frac{\pi}{2}x$, $\mfk{w}(\ttt{G}_{N,\varepsilon})\leq\frac{2}{N}+\frac{2}{N}\sum_{i=1}^{N/2-1}\parens*{1-\sqrt{1-\frac{\sin^2\eta}{\sin^2\theta_i}}}^{1/4}$. Next, using $\sqrt{1-x^2}\geq 1-x^2$,
		\begin{align}
			\mfk{w}(\ttt{G}_{N,\varepsilon})\leq\frac{2}{N}+\frac{2}{N}\sum_{i=1}^{N/2-1}\parens*{1-\parens*{1-\frac{\sin^2\eta}{\sin^2\theta_i}}}^{1/4}=\frac{2}{N}+\frac{2}{N}\sum_{i=1}^{N/2-1}\sqrt{\frac{\sin\eta}{\sin\theta_i}}.
		\end{align}
		First, we see that $\sin\eta=2\varepsilon\sqrt{1-\frac{\varepsilon^2}{4}}\parens*{1-\frac{\varepsilon^2}{2}}\leq2\varepsilon$. Also, $\sin x\geq\frac{2}{\pi}x$, so
		\begin{align}
			\mfk{w}(\ttt{G}_{N,\varepsilon})\leq\frac{2}{N}+\frac{2\sqrt{\pi\varepsilon}}{N}\sum_{i=1}^{N/2-1}\frac{1}{\sqrt{\theta_i}}.
		\end{align}
		Approximating the sum by an integral,
		\begin{align}
			\mfk{w}(\ttt{G}_{N,\varepsilon})\leq\frac{2}{N}+\sqrt{\pi\varepsilon}\int_0^{1}\frac{1}{\sqrt{x}}dx=\frac{2}{N}+2\sqrt{\pi\varepsilon}.
		\end{align}
	\end{proof}

\section{Monogamy games for finite groups}\label{sec:finite}

	In this section, we introduce the coset monogamy game on non-abelian, but finite, groups. This serves to bridge the gap between the monogamy game on~$\Z_2^n$ of \cite{coladangelo2021hidden} and the infinite-dimensional monogamy games of the following sections.

	\subsection{Non-abelian coset states}

		Let $G$ be an arbitrary finite group. The corresponding Hilbert space $\mc{H}_G=\spn_{\C}\set*{\ket{g}}{g\in G}$, where the $\ket{g}$ form an orthonormal basis, is finite-dimensional. Hence, we can, as in the original $\Z_2^n$ case, work with a basis of coset states rather than a measure.

		Let $H\leq G$ be a subgroup and fix $g\in G$. In order to have a basis of coset states, the states on the coset $gH$ need to be an orthonormal basis of $\spn_{\C}\set*{\ket{gh}}{h\in H}$. For an abelian group, it is sufficient to use the characters, giving coset states of the form
		\begin{align}
			\ket{gH^\chi}=\frac{1}{\sqrt{|H|}}\sum_{h\in H}\chi(h)\ket{gh}.
		\end{align}
		However, if $H$ is non-abelian, then the irreducible representations are not one-dimensional and the characters only span the space of class functions, not the space of all functions on $H$. In this case, the natural generalization of the irreducible characters are the matrix elements of the irreducible representations. Let $\gamma:H\rightarrow\mc{U}(d_\gamma)$ be an irreducible representation, and write the $(m,n)$ matrix element, for $1\leq m,n\leq d_\gamma$ as $\gamma_{m,n}:H\rightarrow\C$.  The associated coset state is
		\begin{align}
			\ket{gH^\gamma_{m,n}}=\sqrt{\frac{d_\gamma}{|H|}}\sum_{h\in H}\gamma_{m,n}(h)\ket{gh}.
		\end{align}
		The orthonormality of these states is given by the Schur orthogonality relations: if $\gamma,\gamma'$ are either equal or inequivalent irreducible representations, and $1\leq m,n\leq d_\gamma$, $1\leq m',n'\leq d_{\gamma'}$, then
		\begin{align}
			\sum_{h\in H}\overline{\gamma_{m,n}(h)}\gamma_{m',n'}(h)=\frac{|H|}{d_\gamma}\delta_{\gamma,\gamma'}\delta_{m,m'}\delta_{n,n'}.
		\end{align}
		However, unlike the characters, the matrix elements of a representation are not unique in the sense that the equivalent but unequal representations may have different matrix elements. This implies that the coset states for different choices of irreducible representations and coset representatives are not equal up to global phase, as they are in the abelian case. To remedy that we need to fix choices of representatives. Let $\tsf{CS}(H)$ be a set of coset representatives of $G/H$, let $\tsf{Irr}(H)$ be a full set of inequivalent irreducible representations (irreps) of $H$, and define $\tsf{IB}(H)=\set{\gamma_{m,n}}{\gamma\in\tsf{Irr}(H);1\leq m,n\leq d_\gamma}$. With this, we take the basis of coset states to be
		\begin{align}
			\set*{\ket{gH^\gamma_{m,n}}}{g\in\tsf{CS}(H),\gamma_{m,n}\in\tsf{IB}(H)}.
		\end{align}

		An interesting property of non-abelian coset states is that, whereas abelian coset states are equal superpositions of the elements of $gH$, non-abelian states are no longer, as we do not have always that $|\gamma_{m,n}(h)|=1$. We give a simple example to illustrate that.

		Consider $G=D_N$, the dihedral group defined with generators and relations $D_N=\gen{r,t}{r^N=t^2=(tr)^2=1}$, for $N$ odd and its subgroup $H=\gen{r^{N/p},t}\cong D_p$ for $p$ a divisor of $N$. There are two one-dimensional irreps of $H$: the trival $\gamma_0(r^{N/p})=\gamma_0(t)=1$, and the representation $\gamma_{-1}(r^{N/p})=1$, $\gamma_{-1}(t)=-1$. The remaining irreps are two-dimensional, indexed by $k=1,\ldots,\frac{1}{2}(p-1)$: $\gamma_k(r^{N/p})=\exp(\frac{2\pi i k}{p}Z)$, $\gamma_k(t)=X$. Further, we may take $\tsf{CS}(H)=\set*{1,r,\ldots,r^{N/p-1}}$. Then, the coset states are
		\begin{align}
		\begin{split}
			&\ket{r^qH^0}=\frac{1}{\sqrt{2p}}\sum_{j=0}^{p-1}\parens*{\ket{r^{Nj/p+q}}+\ket{tr^{Nj/p-q}}}\\
			&\ket{r^qH^{-1}}=\frac{1}{\sqrt{2p}}\sum_{j=0}^{p-1}\parens*{\ket{r^{Nj/p+q}}-\ket{tr^{Nj/p-q}}}\\
			&\ket{r^qH^k_{m,n}}=\frac{1}{\sqrt{p}}\sum_{j=0}^{p-1}e^{(-1)^{m}2\pi i\frac{jk}{p_I}}\ket{t^{m+n}r^{Nj/p+(-1)^{m+n}q}}.
		\end{split}
		\end{align}
		Note in particular that each of the states $\ket{r^qH^k_{m,n}}$ is a superposition of only half the elements of $r^qH$.

	\subsection{Non-abelian coset monogamy game}

		\begin{definition}
			A \emph{non-abelian coset monogamy game} is a pair $\ttt{G}=\parens*{G,\mc{S}}$, where $G$ is a finite group and $\mc{S}$ is a finite set of subgroups of $G$.

			A \emph{(quantum) strategy} for a coset measure game $\ttt{G}$ is a tuple $\ttt{S}=\parens*{\mc{B},\mc{C},B,C,\rho}$, where $\mc{B}$ and $\mc{C}$ are Hilbert spaces, $B=\set{B^H:\tsf{CS}(H)\rightarrow\mc{B}(\mc{B})}{H\in\mc{S}}$ and $C=\set{C^H:\tsf{IB}(H)\rightarrow\mc{B}(\mc{C})}{H\in\mc{S}}$ are collections of POVMs, and $\rho\in\mc{D}(\mc{H}_G\otimes\mc{B}\otimes\mc{C})$ is a shared density operator.

			Let $\ttt{G}$ be a coset game and $\ttt{S}$ be a strategy for it. The \emph{winning probability} of $\ttt{S}$ is
			\begin{align}
				\mfk{w}_{\ttt{G}}(\ttt{S})=\expec{H\in\mc{S}}\!\!\!\sum_{\substack{g\in\tsf{CS}(H)\\\gamma_{m,n}\in\tsf{IB}(H)}}\!\!\!\Tr\squ{(\ketbra{gH^\gamma_{m,n}}\otimes B^H_g\otimes C^H_{\gamma_{m,n}})\rho},
			\end{align}
			where the expection with respect to $H\in\mc{S}$ is uniform.

			The winning probability of $\ttt{G}$ is $\mfk{w}(\ttt{G})=\sup_{\ttt{S}}\mfk{w}_{\ttt{G}}(\ttt{S})$.
		\end{definition}

		Our main result in this section is an analogous bound to the one of \cite{CV22} on the winning probability of such a general game.

		\begin{theorem}\label{thm:coset-finite}
			Let $\ttt{G}=(G,\mc{S})$ be a coset monogamy game, and take a set of mutually orthogonal permutations $\pi_i:\mathcal{S}\rightarrow\mathcal{S}$ for $i=1,...,|\mc{S}|$, permutations such that $\pi_i\circ\pi_j^{-1}$ has no fixed points unless $i=j$. Then,
			\begin{align}
				\mfk{w}(\ttt{G})\leq\expec{i}\max_{\substack{H\in\mathcal{S}\\\gamma\in\textsf{Irr}(H)}}\sqrt{d_\gamma\frac{|H\cap\pi_i(H)|}{|H|}}\,.
			\end{align}
		\end{theorem}

		First, we will make use of a lemma of \cite{tomamichel2013monogamy} to bound the winning probability by overlaps of the measurement operators.

		\begin{lemma}[Lemma 2 in \cite{tomamichel2013monogamy}]\label{lem:sum-bound}
			Let $P_1,\ldots,P_n$ be positive semidefinite operators on a Hilbert space. Then
			\begin{align*}
			\Big\|\sum_{i=1}^nP_i\Big\| &\leq\sum_{i=1}^n\,\max_{j=1,\ldots,n}\,\Big\|\sqrt{\vphantom{P_{\pi_i(j)}}P_j}\sqrt{P_{\pi_i(j)}}\Big\|\;,
			\end{align*}
			where $\pi_1,\ldots,\pi_n$ is any set of mutually orthogonal permutations of $\{1,\ldots,n\}$.
		\end{lemma}

		Next, we need to bound the overlaps of particular projectors related to the coset states.

		\begin{lemma}\label{lem:overlaps-finite}
			Let $H,K\leq G$. Then, for any $\gamma_{m,n}\in\tsf{IB}(H)$ and $q\in\tsf{CS}(K)$,
			\begin{align}
				\norm[\Big]{\sum_{g\in\tsf{CS}(H)}\ketbra{gH^\gamma_{m,n}}\sum_{\varrho_{i,j}\in\tsf{IB}(K)}\ketbra{qK^\varrho_{i,j}}}\leq\sqrt{d_\gamma\frac{|H\cap K|}{|H|}}\,.
			\end{align}
		\end{lemma}

		\begin{proof}
			Note that, since the coset states in each of the sums are orthogonal, we're dealing with the overlap of two projectors. Next, since $\set*{\ket{qK^\varrho_{i,j}}}{\varrho_{i,j}\in\tsf{IB}(K)}$ is an orthonormal basis of $\spn_\C\set*{\ket{qk}}{k\in K}$, we get
			\begin{align}
				\sum_{\varrho_{i,j}\in\tsf{IB}(K)}\ketbra{qK^\varrho_{i,j}}=\sum_{k\in K}\ketbra{qk}=:\Pi_{qK},
			\end{align}
			a diagonal projector. Then, the overlap is
			\begin{align}
				\norm[\Big]{\sum_{g\in\tsf{CS}(H)}\ketbra{gH^\gamma_{m,n}}\sum_{\varrho_{i,j}\in\tsf{IB}(K)}\ketbra{qK^\varrho_{i,j}}}=\norm[\Big]{\Pi_{qK}\sum_{g\in\tsf{CS}(H)}\ketbra{gH^\gamma_{m,n}}\Pi_{qK}}^{1/2}.
			\end{align}
			Now, since $\Pi_{qK}\ket{gH^\gamma_{m,n}}$ is a superposition of basis vectors from $qK\cap gH$, these states are orthogonal, giving that the operators $\Pi_{qK}\ketbra{gH^\gamma_{m,n}}\Pi_{qK}$ over $g\in\tsf{CS}(H)$ are Hermitian with orthogonal ranges. Since $\|\sum_s X_s\|\leq \max_s \|X_s\|$ for any $X_i$ Hermitian with orthogonal ranges,
			\begin{align}
			\begin{split}
				\norm[\Big]{\sum_{g\in\tsf{CS}(H)}\Pi_{qK}\ketbra{gH^\gamma_{m,n}}\Pi_{qK}}&\leq\max_{g\in\tsf{CS}(H)}\norm*{\Pi_{qK}\ketbra{gH^\gamma_{m,n}}\Pi_{qK}}\\
				&=\max_{g\in\tsf{CS}(H)}\braket{gH^\gamma_{m,n}}{\Pi_{qK}}{gH^\gamma_{m,n}}.
			\end{split}
			\end{align}
			We get the result by bounding the inner product
			\begin{align}
				\braket{gH^\gamma_{m,n}}{\Pi_{qK}}{gH^\gamma_{m,n}}=\frac{d_\gamma}{|H|}\sum_{h\in H\cap g^{-1}qK}|\gamma_{m,n}(h)|^2\leq d_\gamma\frac{|H\cap g^{-1}qK|}{|H|}\leq d_\gamma\frac{|H\cap K|}{|H|}.
			\end{align}
		\end{proof}

		Now, we can pass to the proof of the theorem.

		\begin{proof}[Proof of \cref{thm:coset-finite}]
				Let $\ttt{S}=\parens*{\mc{B},\mc{C},B,C,\rho}$ be a strategy for $\ttt{G}$. We may assume that $B$ and $C$ are projective. Writing, for each $H\in\mc{S}$, $\Pi^H=\sum_{g,\gamma_{m,n}}\ketbra{gH^\gamma_{m,n}}\otimes B^H_g\otimes C^H_{\gamma_{m,n}}$, the winning probability may be expressed as $\mfk{w}_{\ttt{G}}(\ttt{S})=\expec{H\in\mc{S}}\Tr(\Pi^H\rho)\leq\norm[\big]{\expec{H\in\mc{S}}\Pi^H}$. Now, we can use \cref{lem:sum-bound}:
				\begin{align}
					\mfk{w}_{\ttt{G}}(\ttt{S})\leq\expec{i}\sup_{H\in\mc{S}}\norm[\big]{\Pi^H\Pi^{\pi_i(H)}},
				\end{align}
				since $(\Pi^H)^2=\Pi^H$. Fixing $H,K\in\mc{S}$, it remains to simplify $\norm*{\Pi^H\Pi^K}$. We upper bound
				\begin{align}
				\begin{split}
					&\Pi^H\leq\sum_{\substack{g\in\tsf{CS}(H)\\\gamma_{m,n}\in\tsf{IB}(H)}}\ketbra{gH^\gamma_{m,n}}\otimes \Id_B\otimes C^H_{\gamma_{m,n}}\\
					&\Pi^K\leq\sum_{\substack{q\in\tsf{CS}(K)\\\varrho_{i,j}\in\tsf{IB}(K)}}\ketbra{qK^\varrho_{i,j}}\otimes B^K_q\otimes\Id_K,
				\end{split}
				\end{align}
				and so we can bound the product
				\begin{align}
				\begin{split}
					\norm{\Pi^H\Pi^K}&\leq\norm[\Big]{\sum_{\substack{g,\gamma_{m,n}\\q,\varrho_{i,j}}}\ket{gH^\gamma_{m,n}}\!\braket{gH^\gamma_{m,n}}{qK^\varrho_{i,j}}\!\bra{qK^\varrho_{i,j}}\otimes B^K_q\otimes C^H_{\gamma_{m,n}}}\\
					&=\max_{q,\gamma_{m,n}}\norm[\Big]{\sum_{\substack{g,\varrho_{i,j}}}\ket{gH^\gamma_{m,n}}\!\braket{gH^\gamma_{m,n}}{qK^\varrho_{i,j}}\!\bra{qK^\varrho_{i,j}}},
				\end{split}
				\end{align}
				since Bob and Charlie's measurement projectors have orthogonal supports. Now, we can use \cref{lem:overlaps-finite} and get $\norm{\Pi^H\Pi^K}\leq\max_{q,\gamma_{m,n}}\sqrt{d_\gamma\frac{|H\cap K|}{|H|}}=\max_{\gamma}\sqrt{d_\gamma\frac{|H\cap K|}{|H|}}$. Hence we get the wanted bound
				\begin{align}
					\mfk{w}_{\ttt{G}}(\ttt{S})\leq\expec{i}\sup_{H,\gamma}\sqrt{d_\gamma\frac{|H\cap \pi_i(H)|}{|H|}}.
				\end{align}
			\end{proof}

	\section{Monogamy games for abelian topological groups }\label{sec:abelian-infinite}

	Here, we generalise the coset monogamy game to arbitrary locally compact Abelian groups, on which Pontryagin duality provides a generalization of the Fourier transform. By replacing the coset states with the appropriate operator-valued measure, we find a family of continuous-variable versions of this game, and via a similar analysis we can bound the winning probability, but as a function of not the number of outcomes but the measurement precision. We also study the version of the game where Alice sends damped coset states and find that the same winning probability bound applies.

	\subsection{Infinite-outcome measurements and coset measures}\label{sec:abelian-coset-measure}

		In the context of the original coset monogamy game, we take the space of classical states to be $\Z_2^n$, the space of $n$-bit strings. We are interested in extending the space of classical states to an arbitrary locally compact (Hausdorff) abelian group $G$, for example $U(1)$, $\Z$, or $\R$. In this way, we will be able to work with games where the answer sets are non-trivially infinite. The Hilbert space corresponding to such a set of classical states is $L^2(G)$, the space of square-integrable functions $G\rightarrow\C$, modulo the subspace of almost-everywhere zero functions. We will write $\ket{\psi}\in L^2(G)$ for a class, and the $\psi:G\rightarrow\C$ for a representative of this class. Also, we write the Haar measure $\mu_G:\scr{B}(G)\rightarrow[0,\infty]$, where $\scr{B}(G)$ is the $\sigma$-algebra of Borel sets, and write the Haar integral as $\int f(g)d_Gg$.

		The subgroups $H\leq G$ we consider will always be closed, since this is necessary and sufficient for the quotient space $G/H$ to be Hausdorff as well. Endowing $H$ with the subspace topology and $G/H$ with the quotient topology, both $H$ and $G/H$ are locally compact abelian groups, so they have Haar measures. In fact, fixing Haar measures on $G$ and $H$, there is a unique Haar measure on $G/H$ such that $\int f(g)d_Gg=\iint f(gh)d_Hhd_{G/H}(gH)$. This holds for a much larger class of topological groups, such as compact groups \cite{Nac76}. As much as possible, we will assume the Haar measures to be normalized such that if $G$ is compact $\mu_G(G)=1$ and else if $G$ is discrete $\mu_G(\{1\})=1$.

		In order to work with games with infinitely many answers, we need to be able to handle quantum measurements with infinitely many outcomes. If the outcomes are discrete, then it is possible to handle the measurement as a POVM in the usual way. However, the set of measurement outcomes will not in general be discrete, in which case problems arise. One way to handle this is to use a basis of unnormalizable states to represent the measurement basis. For example, we can take and unnormalizable basis of $L^2(G)$ as $\set*{\ket{g}}{g\in G}$, where each $\ket{g}$ is a Dirac delta function such that $\langle h|g\rangle=0$ when $g\neq h$ and $\braket{g}{h}=\delta_G(g^{-1}h)$. Of course, none of these are elements of $L^2(G)$, but $\ket{\psi}\in L^2(G)$ can nonetheless be represented as $\ket{\psi}=\int\psi(g)\ket{g}d_Gg$. The drawback of this approach is the difficulty of handling convergence of such expressions, which can make the handling of the probability of a measurement outcome ambiguous. To remedy this, we generalise the measurement with respect to a basis not by a generalised basis, but by an operator-valued generalization of a probability measure. For instance, in the previous example, for a measurable subset $E\subseteq G$, we see the probability of measuring some $g\in E$ on some state $\ket{\psi}\in L^2(G)$ is $\int_E\abs*{\psi(g)}^2d_Gg=\braket{\psi}{\int_E\ketbra{g}d_Gg}{\psi}$.

		\begin{definition}
			Let $X$ be a measurable space, $\mc{H}$ be a Hilbert space, and $\scr{S}_X$ be a $\sigma$-algebra on $X$. An \emph{operator-valued measure} is a map $P:\scr{S}_X\rightarrow\mc{B}(\mc{H})$ that is weakly countably additive, \textit{i.e.}, for all countable collections of disjoint measurable sets $\{E_i\}_{i=1}^\infty\subseteq\scr{S}_X$ and $\ket{\phi},\ket{\psi}\in \mc{H}$, \begin{align}
				\braket{\phi}{P\parens[\Big]{\bigcup_{i=1}^\infty E_i}}{\psi}=\sum_{i=1}^\infty\braket{\phi}{P(E_i)}{\psi}.
			\end{align}
			Additionally, $P$ is
			\begin{itemize}
				\item a \emph{POVM} if $P(E)$ is positive for all $E\in\scr{S}_X$ and $P(X)=\Id$;

				\item \emph{projective} (or spectral) if $P(E\cap F)=P(E)P(F)$ for all $E,F\in \scr{S}_X$;

				\item a \emph{PVM} if it is a projective POVM.
			\end{itemize}
		\end{definition}

		Let $P$ be a POVM measure. For any $\rho\in\mc{D}(L^2(G))$, $E\mapsto\Tr(P(E)\rho)$ is a probability measure, representing the probability of measuring an outcome in $E$. POVM measures constitute a physically natural generalization of usual finite POVMs because of the role they play in the spectral theorem for general Hermitian operators. For any Hermitian operator, $A\in\mc{B}(\mc{H})$, there exists a spectral measure $P:\R\rightarrow\mc{B}(\mc{H})$ such that
		$A=\int\lambda dP(\lambda)$ \cite{BB03}; as such, since any observable in quantum mechanics is represented by a Hermitian operator, the distribution of measurement outcomes is naturally represented by a POVM measure. The role of operator-valued measures in quantum theory has been studied from a variety of perspectives \cite{Hol11, Mor17}.

		Before moving on to the particular measure of interest here, we state a general result on POVM measures that allows us to dilate any POVM measure to a PVM measure.
		\begin{theorem}[Naimark \cite{Pau02}]\label{thm:naimark}
			Let $X$ be a measurable space, $\mc{H}$ be a Hilbert space, and $Q:\scr{S}_X\rightarrow\mc{B}(\mc{H})$ be a POVM measure. Then, there exists a Hilbert space $\mc{K}$, an isometry $V:\mc{H}\rightarrow\mc{K}$, and a PVM measure $P:\scr{S}_X\rightarrow\mc{B}(\mc{H})$ such that, for all measurable $E\subseteq X$,
			\begin{align}
				Q(E)=V^\dag P(E) V.
			\end{align}
		\end{theorem}
		We go into detail about integration with respect to POVM measures in \cref{sec:POVM}.

		The original coset monogamy game, given a subspaces $H\leq\Z_2^n$, utilizes bases indexed by $\Z_2^n/H\times \Z_2^n/H^\perp$. $H^\perp$ is defined using the dot product, which can be seen as the family of homomorphisms $\Z_2^n\rightarrow S^1=\set{z\in\C}{\abs{z}=1}$ indexed by $a\in\Z_2^n$ defined $x\mapsto(-1)^{a\cdot x}$. For a general abelian topoloogical group, this natural family of homomorphisms is given not by $G$ itself but by the dual group $\hat{G}$, which the group of continuous homomorphisms $G\rightarrow S^1$, with product $(\gamma\eta)(g)=\gamma(g)\eta(g)$. $G$ is isomorphic to $\hat{G}$ if $G$ is finite. If $G$ is locally compact abelian, then so is $\hat{G}$. Then, the natural generalization of
		\begin{align}
			H^\perp=\set*{\gamma\in\hat{G}}{\gamma(h)=1\forall h\in H},
		\end{align}
		and we have that $\hat{G}/H^\perp\cong \hat{H}$ homeomorphically by Pontryagin duality. Thus, the basis of coset states should generalise to an operator-valued measure on $G/H\times\hat{H}$. An important tool for working with the function spaces on these groups is the Fourier transform.

		\begin{definition}
			Let $G$ be a locally compact abelian group. The \emph{Fourier transform} is the operator $\mc{F}_G\in\mc{B}(L^2(G),L^2(\hat{G}))$ defined on $\ket{\psi}\in L^2(G)$ continuous with compact support as
			\begin{align}
				(\mc{F}_G\ket{\psi})(\gamma)=\int\psi(g)\overline{\gamma(g)}d_Gg,
			\end{align}
			and extended uniquely to all $L^2(G)$ by continuity.
		\end{definition}

		There exists a unique scaling of the Haar measure on $\hat{G}$, called the dual measure, that makes the Fourier transform an isometry. Using that measure, the inverse Fourier transform on $\ket{\psi}\in L^2(\hat{G})$ continuous with compact support is $(\mc{F}_G^{-1}\ket{\psi})(g)=\int\psi(\gamma)\gamma(g)d_{\hat{G}}g$, which again extends to the whole space by continuity.

		Finally, we need the concept of a Fourier transform with respect to a closed subgroup $H\leq G$. For $\ket{\psi}\in L^2(G)$, we can simply take $\mc{F}_H\ket{\psi}=\mc{F}_H\ket{\psi|_H}:\hat{H}\rightarrow\C$. In general, $\ket{\psi|_H}$ is not always in $L^2(H)$, but in the cases we consider, this will be true almost everywhere.

		\begin{definition}\label{def:coset-measure-abelian}
			Let $G$ be a locally compact abelian group and $H\leq G$ be a closed subgroup. The \emph{coset operator measure} is the map $A^H:\scr{B}(G/H)\otimes\scr{B}(\hat{H})\rightarrow\mathcal{B}(L^2(G))$ defined, for $E\subseteq G/H\times\hat{H}$ measurable and $\ket{\phi},\ket{\psi}\in L^2(G)$ as
			\begin{align}
				\braket{\phi}{A^H(E)}{\psi}=\int_E\overline{(\mc{F}_H\ket{\phi\circ g})(\gamma)}(\mc{F}_H\ket{\psi\circ g})(\gamma)d_{_{G/H\times\hat{H}}}(gH,\gamma),
			\end{align}
			where $\ket{\psi\circ g}(h)=\psi(gh)$.
		\end{definition}

		Note that the integral defining $\braket{\phi}{A^H(E)}{\psi}$ is well-defined, since we have that $\mc{F}_H(\ket{\psi\circ gh})(\gamma)=\gamma(h)\mc{F}_H(\ket{\psi\circ g})(\gamma)$ for all $h\in H$, so $\overline{\mc{F}_H(\ket{\phi\circ g})(\gamma)}\mc{F}_H(\ket{\psi\circ g})(\gamma)$ does not depend on the choice of coset representative. Also, as $\int f(g)d_Gg=\iint f(gh)d_Hhd_{G/H}(gH)$, $\ket{\psi\circ g|_H}$ is in $L^2(H)$ $\mu_{G/H}$-almost everywhere, so the integrand is in fact integrable. From the definition, it is direct that $\braket{\phi}{A^H(E)}{\psi}$ is in fact linear in $\ket{\psi}$ and antilinear in $\ket{\phi}$. We can show that it is bounded and that it is a PVM simultaneously.

		\begin{theorem}
			$A^H$ is a PVM measure.
		\end{theorem}

		\begin{proof}
			First, we have that, for any measurable $E\subseteq G/H\times\hat{H}$ and $\ket{\psi}\in L^2(G)$,
			\begin{align}
				\braket{\psi}{A^H(E)}{\psi}=\int_E \abs*{(\mc{F}_H\ket{\psi\circ g})(\gamma)}^2d_{G/H\times\hat{H}}(gH,\gamma)\geq 0,
			\end{align}
			so $A^H(E)$ is positive. Next, if $E\subseteq E'$, $\braket{\psi}{A^H(E)}{\psi}\leq \braket{\psi}{A^H(E')}{\psi}$. In particular, $\braket{\psi}{A^H(E)}{\psi}\leq \braket{\psi}{A^H(G/H\times\hat{H})}{\psi}$. We have, by Plancherel's theorem, that
			\begin{align}
			\begin{split}
				\braket{\psi}{A^H(G/H\times\hat{H})}{\psi}&=\int_{G/H}\int_{\hat{H}} \abs*{(\mc{F}_H\ket{\psi\circ g})(\gamma)}^2d_{\hat{H}}\gamma d_{G/H}(gH)\\
				&=\int_{G/H}\int_H\abs*{\ket{\psi\circ g}(h)}^2d_Hhd_{G/H}(gH)\\
				&=\int_{G/H}\int_H\abs*{\psi(gh)}^2d_Hhd_{G/H}(gH)\\
				&=\int_G|\psi(g)|^2d_Gg=\braket{\psi}.
			\end{split}
			\end{align}
			Thus, $A^H(G/H\times\hat{H})$ is bounded, and therefore $A^H(E)$ is bounded.

			To get that $A^H$ is a POVM, we need next that it is weakly countably additive. Let $\{E_i\}_{i=1}^\infty$ be a countable collection of disjoint measurable sets in $G/H\times\hat{H}$. Then, by monotone convergence
			\begin{align}
			\begin{split}
				\braket{\psi}{A^H\parens[\Big]{\bigcup_{i=1}^\infty E_i}}{\psi}&=\int_{\bigcup_{i=1}^\infty E_i}\abs*{(\mc{F}_H\ket{\psi\circ g})(\gamma)}^2d_{G/H\times\hat{H}}(gH,\gamma)\\
				&=\int\sum_{i=1}^\infty\chi_{E_i}(gH,\gamma)\abs*{(\mc{F}_H\ket{\psi\circ g})(\gamma)}^2d_{G/H\times\hat{H}}(gH,\gamma)\\
				&=\sum_{i=1}^\infty\int\chi_{E_i}(gH,\gamma)\abs*{(\mc{F}_H\ket{\psi\circ g})(\gamma)}^2d_{G/H\times\hat{H}}(gH,\gamma)\\
				&=\sum_{i=1}^\infty\braket{\psi}{A^H(E_i)}{\psi}.
			\end{split}
			\end{align}
			Using the polarization identity gives the full form of weak countable additivity.

			Finally, we want to show that $A^H$ is projective. To see that, note $(A^H\ket{\psi})(g)=\int_{(E)_{gH}}\parens*{\mc{F}_H\ket{\psi\circ g}}(\gamma)d_{\hat{H}}\gamma$ almost everywhere, where $(E)_{gH}=\set{\gamma\in\hat{H}}{(gH,\gamma)\in E}$ is the cross-section of $E$ at $gH$. In fact, for any $\ket{\phi}\in L^2(G)$ continuous with compact support, it follows that
			\begin{align}
			\begin{split}
				\braket{\phi}{A^H(E)}{\psi}&=\int_E\int_H\overline{\phi(gh)}\gamma(h)d_Hh(\mc{F}_H\ket{\psi\circ g})(\gamma)d_{_{G/H\times\hat{H}}}(gH,\gamma)\\
				&=\int_{\hat{H}}\chi_{(E)_{gH}}(\gamma)\int_{G/H}\int_H\overline{\phi(gh)}(\mc{F}_H\ket{\psi\circ gh})(\gamma)d_Hhd_{G/H}gHd_{\hat{H}}\gamma\\
				&=\int_{(E)_{gH}}\int_G\overline{\phi(g)}(\mc{F}_H\ket{\psi\circ g})(\gamma)d_Ggd_{\hat{H}}\gamma\\
				&=\int_G\overline{\phi(g)}\int_{(E)_{gH}}\parens*{\mc{F}_H\ket{\psi\circ g}}(\gamma)d_{\hat{H}}\gamma d_Gg.
			\end{split}
			\end{align}
			Thus, for any measurable $E,F$, we get that $(A^H(E)A^H(F)\ket{\psi})(g)=\int_{(E)_{gH}}\parens*{\mc{F}_H(A^H(F)\ket{\psi}\circ g)}(\gamma)d_{\hat{H}}\gamma$ and since, for $h\in H$,
			\begin{align}
			\begin{split}
				(A^H(F)\ket{\psi}\circ g)(h)&=(A^H(F)\ket{\psi})(gh)=\int_{(F)_{ghH}}(\mc{F}_H\ket{\psi\circ gh})(\gamma)d_{\hat{H}}\gamma\\
				&=\int_{(F)_{gH}}\gamma(h)(\mc{F}_H\ket{\psi\circ g})(\gamma)d_{\hat{H}}\gamma\\
				&=\mc{F}_H^{-1}\parens{\chi_{(F)_{gH}}\mc{F}_H\ket{\psi\circ g}}(h),
			\end{split}
			\end{align}
			if $\mc{F}_H\ket{\psi\circ g}$ is continuous with compact support. By Fourier inversion, this set is dense, so this identity holds for all $\ket{\psi}$. As such,
			\begin{align}
			\begin{split}
				(A^H(E)A^H(F)\ket{\psi})(g)&=\int_{(E)_{gH}}\parens*{\mc{F}_H\parens*{\mc{F}_H^{-1}\parens{\chi_{(F)_{gH}}\mc{F}_H\ket{\psi\circ g}}}}(\gamma)d_{\hat{H}}\gamma\\
				&=\int_{(E)_{gH}}\chi_{(F)_{gH}}(\gamma)\parens{\mc{F}_H\ket{\psi\circ g}}(\gamma)d_{\hat{H}}\gamma\\
				&=\int_{(E\cap F)_{gH}}\parens{\mc{F}_H\ket{\psi\circ g}}(\gamma)d_{\hat{H}}\gamma=(A^H(E\cap F)\ket{\psi})(g).
			\end{split}
			\end{align}
		\end{proof}

	\subsection{The coset measure game}\label{sec:abelian-coset-game}

	\begin{definition}
		An \emph{abelian coset measure monogamy game} is a tuple $\ttt{G}=\parens*{G,\mc{S},E,F}$, where $G$ is a locally compact Hausdorff abelian group, $\mc{S}$ is a finite set of closed subgroups of $G$, and $E\subseteq G$ and $F\subseteq\hat{G}$ are symmetric neighborhoods of identity, \emph{i.e.}\ $E=E^{-1}$ and $1\in E$ and the same for $F$.

		A \emph{(quantum) strategy} for a coset measure game $\ttt{G}$ is a tuple $\ttt{S}=\parens*{\mc{B},\mc{C},B,C,\rho}$, where $\mc{B}$ and $\mc{C}$ are Hilbert spaces, $B=\set{B^H:\scr{B}(G/H)\rightarrow\mc{B}(\mc{B})}{H\in\mc{S}}$ and $C=\set{C^H:\scr{B}(\hat{H})\rightarrow\mc{B}(\mc{C})}{H\in\mc{S}}$ are collections of POVM measures, and $\rho\in\mc{D}(L^2(G)\otimes\mc{B}\otimes\mc{C})$.

		Let $\ttt{G}$ be a coset measure game and $\ttt{S}$ be a strategy for it. The \emph{winning probability} of $\ttt{S}$ is
		\begin{align}
			\mfk{w}_{\ttt{G}}(\ttt{S})=\expec{H\in\mc{S}}\Tr\squ*{(A^H\otimes B^H\otimes C^H)(E_H)\rho},
		\end{align}
		where $E_H=\set{(gH,\gamma|_H,egH,(\varphi\gamma)|_H)}{g\in G,\gamma\in\hat{G},e\in E,\varphi\in F}\in G/H\times\hat{H}\times G/H\times\hat{H}$ and the expection with respect to $H\in\mc{S}$ is uniform.

		The winning probability of $\ttt{G}$ is $\mfk{w}(\ttt{G})=\sup_{\ttt{S}}\mfk{w}_{\ttt{G}}(\ttt{S})$.
	\end{definition}

	\begin{theorem}\label{thm:bound-abelian}
		Let $\ttt{G}=(G,\mc{S},E,F)$ be a coset measure game for second-countable $G$, and let $\pi^i:\mc{S}\rightarrow\mc{S}$ for $i=1,\ldots,|\mc{S}|$ be a set of orthogonal permutations, \emph{i.e.}\ $\pi_i\circ\pi_j^{-1}$ has no fixed points unless $i=j$. Then, the winning probability of $\ttt{G}$ is bounded above as
		\begin{align}
			\mfk{w}(\ttt{G})\leq\expec{i}\sup_{H\in\mc{S},g\in G}\sqrt{\mu_H\parens[\big]{H\cap gE\pi_i(H)}\mu_{\hat{H}}\parens[\big]{F}}.
		\end{align}
	\end{theorem}

	\begin{lemma}\label{lem:int-bound}
		Let $\mc{H}$ be a Hilbert space and $\mc{K}$ be a separable Hilbert space, $X$ be a measurable space, $P:\scr{S}_X\rightarrow\mc{B}(\mc{H})$ be a POVM measure, and $F:X\rightarrow\mc{B}(\mc{K})$ be a bounded weakly measurable function. Then,
		\begin{align}
			\norm[\Big]{\int F\otimes dP}\leq\sup_{x\in X}\norm{F(x)}.
		\end{align}
	\end{lemma}

	\begin{proof}
		Let $\varepsilon>0$. Using \cref{lem:simp-limit}, there exists a sequence of simple functions $\parens{F_n=\sum_lM^n_l\chi_{E^n_l}}$ that converges strongly pointwise to $F$ and $\norm{F_n(x)}\leq\sup_y\norm{F(y)}+\varepsilon$. Since this implies $\sup_n\norm*{\int F_n\otimes dP}<\infty$, \cref{thm:int-simp-limit} gives that $\int F_n\otimes dP\rightarrow\int F\otimes dP$ weakly. We know that there exist $\ket{u},\ket{v}\in\mc{K}\otimes\mc{H}$ with $\norm{\ket{u}},\norm{\ket{v}}\leq 1$ such that $\abs*{\braket{u}{\int F\otimes dP}{v}}>\norm*{\int F\otimes dP}-\varepsilon$. Then, using the weak convergence, there exists $N\in\N$ such that $\abs*{\braket{u}{\int (F_n-F)\otimes dP}{v}}<\varepsilon$ $\forall n\geq N$, so
		\begin{align}
		\begin{split}
			\norm[\Big]{\int F\otimes dP}&<\abs[\Big]{\braket{u}{\int F\otimes dP}{v}}+\varepsilon\\
			&<\abs[\Big]{\braket{u}{\int F_n\otimes dP}{v}}+2\varepsilon\\
			&\leq\norm[\Big]{\sum_l M^n_l\otimes P(E^n_l)}+2\varepsilon\\
			&\leq\sup_l\norm{M^n_l}+2\varepsilon\\
			&<\sup_{x\in X}\norm{F(x)}+3\varepsilon.
		\end{split}
		\end{align}
		As this holds for all $\varepsilon>0$, we get $\norm[\Big]{\int F\otimes dP}\leq \sup_{x\in X}\norm{F(x)}$ as wanted.
	\end{proof}

	\begin{lemma}\label{lem:abelian-overlap}
		Let $H,K\leq G$ be closed subgroups, and let $E\subseteq G$ and $F\subseteq\hat{G}$ be symmetric neighborhoods of the identity. Then, for any $\eta\in\hat{G}$ and $qK\in G/K$
		\begin{align}
			\norm*{A^H\parens[\big]{G/H\times(F\eta)|_H}A^K\parens[\big]{EqK/K\times\hat{K}}}\leq\sup_{g\in G}\sqrt{\mu_H\parens[\big]{H\cap gEK}\mu_{\hat{H}}\parens[\big]{F}}
		\end{align}
	\end{lemma}

	\begin{proof}
		First, we use Fourier inversion as a continuous-valued version of orthogonality of the characters to simplify $A^K\parens[\big]{EqK/K\times\hat{K}}$. For any $\ket{\psi}\in L^2(G)$ continuous with compact support and $\norm{\ket{\psi}}=1$, and $g\in G$, since $(EqK/K\times \hat{K})_{gK}=\hat{K}$ if $g\in EqK$ and $\varnothing$ otherwise,
		\begin{align}
		\begin{split}
			\parens*{A^K\parens[\big]{EqK/K\times\hat{K}}\ket{\psi}}(g)&=\chi_{EqK}(g)\int_{\hat{K}}(\mc{F}_K\ket{\psi\circ g})(\gamma)d_{\hat{K}}\gamma\\
			&=\chi_{EqK}(g)(\mc{F}_{K}^{-1}\mc{F}_K\ket{\psi\circ g})(1)\\
			&=\chi_{EqK}(g)\ket{\psi\circ g}(1)=\chi_{EqK}(g)\psi(g).
		\end{split}
		\end{align}
		Write $A^K\parens[\big]{EqK/K\times\hat{K}}=\Pi_{EqK}$, the projector onto the subspace of $L^2(G)$ where $\psi(g)=0$ if $g\notin EqK$. Then, using the Cauchy-Schwarz inequality,
		\begin{align}
		\begin{split}
			\norm{A^H(G/H\times(F\eta)|_H)\Pi_{EqK}\ket{\psi}}^2&=\int_{G/H\times(F\eta)|_H}\abs*{(\mc{F}_H(\Pi_{EqK}\ket{\psi})\circ g)(\gamma)}^2d_{_{G/H\times\hat{H}}}(gH,\gamma)\\
			&=\int_{G/H\times(F\eta)|_H}\abs*{\int_H\chi_{EqK}(gh)\psi(gh)\overline{\gamma(h)}d_Hh}^2d_{_{G/H\times\hat{H}}}(gH,\gamma)\\
			&\leq\int_{G/H\times(F\eta)|_H}\int_H\abs*{\psi(gh)}^2d_Hh\int_H\abs*{\chi_{EqK}(gh)\overline{\gamma(h)}}^2d_Hhd_{_{G/H\times\hat{H}}}(gH,\gamma)\\
			&=\int_{G/H\times(F\eta)|_H}\int_H\abs*{\psi(gh)}^2d_Hh\mu_H(H\cap g^{-1}EqK)d_{_{G/H\times\hat{H}}}(gH,\gamma)\\
			&\leq\int_{G/H}\int_H\abs*{\psi(gh)}^2d_Hhd_{G/H}gH\sup_{g\in G}\mu_H(H\cap g^{-1}EqK)\mu_{\hat{H}}{(F\eta)|_H}\\
			&\leq\norm{\ket{\psi}}^2\sup_{g\in G}\mu_H(H\cap g^{-1}EqK)\mu_{\hat{H}}\parens*{(F\eta)|_H}.
		\end{split}
		\end{align}
		Finally, we note by Haar invariance that $\mu_{\hat{H}}\parens*{(F\eta)|_H}=\mu_{\hat{H}}\parens*{F(\eta|_H)}=\mu_{\hat{H}}\parens*{F}$, giving the result.
	\end{proof}

	\begin{proof}[Proof of \cref{thm:bound-abelian}]
		Let $\ttt{S}=\parens*{\mc{B},\mc{C},B,C,\rho}$ be a strategy for $\ttt{G}=\parens*{G,\mc{S},E,F}$. Due to Naimark's theorem \cref{thm:naimark}, we may assume that $B$ and $C$ are PVMs. Writing, for each $H\in\mc{S}$, $\Pi^H=(A^H\otimes B^H\otimes C^H)(E_H)$, the winning probability may be expressed as $\mfk{w}_{\ttt{G}}(\ttt{S})=\expec{H\in\mc{S}}\Tr(\Pi^H\rho)$. Since the overlap result \cref{lem:sum-bound} doesn't depend on the dimension of the underlying space, we apply it to upper bound the winning probability of the game, getting
		\begin{align}
			\mfk{w}_{\ttt{G}}(\ttt{S})\leq\norm[\Big]{\expec{S\in\mc{S}}\Pi^H}\leq\expec{i}\sup_{H\in\mc{S}}\norm*{\Pi^H\Pi^{\pi_i(H)}},
		\end{align}
		since $(\Pi^H)^2=\Pi^H$. Fixing $H,K\in\mc{S}$, it remains to simplify $\norm*{\Pi^H\Pi^K}$. Given that
		\begin{align}
		\begin{split}
			E_H&\subseteq \set*{(gH,\gamma|_H,g'H,(\varphi\gamma)|_H)}{g,g'\in G;\gamma\in\hat{G},\varphi\in F}\\
			&=\bigcup_{\gamma\in\hat{G}}G/H\times\{\gamma|_H\}\times G/H\times (F\gamma)|_H,
		\end{split}
		\end{align}
		we get that
		\begin{align}
		\begin{split}
			\Pi^H&\leq\int_{\hat{H}}(A^H\otimes B^H)(G/H\times (F\eta)|_H\times G/H)\otimes dC^H(\eta)\\
			&=\int_{\hat{H}}A^H(G/H\times (F\eta)|_H)\otimes\Id_B\otimes dC^H(\eta).
		\end{split}
		\end{align}
		Similarly, $E_K\subseteq \bigcup_{g\in G}\{gK\}\times\hat{K}\times EgK/K\times\hat{K}$, so $\Pi^K\leq\int_{G/H}A^K(EgK/K\times\hat{K})\otimes dB^K(gK)\otimes\Id_C$. Thus, we may bound the norms
		\begin{align}
		\begin{split}
			\norm*{\Pi^H\Pi^K}&\leq\norm[\Big]{\int \!A^H(G/H\times (F\eta)|_H)\otimes\Id_B\otimes dC^H(\eta)\int\! A^K(EgK/K\times\hat{K})\otimes dB^K(gK)\otimes\Id_C}\\
			&=\norm[\Big]{\int A^H(G/H\times (F\eta)|_H)A^K(EgK/K\times\hat{K}) \otimes d(B^K\otimes C^H)(gK,\eta)}.
		\end{split}
		\end{align}
		Using \cref{lem:int-bound},  $\norm*{\Pi^H\Pi^K}\leq \sup_{gK\in G/K,\eta\in\hat{H}}\norm*{A^H(G/H\times (F\eta)|_H)A^K(EgK/K\times\hat{K})}$; and then using \cref{lem:abelian-overlap}. $\norm*{\Pi^H\Pi^K}\leq\sup_{g\in G}\sqrt{\mu_H(H\cap gEK)\mu_{\hat{H}}(F)}$. Thus, putting it together,
		\begin{align}
			\mfk{w}_{\ttt{G}}(\ttt{S})\leq\expec{i}\sup_{H\in\mc{S},g\in G}\sqrt{\mu_H(H\cap gE\pi_i(H))\mu_{\hat{H}}(F)},
		\end{align}
		and since this holds for every strategy, we have the wanted result.
	\end{proof}

	Note that it follows directly from the proof that we also have the bound
	\begin{align}
		\mfk{w}(\ttt{G})\leq\expec{i}\sup_{H\in\mc{S},g\in G}\min\set*{1,\sqrt{\mu_H\parens[\big]{H\cap gE\pi_i(H)}\mu_{\hat{H}}\parens[\big]{F}}},
	\end{align}
	which may be better in the case that the Haar measure of one of the groups is not $1$.

	\subsection{State-sending version of the game}

	In this section, we additionally assume that, for the subgroups $H\leq G$ we consider, the Haar measures on $G/H$ and $\hat{H}$ are $\sigma$-finite, that is the sets may be written as the countable union of sets of finite measure. This holds for all our groups of interest here, in particular $\R^n$, $\C$, and $U(1)$, which appear above.

	We use approximate maximally entangled states to be able to interpret the coset monogamy game as a game where Alice samples a random pair $(gH,\gamma)$, prepares a damped version of the associated unnormalizable coset states, and then sends to it to Bob and Charlie through an adversarially-chosen channel. This will allow us to generalise the original interpretation of the coset monogamy game as a quantum encoding of classical messages as in \cite{coladangelo2021hidden}.

	\begin{definition}
		Let $H\leq G$ be a closed subgroup, let $c:L^2(G)\rightarrow L^2(G)$ be a complex conjugate, and $(\Delta_n)$ be a damping sequence in $\mc{B}(L^2(G))$ that $\mu_{G/H}\times\mu_{\hat{H}}$-damps $A^H$. By \cref{lemma:radon-nikodym}, there exists measurable $\rho_n:G/H\times\hat{H}\rightarrow\mc{D}(L^2(G))$ and integrable $\pi_n:G/H\times\hat{H}\rightarrow[0,\infty)$ such that $\Delta_n^\dag A^H(E)\Delta_n=\int\rho_n(gH,\gamma)\pi_n(gH,\gamma)d(gH,\gamma)$. The \emph{state associated to $H$} is the measurable function \begin{align}
		\begin{split}
			&\sigma^H_n:gH\times\hat{H}\rightarrow\mc{D}(L^2(G))\\
			&\sigma^H_n(gH,\gamma)=c\rho_n(gH,\gamma)c;
		\end{split}
		\end{align}
		and the \emph{probability measure associated to $H$} is
		\begin{align}
		\begin{split}
			&\mu^H_n:\scr{B}(G/H)\otimes\scr{B}(\hat{H})\rightarrow[0,1]\\
			&\mu^H_n(E)=\frac{1}{\norm{\Delta_n}_2^2}\int_E\pi_n(gH,\gamma)d(gH,\gamma).
		\end{split}
		\end{align}
	\end{definition}
	$\mu^H_n$ is in fact a probability measure, as
	\begin{align}
		\mu^H_n(G/H\times\hat{H})=\frac{1}{\norm{\Delta_n}_2^2}\int\Tr(\rho_n)\mu_nd\mu=\frac{\Tr(\Delta_n^\dag A^H(G/H\times\hat{H})\Delta_n)}{\norm{\Delta_n}_2^2}=\frac{\Tr(\Delta_n^\dag\Delta_n)}{\norm{\Delta_n}_2^2}=1,
	\end{align}
	and we have $\frac{1}{\norm{\Delta_n}_2^2}c\Delta_n^\dag A^H(E)\Delta_nc=\int_E\sigma^H_nd\mu^H_n$.

	It is also important to note that, due to the definition of $A^H$ in terms of the measure $\mu_{G/H}\times\mu_{\hat{H}}$, \emph{any} damping sequence $\mu_{G/H}\times\mu_{\hat{H}}$-damps it.

	\begin{definition}
		Let $\ttt{G}=\parens*{G,\mc{S},E,F}$ be a coset monogamy game, $c:L^2(G)\rightarrow L^2(G)$ be a complex conjugate, and $(\Delta_n)$ be a damping sequence in $\mc{B}(L^2(G))$ that $\mu_{G/H}\times\mu_{\hat{H}}$-damps $A^H$ for each $H\in\mc{S}$. The \emph{state-sending version} of $\ttt{G}$ is the sequence of tuples $\ttt{G}_n=(G,\mc{S},E,F,\Delta_n,c)$.

		A \emph{strategy} for $\ttt{G}_n$ is a tuple $\ttt{S}=(\mc{B},\mc{C},B,C,\Phi)$, where $\mc{B}$ and $\mc{C}$ are the Hilbert spaces held by Bob and Charlie, $B=\set{B^H:\scr{B}(G/H)\rightarrow\mc{B}(\mc{B})}{H\in\mc{S}}$ and $C=\set{C^H:\scr{B}(\hat{H})\rightarrow\mc{B}(\mc{C})}{H\in\mc{S}}$ are collections of POVM measures, and $\Phi:\mc{T}_1(L^2(G))\rightarrow\mc{T}_1(\mc{B}\otimes\mc{C})$ is a completely positive trace preserving map.

		The \emph{winning probability} of $\ttt{S}$ is
		\begin{align}
			\mfk{w}_{\ttt{G}_n}(\ttt{S})=\expec{H\in\mc{S}}\int\Tr\squ*{(B^H\otimes C^H)(EgH\times(F\gamma)|_H)\Phi(\sigma^H_n(gH,\gamma))}d\mu^H_n(gH,\gamma).
		\end{align}
		The winning probability of $\ttt{G}_n$ is $\mfk{w}(\ttt{G}_n)=\sup_{\ttt{S}}\mfk{w}_{\ttt{G}_n}(\ttt{S})$.
	\end{definition}
	Now, we bound the winning probability of the state-sending version by the winning probability monogamy version.

	\begin{lemma}\label{lemma:max-to-damp}
		Let $\mc{H}$ be a Hilbert space, $(\Delta_n)$ in $\mc{B}(\mc{H})$ be damping sequence, $c:\mc{H}\rightarrow\mc{H}$ be a complex conjugate, and $(\ket{\Psi_n})$ be the associated approximate maximally entangled state. Writing $\Tr_1:\mc{T}_1(\mc{H}\otimes\mc{H})\rightarrow\mc{T}_1(\mc{H})$ the partial trace on the first copy on $\mc{H}$, we have that for any $A\in\mc{B}(\mc{H})$,
		\begin{align}
			\Tr_1\squ*{(A\otimes\Id)\ketbra{\Psi_n}}=\frac{1}{\norm{\Delta_n}_2^2}c\Delta_n^\dag A^\dag\Delta_n c.
		\end{align}
	\end{lemma}

	\begin{proof}
		Using the singular-value decomposition $\Delta_n=\sum_{i=1}^\infty s_{n,i}\ketbra{\phi_{n,i}}{\chi_{n,i}}$, the associated approximate maximally entangled state $\ket{\Psi_n}=\frac{1}{\norm{\Delta_n}_2}\sum_{i=1}^\infty s_{n,i}\ket{\phi_{n,i}}\otimes c\ket{\chi_{n,i}}$. Then,
		\begin{align}
		\begin{split}
			\Tr_1\squ*{(A\otimes\Id)\ketbra{\Psi_n}}&=\frac{1}{\norm{\Delta_n}_2^2}\sum_{i,j} s_{n,i}s_{n,j}\braket{\phi_{n,j}}{A\phi_{n,i}}\ketbra{c\chi_{n,i}}{c\chi_{n,j}}\\
			&=\frac{1}{\norm{\Delta_n}_2^2}\sum_{i,j} s_{n,i}s_{n,j}\braket{cA\phi_{n,i}}{c\phi_{n,j}}\ketbra{c\chi_{n,i}}{c\chi_{n,j}}\\
			&=\frac{1}{\norm{\Delta_n}_2^2}\sum_{i,j} s_{n,i}s_{n,j}\ketbra{c\chi_{n,i}}{c\phi_{n,i}}(cAc)^\dag\ketbra{c\phi_{n,j}}{c\chi_{n,j}}.
		\end{split}
		\end{align}
		To simplify this, note that for any $\ket{\phi},\ket{\psi}\in\mc{H}$,
		\begin{align}
			\sum_{i} s_{n,i}\braket{\psi}{c\chi_{n,i}}\!\!\braket{c\phi_{n,i}}{\phi}=\sum_{i} s_{n,i}\braket{c\phi}{\phi_{n,i}}\!\!\braket{\chi_{n,i}}{c\psi}=\braket{c\phi}{\Delta_nc\psi}=\braket{c\Delta_nc\psi}{\phi}=\braket{\psi}{(c\Delta_nc)^\dag\phi},
		\end{align}
		so $\Tr_1\squ*{(A\otimes\Id)\ketbra{\Psi_n}}=\frac{1}{\norm{\Delta_n}_2^2}(c\Delta_n c)^\dag(cAc)^\dag(c\Delta_n c)$. Finally, note that conjugation by a complex conjugate commutes with the adjoint: for any $T\in\mc{B}(\mc{H})$,
		\begin{align}
			\braket{\psi}{(cTc)^\dag}{\phi}=\braket{cTc\psi}{\phi}=\braket{c\phi}{Tc\psi}=\braket{T^\dag c\phi}{c\psi}=\braket{\psi}{cT^\dag c\phi},
		\end{align}
		and therefore $\Tr_1\squ*{(A\otimes\Id)\ketbra{\Psi_n}}=\frac{1}{\norm{\Delta_n}_2^2}c\Delta_n^\dag ccA^\dag cc\Delta_n c=\frac{1}{\norm{\Delta_n}_2^2}c\Delta_n^\dag A^\dag\Delta_n c$.
	\end{proof}

	\begin{theorem}\label{thm:state-sending-game}
		Let $\ttt{G}_n=(G,\mc{S},E,F,\Delta_n,c)$ be the state-sending version of a coset monogamy game $\ttt{G}$, and $(\ket{\Psi_n})$ be the approximate maximally entangled state associated to $(\Delta_n)$. For any strategy $\ttt{S}=(\mc{B},\mc{C},B,C,\Phi)$ for $\ttt{G}_n$, write $\ttt{S}_n=\parens*{\mc{B},\mc{C},B,C,(\Id\otimes\Phi)(\ketbra{\Psi_n})}$, which is a strategy for $\ttt{G}$. We have that
		\begin{align}
			\mfk{w}_{\ttt{G}_n}(\ttt{S})=\mfk{w}_{\ttt{G}}(\ttt{S}_n).
		\end{align}
	\end{theorem}

	In particular, this implies that $\mfk{w}(\ttt{G}_n)\leq\mfk{w}(\ttt{G})$. However, we do not necessarily have that $\lim_{n\rightarrow\infty}\mfk{w}(\ttt{G}_n)=\mfk{w}(\ttt{G})$, as this does not even always hold in the case of finite information.

	\begin{proof}
		First, we may write
		\begin{align}
		\begin{split}
			\mfk{w}_{\ttt{G}}(\ttt{S}_n)&=\expec{H\in\mc{S}}\Tr\squ*{(A^H\otimes B^H\otimes C^H)(E_H)(\Id\otimes\Phi)(\ketbra{\Psi_n})}\\
			&=\expec{H\in\mc{S}}\Tr\squ*{\parens*{\int dA^H(gH,\gamma)\otimes(B^H\otimes C^H)(EgH\times (F\gamma)|_H)}(\Id\otimes\Phi)(\ketbra{\Psi_n})}.
		\end{split}
		\end{align}
		Fix some $H\in\mc{S}$ and let $(F_k=\sum_{i}M^k_i\chi_{E^k_i})$ be a sequence of simple functions $G/H\times\hat{H}\rightarrow\mc{B}(\mc{B}\otimes\mc{H})$ that converges strongly pointwise to $(gH,\gamma)\mapsto(B^H\otimes C^H)(EgH\times (F\gamma)|_H)$. Then.
		\begin{align}
		\begin{split}
			\Tr\squ*{\parens*{\int dA^H\otimes F_k}(\Id\otimes\Phi)(\ketbra{\Psi_n})}&=\sum_i\Tr\squ*{\parens*{A^H(E^k_i)\otimes M^k_i}(\Id\otimes\Phi)(\ketbra{\Psi_n})}\\
			&=\sum_i\Tr\squ*{M^k_i\Phi\parens*{\Tr_1\squ*{(A^H(E^k_i)\otimes\Id)\ketbra{\Psi_n}}}}\\
			&=\sum_i\frac{1}{\norm{\Delta_n}_2^2}\Tr\squ*{M^k_i\Phi(c\Delta_n^\dag A^H(E^k_i)\Delta_n c)}\\
			&=\sum_i\int_{E^k_i}\Tr\squ*{M^k_i\Phi(\sigma^H_n)}d\mu^H_n\\
			&=\int\Tr\squ*{F_k(gH,\gamma)\Phi(\sigma^H_n(gH,\gamma))}d\mu^H_n(gH,\gamma).
		\end{split}
		\end{align}
		By weak convergence of the integral, we get that
		\begin{align}
		\begin{split}
			\mfk{w}_{\ttt{G}}(\ttt{S}_n)&=\lim_{k\rightarrow\infty}\expec{H\in\mc{S}}\int\Tr\squ*{F_k(gH,\gamma)\Phi(\sigma^H_n(gH,\gamma))}d\mu^H_n(gH,\gamma)\\
			&=\expec{H\in\mc{S}}\int\Tr\squ*{(B^H\otimes C^H)(EgH\times (F\gamma)|_H)\Phi(\sigma^H_n(gH,\gamma))}d\mu^H_n(gH,\gamma)\\
			&=\mfk{w}_{\ttt{G}_n}(\ttt{S}).
		\end{split}
		\end{align}
	\end{proof}

	\section{Monogamy games for compact groups}\label{sec:compact}

	In this section, we extend the bound to games on compact groups -- groups that may be both non-abelian and infinite. However, the compactness is required here, unlike in the abelian case of \cref{sec:abelian-infinite}, in order to be able to make use of the Peter-Weyl theorem when considering the representations. This provides a class of groups that, although infinite, are conceptually more similar to the finite groups of \cref{sec:finite}.

	\subsection{Compact coset measure}

	Let $G$ be a compact Hausdorff group. The space of quantum states on $G$ is again $L^2(G)$ with inner product given by the Haar integral. We know, by the Peter-Weyl theorem, that the matrix elements of the finite-dimensional irreducible representations of $G$ are orthogonal, and span $L^2(G)$ as a Hilbert space. We can again choose a full set of finite-dimensional irreps $\tsf{Irr}(G)$, by making use of the axiom of choice, and their matrix elements $\tsf{IB}(G)$. Then, for $\gamma_{m,n},\gamma'_{m',n'}\in\tsf{IB}(G)$, we get the orthogonality relation extending Schur orthogonality
	\begin{align}
		\ang*{\gamma_{m,n},\gamma'_{m',n'}}_G:=\int\overline{\gamma_{m,n}(g)}\gamma'_{m',n'}(g)d_Gg=\frac{1}{d_\gamma}\delta_{\gamma,\gamma'}\delta_{m,m'}\delta_{n,n'}.
	\end{align}
	Also, since we will be summing over elements of $\tsf{IB}(G)$, we can trivially endow it with the $\sigma$-algebra of all sets $\scr{P}(\tsf{IB}(G))$ and the counting measure.

	Let $H\leq G$ be a closed subgroup. Unlike the abelian case, $G/H$ is not necessarily a group, so we will need to fix a set of coset representatives $\tsf{CS}(H)$, again using the axiom of choice. For $g\in G$, write $[g]_H\in\tsf{CS}(H)$ for its representative, dropping the subscript when $H$ is obvious. As the Haar measures of $G$ and $H$ induce a measure on $G/H$, they induce a measure on $\tsf{CS}(H)$. The $\sigma$-algebra is the Borel algebra of the quotient topology $\scr{B}(\tsf{CS}(H))=\set*{E\subseteq\tsf{CS}(H)}{EH\in\scr{B}(G)}$, and the corresponding measure $\mu_{\tsf{CS}(H)}(E)=\mu_G(EH)$. We denote $d[g]=d\mu_{\tsf{CS}(H)}([g])$. This measure interacts well with the Haar measures on $G$ and $H$.

	\begin{lemma}\hphantom{}
		\begin{itemize}
			\item The measurable functions $\tsf{CS}(H)\rightarrow\C$ may be identified with the measurable functions $G\rightarrow\C$ that are constant on the cosets of $H$. If $f:\tsf{CS}(H)\rightarrow\C$ is integrable
			$$\int_{\tsf{CS}(H)}f([g])d[g]=\int_{G}f([g])d_Gg.$$

			\item If $f:G\rightarrow\C$ is integrable, $$\int_G f(g)d_Gg=\int_{\tsf{CS}(H)}\int_Hf([g]h)d_Hhd[g].$$

			\item If $E\in\scr{B}(G)$, then $$\mu_G(E)=\int_{\tsf{CS}(H)}\mu_H([g]^{-1}E\cap H)d[g].$$
		\end{itemize}
	\end{lemma}

	\begin{proof}
		Let $f:\tsf{CS}(H)\rightarrow\C$ be measurable. So for each Borel set of $B\subseteq\C$, $f^{-1}(B)\in\scr{B}(\tsf{CS}(H))$ $\iff$ $f^{-1}(B)H\in\scr{B}(G)$. Letting $f':G\rightarrow\C$ be $f'(g)=f([g])$, $(f')^{-1}(B)=\{g\in G|f([g])\in B\}=f^{-1}(B)H$, so $f'$ is measurable. Conversely, if $f':G\rightarrow\C$ is a measurable function constant on the cosets, let $f''=f'|_{\tsf{CS}(H)}$. Then, $(f'')^{-1}(B)H=((f')^{-1}B\cap\tsf{CS}(H))H=(f')^{-1}(B)\in\scr{B}(G)$. So, $f\leftrightarrow f'$ provides the bijective correspondence we were looking for. Let $f:\tsf{CS}(H)\rightarrow\C$ is integrable. There exists a sequence of simple functions that $(f_n)$ that converges pointwise to $f$. Then, $(f_n')$ converges to $f'$. For some $n$, writing $f_n=\sum_ic_i\chi_{E_i}$, $$\int f_n([g])d[g]=\sum_ic_i\mu_{\tsf{CS}(H)}(E_i)=\sum_ic_i\mu_G(E_iH)=\int f_n'(g)d_Gg.$$ Therefore, $$\int f([g])d[g]=\lim_{n\rightarrow\infty}\int f_n([g])d[g]=\lim_{n\rightarrow\infty}\int f_n'(g)d_Gg=\int f'(g)d_Gg=\int f([g])d_Gg.$$

		Let $f:G\rightarrow\C$ be integrable. Consider the function $G\times H\rightarrow\C$, $(g,h)\mapsto f(gh)$. As the product is continuous, this is a composition of measurable functions, so measurable. Also, by invariance of the Haar measure on $G$, $$\int_H \int_G|f(gh)|d_Ggd_Hh=\int_H \int_G|f(g)|d_Ggd_Hh=\int_G|f(g)|d_Gg,$$
		which means by Tonelli's theorem that the map is integrable on the product measure space. So, the integral $\int f(gh)d_Hh$ exists for almost every $g\in G$ and the map $g\mapsto\int f(gh)d_Hh$ is measurable with
		$$\int_G\int_Hf(gh)d_Hhd_Gg=\int_H \int_Gf(gh)d_Ggd_Hh=\int_Gf(g)d_Gg$$
		by Fubini's theorem. Finally, by invariance of the Haar measure on $H$, $\int_H f(gh)d_Hh$ is constant on the cosets of $H$, so
		$$\int_Gf(g)d_Gg=\int_G\int_Hf(gh)d_Hhd_Gg=\int_{\tsf{CS}(H)}\int_Hf([g]h)d_Hhd[g].$$

		For the last point, take $f:G\rightarrow\C$ to be $f=\chi_E$. Then, $f$ is bounded and measurable, so integrable, giving
		$$\mu_G(E)=\int_G\chi_E(g)d_Gg=\int_{\tsf{CS}(H)}\int_H\chi_E([g]h)d_Hhd[g]=\int_{\tsf{CS}(H)}\mu_{H}([g]^{-1}E\cap H)d[g].$$
	\end{proof}

	Now, similarly to \cref{def:coset-measure-abelian}, we define the coset operator measure for $H$.

	\begin{definition}
		Let $G$ be a compact Hausdorff group and $H\leq G$ be a closed subgroup. The \emph{coset operator measure} is the map $A^H:\scr{B}(\tsf{CS}(H))\otimes\scr{P}(\tsf{IB}(H))\rightarrow\mc{B}(L^2(G))$ defined, for $E=\bigcup_{\gamma_{m,n}}E_{\gamma_{m,n}}\times\{\gamma_{m,n}\}\subseteq\tsf{CS}(H)\times\tsf{IB}(H)$ measurable and $\ket{\phi},\ket{\psi}\in L^2(G)$, as
		\begin{align}
			\braket{\phi}{A^H(E)}{\psi}=\sum_{\gamma_{m,n}}d_\gamma\int_{E_{\gamma_{m,n}}}\ang*{\phi\circ[g],\gamma_{m,n}}_H\ang*{\gamma_{m,n},\psi\circ[g]}_Hd[g],
		\end{align}
		where $[g]$ is seen as the left shift $[g](h)=[g]h$.
	\end{definition}

	It is direct to see that $A^H(E)$ is linear, where it is defined, by linearity of the integral. Next, it is well-defined as a function on $L^2(G)$ as it does not depend on the choice of representative of $\ket{\psi}$. In fact, if $\ket{\psi}=0$, then $\int_H |\psi([g]h)|^2d_Hh=0$ for almost all $[g]\in\tsf{CS}(H)$, so $\ang{\gamma_{m,n},\psi\circ[g]}_H=0$ almost everywhere, and therefore $A^H(E)\ket{\psi}=0$. Positivity is due to the fact that
	\begin{align}
		\braket{\psi}{A^H(E)}{\psi}=\sum_{\gamma_{m,n}}d_{\gamma}\int_{E_{\gamma_{m,n}}}\abs{\ang{\gamma_{m,n},\psi\circ[g]}_H}^2d[g]\geq 0,
	\end{align}
	and that $A^H(E)$ is Hermitian, given that
	\begin{align}
		\braket{\varphi}{A^H(E)}{\psi}=\sum_{\gamma_{m,n}}d_{\gamma}\int_{E_{\gamma_{m,n}}}\overline{\ang{\gamma_{m,n},\varphi\circ[g]}_H\ang{\psi\circ[g],\gamma_{m,n}}_H}d[g]=\overline{\braket{\psi}{A^H(E)}{\varphi}}.
	\end{align}
	To show boundedness, it follows from the polarization identity that we need only show that the collection of $\braket{\psi}{A^H(E)}{\psi}$ for all $\norm{\ket{\psi}}=1$ is bounded. First, note that by monotonicity of the integral, if $E\subseteq E'$, we have $E_{\gamma_{m,n}}\subseteq E'_{\gamma_{m,n}}$ for all $\gamma_{m,n}$, and therefore
	\begin{align}
	\begin{split}
	\braket{\psi}{A^H(E)}{\psi}&=\sum_{\gamma_{m,n}}d_{\gamma}\int_{E_{\gamma_{m,n}}}\abs{\ang{\gamma_{m,n},\psi\circ[g]}_H}^2d[g]\\
	&\leq\sum_{\gamma_{m,n}}d_{\gamma}\int_{E'_{\gamma_{m,n}}}\abs{\ang{\gamma_{m,n},\psi\circ[g]}_H}^2d[g]=\braket{\psi}{A^H(E')}{\psi}.
	\end{split}
	\end{align}
	So, $\braket{\psi}{A^H(E)}{\psi}\leq \braket{\psi}{A^H(\tsf{CS}(H)\times\tsf{IB}(H))}{\psi}$, and using Tonelli's theorem
	\begin{align}
	\begin{split}
	\braket{\psi}{A^H(\tsf{CS}(H)\times\tsf{IB}(H))}{\psi}&=\int_{\tsf{CS}(H)}\sum_{\gamma_{m,n}}d_{\gamma}\ang{\psi\circ[g],\gamma_{m,n}}_H\ang{\gamma_{m,n},\psi\circ[g]}_Hd[g]\\
	&=\int_{\tsf{CS}(H)}\ang{\psi\circ[g],\psi\circ[g]}_Hd[g]\\
	&=\int_{\tsf{CS}(H)}\int_H|\psi([g]h)|^2d_Hhd[g]\\
	&=\int_G|\psi(g)|^2d_Gg=1.
	\end{split}
	\end{align}
	Thus, $\braket{\psi}{A^H(E)}{\psi}\leq 1$ and the positivity implies $\norm{A^H(E)}\leq 1$.

	\begin{lemma}
		$A^H$ is a PVM measure.
	\end{lemma}

	\begin{proof}
		Most importantly, we need that $A^H$ is weakly countably additive. Let $E^1,E^2,\ldots\subseteq \tsf{CS}(H)\times\tsf{IB}(H)$ be a countable disjoint collection of measurable sets. We have that, for each $\gamma_{m,n}\in\tsf{IB}(H)$, $E^1_{\gamma_{m,n}},E^2_{\gamma_{m,n}},\ldots$ are disjoint. So, using monotone convergence,
		\begin{align}
		\begin{split}
		\braket{\psi}{A^H\parens[\Big]{\bigcup_i E^i}}{\psi}&=\sum_{\gamma_{m,n}}d_{\gamma}\int\chi_{\bigcup_iE^i_{\gamma_{m,n}}}([g])\abs{\ang{\gamma_{m,n},\psi\circ[g]}_H}^2d[g]\\
		&=\sum_{\gamma_{m,n}}d_{\gamma}\int\sum_{i=1}^\infty\chi_{E^i_{\gamma_{m,n}}}([g])\abs{\ang{\gamma_{m,n},\psi\circ[g]}_H}^2d[g]\\
		&=\sum_{i=1}^\infty\sum_{\gamma_{m,n}}d_{\gamma}\int\chi_{E^i_{\gamma_{m,n}}}([g])\abs{\ang{\gamma_{m,n},\psi\circ[g]}_H}^2d[g]\\
		&=\sum_{i=1}^\infty\braket{\psi}{A^H(E^i)}{\psi},
		\end{split}
		\end{align}
		where $\chi_{E}$ is the characteristic function, that is $\chi_E(x)=1$ if $x\in E$ and $\chi_E(x)=0$ if $x\notin E$. Therefore, by the polarization identity, $A^H\parens*{\bigcup_i E^i}=\sum_{i=1}^\infty A^H(E^i)$ weakly.

		Next, we want that $A^H$ is projective. From the work before the lemma, we have that $A^H(\tsf{CS}(H)\otimes\tsf{IB}(H))=\Id_{L^2(G)}$, so it is a measurement. To show that $A^H$ is projective, note that the definition directly implies that
		\begin{align}
			(A^H(E)\ket{\psi})(g)=\sum_{\gamma_{m,n}}d_\gamma\chi_{E_{\gamma_{m,n}}}([g])\gamma_{m,n}([g]^{-1}g)\ang{\gamma_{m,n},\psi\circ[g]}_H.
		\end{align}
		So, as everything is nice and bounded, we can invoke Fubini's theorem a couple times to get that
		\begin{align}
		\begin{split}
		&\braket{\psi}{A^H(E)A^H(F)}{\psi}=\int_G\overline{(A^H(E)\ket{\psi})(g)}(A^H(F)\ket{\psi})(g)d_Gg\\
		&=\int_G\sum_{\gamma_{m,n},\gamma'_{m',n'}}d_\gamma d_{\gamma'}\chi_{E_{\gamma_{m,n}}\cap F_{\gamma'_{m',n'}}}([g])\overline{\gamma_{m,n}}([g]^{-1}g)\gamma'_{m',n'}([g]^{-1}g)\ang{\psi\circ[g],\gamma_{m,n}}_H\ang{\gamma'_{m',n'},\psi\circ[g]}_Hd_Gg\\
		&=\sum_{\gamma_{m,n},\gamma'_{m',n'}}d_\gamma d_{\gamma'}\int_{E_{\gamma_{m,n}}\cap F_{\gamma'_{m',n'}}}\ang{\psi\circ[g],\gamma_{m,n}}_H\ang{\gamma'_{m',n'},\psi\circ[g]}_H\ang{\gamma_{m,n},\gamma'_{m',n'}}_Hd_G[g]\\
		&=\sum_{\gamma_{m,n}}d_\gamma\int_{E_{\gamma_{m,n}}\cap F_{\gamma_{m,n}}}\ang{\psi\circ[g],\gamma_{m,n}}_H\ang{\gamma_{m,n},\psi\circ[g]}_Hd_G[g]\\
		&=\braket{\psi}{A^H(E\cap F)}{\psi}.
		\end{split}
		\end{align}
		Therefore, $A^H(E)A^H(F)=A^H(E\cap F)$.
	\end{proof}

	\subsection{Compact coset measure game}

	\begin{definition}
		An \emph{compact coset measure monogamy game} is a tuple $\ttt{G}=\parens*{G,\mc{S},E}$, where $G$ is a compact Hausdorff group, $\mc{S}$ is a finite set of closed subgroups of $G$, and $E\subseteq G$ and is a symmetric neighborhood of identity, \emph{i.e.}\ $E=E^{-1}$ and $1\in E$.

		A \emph{(quantum) strategy} for a coset measure game $\ttt{G}$ is a tuple $\ttt{S}=\parens*{\mc{B},\mc{C},B,C,\rho}$, where $\mc{B}$ and $\mc{C}$ are Hilbert spaces, $B=\set{B^H:\scr{B}(\tsf{CS}(H))\rightarrow\mc{B}(\mc{B})}{H\in\mc{S}}$ and $C=\set{C^H:\scr{P}(\tsf{IB}(H))\rightarrow\mc{B}(\mc{C})}{H\in\mc{S}}$ are collections of POVM measures, and $\rho\in\mc{D}(L^2(G)\otimes\mc{B}\otimes\mc{C})$.

		Let $\ttt{G}$ be a coset measure game and $\ttt{S}$ be a strategy for it. The \emph{winning probability} of $\ttt{S}$ is
		\begin{align}
			\mfk{w}_{\ttt{G}}(\ttt{S})=\expec{H\in\mc{S}}\Tr\squ*{(A^H\otimes B^H\otimes C^H)\parens{E_H}\rho},
		\end{align}
		where $E_H=\set*{([g]_H,\gamma_{m,n},[eg]_H,\gamma_{m,n})}{g\in G,\gamma_{m,n}\in\tsf{IB}(H),e\in E}$ and the expectation with respect to $H\in\mc{S}$ is uniform.

		The winning probability of $\ttt{G}$ is $\mfk{w}(\ttt{G})=\sup_{\ttt{S}}\mfk{w}_{\ttt{G}}(\ttt{S})$.
	\end{definition}

	\begin{theorem}\label{thm:compact-coset-bound}
		Let $\ttt{G}=(G,\mc{S},E)$ be a compact coset measure game. Then, for any complete set of orthogonal permutations $\pi_i:\mc{S}\rightarrow\mc{S}$, the winning probability is bounded by
		\begin{align}
			w(\ttt{G})\leq \expec{i}\sup_{\substack{H\in\mc{S}\\\gamma\in\tsf{Irr}(H)\\g\in G}}\sqrt{d_\gamma\mu_H(H\cap gE\pi_i(H))}.
		\end{align}
	\end{theorem}

	\begin{lemma} \label{lem:compact-overlaps}
		Let $H,K\leq G$ be closed subgroups and let $E\subseteq G$ be a symmetric open neighborhood of the identity. Then, for any $\gamma_{m,n}\in\tsf{IB}(H)$ and $q'\in\tsf{CS}(K)$,
		\begin{align}
		\norm*{A^H\parens*{\tsf{CS}(H)\times\{\gamma_{m,n}\}}A^K\parens*{[Eq']_K\times\tsf{IB}(K)}}\leq\sup_{g\in G}\sqrt{d_\gamma\mu_H(H\cap gEK)}.
		\end{align}
	\end{lemma}

	\begin{proof}
		Since $\tsf{IB}(K)$ gives rise to an orthonormal basis, for any $\ket{\varphi},\ket{\psi}\in L^2(G)$,
		\begin{align}
		\begin{split}
		\braket{\varphi}{A^K\parens*{[Eq']_K\times\tsf{IB}(K)}}{\psi}&=\sum_{\varrho_{i,j}\in\tsf{IB}(K)}d_\varrho\int_{[Eq']_K}\ang*{\varphi\circ[q],\varrho_{i,j}}_K\ang*{\varrho_{i,j},\psi\circ[q]}_Kd[q]\\
		&=\int_{[Eq']_K}\ang*{\varphi\circ[q],\psi\circ[q]}_Kd[q]=\int_{[Eq']_K}\int_K\overline{\varphi}([q]k)\psi([q]k)d_Kkd[q]\\
		&=\int_{Eq'K}\overline{\varphi}(q)\psi(q)d_Gq.
		\end{split}
		\end{align}
		That is, $A^K\parens*{[Eq']_K\times\tsf{IB}(K)}=\Pi_{Eq'K}$, the projector onto the subspace of states $\ket{\psi}$ such that $\psi(g)=0$ for almost all $g\notin Eq'K$. To get the norm, we consider $\norm{A^H\parens*{\tsf{CS}(H)\times\{\gamma_{m,n}\}}\Pi_{Eq'K}\ket{\psi}}$ for some unit vector $\ket{\psi}\in L^2(G)$ and then take the supremum. First,
		\begin{align}
		\begin{split}
		\norm{A^H\parens*{\tsf{CS}(H)\times\{\gamma_{m,n}\}}\Pi_{Eq'K}\ket{\psi}}^2&=\braket{\psi}{\Pi_{Eq'H}A^H\parens*{\tsf{CS}(H)\times\{\gamma_{m,n}\}}\Pi_{Eq'K}}{\psi}\\
		&=d_{\gamma}\int_{\tsf{CS}(H)}\abs*{\ang*{\gamma_{m,n},(\chi_{Eq'K}\psi)\circ[g]}_H}^2d[g].
		\end{split}
		\end{align}
		Now, let $p_{[g]}=\int_H|\psi([g]h)|^2d_Hh$ -- this is finite for almost every $[g]\in\tsf{CS}(H)$ and the map $[g]\mapsto p_{[g]}$ is in $L^1(\tsf{CS}(H))$ -- and then $\ket{\psi_{[g]}}=\frac{\ket{\psi}}{\sqrt{p_{[g]}}}$. In particular, $\int_H|\psi_{[g]}([g]h)|^2d_Hh=1$ for almost all $[g]$, and we have
		\begin{align}
		\begin{split}
		\norm{A^H\parens*{\tsf{CS}(H)\times\{\gamma_{m,n}\}}\Pi_{Eq'K}\ket{\psi}}&=\sqrt{d_{\gamma}\int_{\tsf{CS}(H)}p_{[g]}\abs*{\ang*{\gamma_{m,n},(\chi_{Eq'K}\psi_{[g]})\circ[g]}_H}^2d[g]}\\
		&\leq\sqrt{d_{\gamma}}\sup_{[g]\in\tsf{CS}(H)}\abs*{\ang*{\gamma_{m,n},(\chi_{Eq'K}\psi_{[g]})\circ[g]}_H}\\
		&=\sqrt{d_{\gamma}}\sup_{[g]\in\tsf{CS}(H)}\abs[\Big]{\int_H\overline{\gamma_{m,n}}(h)\chi_{Eq'K}([g]h)\psi_{[g]}([g]h)d_Hh}\\
		&\leq\sqrt{d_{\gamma}}\sup_{[g]\in\tsf{CS}(H)}\int_{H\cap[g]^{-1}Eq'K}|\gamma_{m,n}(h)||\psi_{[g]}([g]h)|d_Hh\\
		&\leq\sqrt{d_{\gamma}}\sup_{[g]\in\tsf{CS}(H)}\int_{H\cap[g]^{-1}Eq'K}|\psi_{[g]}([g]h)|d_Hh.
		\end{split}
		\end{align}
		As $H\cap[g]^{-1}Eq'K$ is a finite measure space $\norm{f}_1\leq\sqrt{\mu_H(H\cap[g]^{-1}Eq'K)}\norm{f}_2$, giving
		\begin{align}
		\begin{split}
		\abs*{\ang*{\gamma_{m,n},(\chi_{Eq'K}\psi)\circ[g]}}&\leq\sqrt{\mu_H(H\cap[g]^{-1}Eq'K)}\int_{H\cap[g]^{-1}Eq'K}|\psi_{[g]}([g]h)|^2d_Hh\\
		&\leq\sqrt{\mu_H(H\cap[g]^{-1}Eq'K)}.
		\end{split}
		\end{align}
		Taking the supremum over $\ket{\psi}$ gives
		\begin{align}
		\norm*{A^H\parens*{\tsf{CS}(H)\times\{\gamma_{m,n}\}}\Pi_{Eq'K}}\leq \sup_{[g]\in\tsf{CS}(H)}\sqrt{d_\gamma\mu_H(H\cap[g]^{-1}Eq'K)}\leq\sup_{g\in G}\sqrt{d_\gamma\mu_H(H\cap gEK)}.
		\end{align}
	\end{proof}

	\begin{proof}[Proof of Theorem~\ref{thm:compact-coset-bound}]
		We proceed similarly to \cref{thm:bound-abelian}. Let $\ttt{S}=\parens*{\mc{B},\mc{C},B,C,\rho}$ be a strategy for $\ttt{G}=\parens*{G,\mc{S},E,F}$, and we assume that $B$ and $C$ are PVMs by Naimark's theorem \cref{thm:naimark}. Writing, for each $H\in\mc{S}$, $\Pi^H=(A^H\otimes B^H\otimes C^H)(E_H)$, we use the overlap result \cref{lem:sum-bound}, getting
		\begin{align}
			\mfk{w}_{\ttt{G}}(\ttt{S})\leq\norm[\Big]{\expec{S\in\mc{S}}\Pi^H}\leq\expec{i}\sup_{H\in\mc{S}}\norm*{\Pi^H\Pi^{\pi_i(H)}},
		\end{align}
		since $(\Pi^H)^2=\Pi^H$. Fixing $H,K\in\mc{S}$, it remains to simplify $\norm*{\Pi^H\Pi^K}$. We have that
		the set of correct answers $E_H\subseteq\{([g]_H,\gamma_{m,n},[g']_H,\gamma_{m,n})|g,g'\in G;\gamma_{m,n}\in\tsf{IB}(H);e\in E\},$ so
		\begin{align}
			(A^H\otimes B^H\otimes C^H)(E_H)\leq(A^H\otimes C^H)\parens*{\{([g]_H,\gamma_{m,n},\gamma_{m,n})|g\in G,\gamma_{m,n}\in\tsf{IB}(H),e\in E}\}\otimes\Id_{B};
		\end{align}
		and $E_K\subseteq\{([eq]_K,\varrho_{i,j},[q]_K,\varrho'_{i',j'})|q\in G;\varrho_{i,j},\varrho'_{i',j'}\in\tsf{IB}(K);e\in E\},$ so
		\begin{align}
			(A^K\otimes B^K\otimes C^K)(E_K)\leq(A^K\otimes B^K)\parens*{\{([eq]_K,\varrho_{i,j},[q]_K)|q\in G,\varrho_{i,j}\in\tsf{IB}(K),e\in E\}}\otimes\Id_{C}.
		\end{align}
		Using the definition of the tensor product operator measure, $\Pi_H\leq\int_{\tsf{IB}(H)}A^H(\tsf{CS}(H)\times\{\gamma_{m,n}\})\otimes\Id_B\otimes dC^H(\gamma_{m,n})$ and $\Pi_K\leq\int_{\tsf{CS}(K)}A^K([Eq]_K\times\tsf{IB}(K))\otimes dB^K([q]_K)\otimes\Id_C$, so
		\begin{align}
		\begin{split}
			\norm{\Pi^H\Pi^K}&\leq\norm[\Big]{\int_{\tsf{CS}(K)\times\tsf{IB}(H)}A^H(\tsf{CS}(H)\times\{\gamma_{m,n}\})A^K([Eq]_K\times\tsf{IB}(K))\otimes d(B^K\otimes C^H)([q]_k,\gamma_{m,n})}\\
			&\leq\sup_{\substack{[q]_K\in\tsf{CS}(K)\\\gamma_{m,n}\in\tsf{IB}(H)}}\norm*{A^H(\tsf{CS}(H)\times\{\gamma_{m,n}\})A^K([Eq]_K\times\tsf{IB}(K))},
		\end{split}
		\end{align}
		using the result of Lemma~\ref{lem:int-bound}. Then, using Lemma~\ref{lem:compact-overlaps}, we get
		\begin{align}
			\norm{\Pi^H\Pi^K}\leq\sup_{\substack{g\in G\\\gamma\in\tsf{Irr}(H)}}\sqrt{d_\gamma\mu_H(H\cap gEK)},
		\end{align}
		 which gives the result.
	\end{proof}

	\appendix

	\section{Integration by an operator-valued measure}\label{sec:POVM}

	In this section, we fix a measurable space $(X,\scr{S})$, a Hilbert space $\mc{H}$, a separable Hilbert space $\mc{K}$, and $P:\scr{S}\rightarrow\mc{B}(\mc{H})$ a POVM measure. We generally follow the standard notation of \cite{Axl20}.

	\begin{definition}
		Let $f:X\rightarrow\C$ be a measurable function. If
		\begin{align}
			\sup\set*{\int|f|d\braket{v}{P}{v}}{\ket{v}\in\mc{H},\norm{\ket{v}}\leq 1}<\infty,
		\end{align}
		we say $f$ is operator integrable and define the integral $\int fdP\in\mc{B}(\mc{H})$ as the operator such that
		\begin{align}
			\braket{u}{\int fdP}{v}=\int fd\braket{u}{P}{v}
		\end{align}
		for all $\ket{u},\ket{v}\in\mc{H}$.
	\end{definition}

	As the complex measure $\braket{u}{P}{v}$ is linear in $\ket{v}$ and antilinear in $\ket{u}$, $\int f dP$ is in fact a linear operator; and, by the polarization identity, $\norm{\int fdP}<\infty$.

	\begin{lemma}
		Let $f=\sum_ic_i\chi_{A_i}$ be a measurable simple function. Then, $\int fdP=\sum_i c_iP(A_i)$.
	\end{lemma}

	\begin{proof}
		Let $\ket{u}, \ket{v}\in\mc{H}$. Then,
		\begin{align}
			\braket{u}{\int fdP}{v}=\sum_ic_i\int\chi_{A_i}d\braket{u}{P}{v}=\braket{u}{\sum_ic_iP(A_i)}{v}
		\end{align}
	\end{proof}

	We can, as for a usual measure, approximate the operator integral of a function by the integrals of simple functions that approximate it. Let $(f_n)$ be a sequence of simple functions that converges pointwise to $f$. Then, as for each $\ket{u},\ket{v}\in\mc{H}$, $\int f_nd\braket{u}{P}{v}\rightarrow\int fd\braket{u}{P}{v}$, the sequence $\parens*{\int f_ndP}$ in $\mc{B}(\mc{H})$ converges \emph{weakly} to $\int f dP$.

	We would like to able to integrate not just scalar functions, but operator-valued functions by a POVM measure. The following proceeds similarly to Pettis' integral \cite{Yos95}.

	\begin{definition}
		We say that $F:X\rightarrow\mc{B}(\mc{K})$ is \emph{weakly measurable} if, for every $\ket{u},\ket{v}\in\mc{K}$, $\braket{u}{F}{v}$ is measurable.
	\end{definition}

	\begin{lemma}
		Let $\Delta\geq 0$. The open ball
		\begin{align}
			B_\Delta:=\set*{A\in\mc{B}(\mc{K})}{\norm{A}<\Delta}
		\end{align}
		is separable in the strong operator topology.
	\end{lemma}

	In particular, $\mc{B}(\mc{K})=\bigcup_{N=1}^\infty B_N$ is separable.

	\begin{proof}
		Let $D$ be a countable dense subset of $\mc{K}$ and $\set*{\ket{n}}{n\in\N}$ be an orthonormal basis. Set
		\begin{align}
			D_\Delta=\set{\ketbra{v_1}{1}+\ldots+\ketbra{v_n}{n}}{n\in\N;\ket{v_1},\ldots,\ket{v_n}\in D}\cap B_\Delta.
		\end{align}
		$D_\Delta$ is countable as $D$ is countable, and I claim $D_\Delta$ is dense in $B_\Delta$. Let $A\in B_\Delta$. By density of $D$, for each $n\in\N$ and $i\leq n$, there exists $\ket{v^n_i}\in D$ such that $\norm{\ket{v^n_i}-A\ket{i}}<\min\set*{\frac{1}{n^2},\frac{\Delta-\norm{A}}{n}}$. Set $A_n=\sum_{i=1}^n\ketbra{v^n_i}{i}$, so we have
		\begin{align}
		\begin{split}
			\norm{A_n}\leq\norm[\Big]{A\sum_{i=1}^n\ketbra{i}}+\sum_{i=1}^n\norm{\ket{v^n_i}-A\ket{i}}<\norm{A}+(\Delta-\norm{A})=\Delta,
		\end{split}
		\end{align}
		giving $A_n\in D_\Delta$. I claim now that the sequence $(A_n)$ converges strongly to $A$. Let $\varepsilon>0$ and $\ket{v}\in\mc{K}$. There exists $N\in\N$ such that both $\sum_{i=N+1}^\infty\abs*{\braket{i}{v}}^2<\parens*{\frac{\varepsilon}{2\Delta}}^2$ and $\frac{\norm{\ket{v}}}{N}<\frac{\varepsilon}{2}$. Then, for $n\geq N$,
		\begin{align}
		\begin{split}
			\norm{(A-A_n)\ket{v}}&\leq\norm[\Big]{\sum_{i=1}^n(A\ket{i}-\ket{v^n_i})\braket{i}{v}}+\norm[\Big]{A\sum_{i=n+1}^\infty\braket{i}{v}\ket{i}}\\
			&\leq\sum_{i=1}^n\norm{A\ket{i}-\ket{v^n_i}}\norm{v}+\norm{A}\sqrt{\sum_{i=n+1}^\infty\abs*{\braket{i}{v}}^2}\\
			&<\frac{\norm{v}}{n}+\Delta\frac{\varepsilon}{2\Delta}<\varepsilon.
		\end{split}
		\end{align}
	\end{proof}

	\begin{lemma}\label{lem:simp-limit}
		Let $F:X\rightarrow\mc{B}(\mc{K})$. Then, there exists a sequence of simple functions $(F_n)$ that converges strongly pointwise to $F$. Additionally, if $F$ is weakly measurable, then each $F_n$ can be chosen to be weakly measurable; and if $F$ is bounded, then for all $\delta>0$, $(F_n)$ can be chosen so that $\norm{F_n(x)}<\sup_x\norm{F(x)}+\delta$.
	\end{lemma}

	\begin{proof}
		Let $\delta>0$. Then, if $F$ is bounded, we can corestrict it to $U=B_{\sup_x \norm{F(X)}+\delta}$; and if $F$ is not bounded, we do not change its codomain, $U=\mc{B}(\mc{K})$. Since $\mc{K}$ is separable, there exists a countable dense sequence $(\ket{v_n})_{n\in\N}$ in $\mc{K}$. For $n\in\N$ and $M\in\mc{B}(\mc{K})$, define the cylinder sets
		\begin{align}
			B_n(M)=\set*{A\in U}{\norm{A}<\norm{M}+1;\norm{(M-A)\ket{v_i}}<\frac{1}{n}\text{ for }i=1,\ldots,n}.
		\end{align}
		By separability of $U$, there exists a countable collection $M^n_1,M^n_2,\ldots\in U$ such that $U\subseteq\bigcup_{i=1}^\infty B_n(M^n_i)$. Next, define the sets $E^n_m$ as
		\begin{align}
			E^n_m=F^{-1}\parens*{B_n(M^n_m)}\backslash\bigcup_{i=1}^{m-1}F^{-1}\parens*{B_n(M^n_i)},
		\end{align}
		which are in $\scr{S}$ if $F$ is weakly measurable, and set $F^n_m=\sum_{i=1}^m M^n_m\chi_{E^n_m}$. Finally, define $F_n:X\rightarrow U$ as $F_n(x)=F^i_n(x)$, where $i\leq n$ is the largest such that $F^i_n(x)\neq 0$, and $F_n(x)=0$ if no such $i$ exists. Then, $F_n$ is weakly measurable if $F$ is, takes at most $n^2$ values, so is simple, and $F_n(X)\subseteq U$.

		It remains to show that $(F_n)$ converges pointwise strongly to $F$. Let $\varepsilon>0$, $x\in X$, and $\ket{v}\in\mc{K}$. Then, there exists $N\in\N$ such that $\frac{1}{N}<\frac{\varepsilon}{2}$ and $\norm{\ket{v_N}-\ket{v}}<\frac{\varepsilon}{2(2\norm{F(x)}+1)}$ As $F(x)\in U$, there exists $i$ such that $F(x)\in B_n(M^N_i)$. Let $n\geq\max\{i,N\}$. Then, there exists $m\geq N$ and $j\leq n$ such that $F_n(x)=M^m_j$. We must have $F(x)\in B_m(M^m_j)$, so
		\begin{align}
		\begin{split}
			\norm{(F(x)-F_n(x))\ket{v}}&\leq\norm{(F(x)-M^m_j)\ket{v_N}}+\norm{F(x)-M^m_j}\norm{\ket{v}-\ket{v_N}}\\
			&<\frac{1}{m}+(2\norm{F(x)}+1)\frac{\varepsilon}{2(2\norm{F(x)}+1)}\\
			&<\varepsilon.
		\end{split}
		\end{align}
	\end{proof}

	\begin{definition}
		Let $F:X\rightarrow\mc{B}(\mc{K})$. If the set
		\begin{align}
			\set*{\sum_{i,j=1}^m\sqrt{p_ip_j}\int\braket{k_i}{F}{k_j}d\braket{h_i}{P}{h_j}}{m\in\N;\sum_{i=1}^mp_i\leq1;\{\ket{h_i}\}\subseteq\mc{H},\{\ket{k_i}\}\subseteq\mc{K}\text{ orthonormal}}
		\end{align}
		is bounded in $\C$, we say $F$ is \emph{weakly operator integrable}. In that case, we take $\int F\otimes dP\in\mc{B}(\mc{K}\otimes\mc{H})$ to be the operator such that
		\begin{align}
			\bra{k}\otimes\bra{h}\int F\otimes dP\ket{k'}\otimes\ket{h'}=\int\braket{k}{F}{k'}d\braket{h}{P}{h'},
		\end{align}
		for any $\ket{h},\ket{h'}\in\mc{H}$ and $\ket{k},\ket{k'}\in\mc{K}$.
	\end{definition}

	As for integration of scalar-valued functions. this is well-defined as a bounded linear operator via the polarization identity and linearity of integrals with respect to complex measures.

	\begin{lemma}
		Let $F=\sum_lM_l\chi_{E_l}$ be a simple function. Then, $F$ is weakly operator integrable and
		\begin{align}
			\int F\otimes dP=\sum_l M_l\otimes P(E_l).
		\end{align}
	\end{lemma}

	\begin{proof}
		Let $\set{\ket{h_i}}{i=1,\ldots,m}\subseteq\mc{H}$ and $\set{\ket{k_i}}{i=1,\ldots,m}\subseteq\mc{K}$ be orthonormal, and $\sum_{i=1}^mp_i\leq 1$. Write $\ket{v}=\sum_{i=1}^m\sqrt{p_i}\ket{k_i}\otimes\ket{h_i}\in\mc{K}\otimes\mc{H}$. Then, for each $i,j$, $\sum_l\braket{k_i}{F}{k_j}\chi_{E_l}$ is simple and weakly operator integrable, so
		\begin{align}
		\begin{split}
			\abs*{\sum_{i,j=1}^m\sqrt{p_ip_j}\int\braket{k_i}{F}{k_j}d\braket{h_i}{P}{h_j}}&=\abs*{\sum_{i,j=1}^m\sqrt{p_ip_j}\sum_{l}\braket{k_i}{M_l}{k_j}\braket{h_i}{P(E_l)}{h_j}}\\
			&=\abs*{\braket{v}{\sum_l M_l\otimes P(E_l)}{v}}\\
			&\leq\max_l\norm*{M_l}<\infty.
		\end{split}
		\end{align}
		Thus, $F$ is weakly operator integrable and its integral is $\sum_l M_l\otimes P(E_l)$.
	\end{proof}

	\begin{lemma}
		Suppose $F:X\rightarrow\mc{B}(\mc{K})$ is bounded and weakly measurable. Then, $F$ is weakly operator integrable.
	\end{lemma}

	\begin{proof}
		Fix $m\in\N$, $\set*{\ket{h_i}}{i=1,\ldots,m}\subseteq\mc{H}$ and $\set*{\ket{k_i}}{i=1,\ldots,m}\subseteq\mc{K}$ orthonormal, and $p_i\geq 0$ such that $\sum_{i=1}^mp_i\leq 1$; set $S=\abs*{\sum_{i,j=1}^m\sqrt{p_ip_j}\int\braket{k_i}{F}{k_j}d\braket{h_i}{P}{h_j}}$. To show that $F$ is weakly operator integrable, it suffices to show that we can upper bound $S$ by a constant independent of $m$, or the $\ket{h_i}$, $\ket{k_i}$, and $p_i$. Using \cref{lem:simp-limit}, let $\parens*{F_n=\sum_{l}M_l^n\chi_{E^n_l}}$ be a sequence of simple functions that converges strongly pointwise to $F$ and $\norm{F_n(x)}\leq\sup_y\norm{F(y)}+1$ for all $x\in X$. In particular, $\braket{k_i}{F_n}{k_j}\rightarrow\braket{k_i}{F}{k_j}$ pointwise for all $i,j$, so there exists $N\in\N$ such that for all $n\geq N$, \begin{align}
			\abs[\Big]{\int\braket{k_i}{F_n-F}{k_j}d\braket{h_i}{P}{h_j}}\leq\sqrt{p_ip_j}.
		\end{align}
		As such, writing $\ket{v}=\sum_{i=1}^n\sqrt{p_i}\ket{k_i}\otimes\ket{h_i}$, we get that
		\begin{align}
		\begin{split}
			S&\leq\abs[\Big]{\sum_{i,j=1}^m\sqrt{p_ip_j}\int\braket{k_i}{F_n}{k_j}d\braket{h_i}{P}{h_j}}+\sum_{i,j}^mp_ip_j=\abs[\Big]{\braket{v}{\sum_{l}M^n_l\otimes P(E^n_l)}{v}}+1\\
			&\leq\sup_l\norm{M^n_l}+1=\sup_{x\in X}\norm{F_n(x)}+1,
		\end{split}
		\end{align}
		and thus $S\leq\sup_x\norm{F(x)}+2$, giving the result.
	\end{proof}

	\begin{theorem}\label{thm:int-simp-limit}
		Let $F:X\rightarrow\mc{B}(\mc{K})$ be norm bounded and weakly measurable. For any sequence of weakly operator integrable simple functions $(F_n)$ that converges to $F$ strongly pointwise and $\sup_n\norm*{\int F_n\otimes dP}<\infty$, we have that $\int F_n\otimes dP\rightarrow\int F\otimes dP$ weakly.
	\end{theorem}

	Note that \cref{lem:simp-limit} shows that such a sequence of simple functions exists.

	\begin{proof}
		Write $F_n=\sum_l M^n_l\chi_{E^n_l}$ and $W=\max\set*{\norm{\int F\otimes dP},\sup_{n}}<\infty$. Fix $\ket{v}\in\mc{K}\otimes\mc{H}$ and $\varepsilon>0$. Via the Schmidt decomposition, we may write $\ket{v}=\sum_{i}\sqrt{p_i}\ket{k_i}\otimes\ket{h_i}$. As the norm of $\ket{v}$ is finite, there exists $I\in\N$ such that $\sum_{i=I+1}^\infty p_i<\min\{\frac{\varepsilon}{4(2\norm{\ket{v}}+1)W},1\}$; let $\ket{v'}=\sum_{i=I+1}^\infty\sqrt{p_i}\ket{k_i}\otimes\ket{h_i}$. Next, as $\braket{k_i}{F_n}{k_j}\rightarrow\braket{k_i}{F}{k_j}$ pointwise, there exists $N\in\N$ such that for all $n\geq N$, $\abs*{\int\braket{k_i}{F_n-F}{k_j}d\braket{h_i}{P}{h_j}}\leq\frac{\varepsilon}{2I\norm{\ket{v}}^2}$ for all $i,j\leq I$. Then, for any $n\geq N$,
		\begin{align}
		\begin{split}
			\abs[\Big]{\braket{v}{\int(F_n-F)\otimes dP}{v}}&\leq\abs[\Big]{\sum_{i,j=1}^I\sqrt{p_ip_j}\int\braket{k_i}{F_n-F}{k_j}d\braket{h_i}{P}{h_j}}\\
			&\qquad+(2\norm*{\ket{v}}+1)\norm{\ket{v'}}\norm[\Big]{\int(F_n-F)\otimes dP}\\
			&\leq\sum_{i,j=1}^I\sqrt{p_ip_j}\abs[\Big]{\int\braket{k_i}{F_n-F}{k_j}d\braket{h_i}{P}{h_j}}+(2\norm*{\ket{v}}+1)\norm{\ket{v'}}2W\\
			&<\parens[\Big]{\sum_{i=1}^I\sqrt{p_i}}^2\frac{\varepsilon}{2I\norm{\ket{v}}^2}+\frac{\varepsilon}{2}\leq\varepsilon.
		\end{split}
		\end{align}
	\end{proof}

	\begin{definition}
		Let $(X,\scr{S},P)$ and $(Y,\scr{T},Q)$ be POVM operator measure spaces. Then, the \emph{tensor product measure} on the measurable space $(X\times Y,\scr{S}\otimes\scr{T})$ is
		\begin{align}
			(P\otimes Q)(E)=\int_Y P((E)^y)\otimes dQ(y),
		\end{align}
		where $(E)^y=\set*{x\in X}{(x,y)\in E}$ is the cross section of $E$ at $Y$.
	\end{definition}

	This definition is sensible since the map $y\mapsto P((E)^y)$ is weakly measurable and bounded in norm by $1$. This definition allows us to immediately generalise Fubini's theorem, as the matrix elements must all satify it: for any weakly operator integrable $F:X\times Y\rightarrow \mc{B}(\mc{K})$,
	\begin{align}
		\int F\otimes d(P\otimes Q)=\int_Y\parens[\Big]{\int_X F(x,y)\otimes dP(x)}\otimes dQ(y)=\int_X\int_Y F(x,y)\otimes dP(x)\otimes dQ(y).
	\end{align}

	\section{Damping operators and maximally-entangled states}\label{sec:damping}

	In this section, let $\mc{H}$ be any Hilbert space. Write $\mc{T}_1(\mc{H})$ for the set of trace-class operator and for $T\in\mc{T}_1(\mc{H})$ the trace norm $\norm{T}_1=\Tr(\sqrt{A^\dag A})$; and write $\mc{T}_2(\mc{H})$ for the Hilbert-Schmidt operators and for $T\in\mc{T}_2(\mc{H})$ the Hilbert-Schmidt norm $\norm{T}_2=\sqrt{\Tr(A^\dag A)}$.

	\begin{definition}
		A \emph{complex conjugate} on $\mc{H}$ is an antilinear involutive isometry $c:\mc{H}\rightarrow\mc{H}$, \emph{i.e.}\ $c(\alpha\ket{\psi}+\beta\ket{\phi})=\overline{\alpha}c\ket{\psi}+\overline{\beta}c\ket{\phi}$, $c^2=1$, and $\norm*{c\ket{\psi}}=\norm*{\ket{\psi}}$ for all $\alpha,\beta\in\C$ and $\ket{\psi},\ket{\phi}\in\mc{H}$.
	\end{definition}

	We quickly work out some properties of a complex conjugate, and use that to show that any complex conjugate as defined above may be expressed as a complex conjugate with respect to some basis --- the map that acts by taking the conjugates of the coefficients in the expansion with respect to a fixed orthonormal basis. For any $\ket{\psi}\in\mc{H}$, $\ket{\psi}=\frac{1}{2}(\ket{\psi}+c\ket{\psi})+\frac{1}{2}(\ket{\psi}-c\ket{\psi})$ so we may decompose $\mc{H}=\mc{H}_+\oplus \mc{H}_-$ where $H_\pm=\set*{\ket{\psi}\in H}{c\ket{\psi}=\pm\ket{\psi}}$. As $c$ is $\R$-linear, $\mc{H}_+$ and $\mc{H}_-$ are $\R$-vector spaces, and $\mc{H}_-=i\mc{H}_+$. For $\ket{\psi},\ket{\phi}\in\mc{H}$
	\begin{align}
		\braket{c\psi}{c\phi}=\frac{1}{4}\sum_{k=0}^3i^k\norm{c\ket{\psi}+i^kc\ket{\phi}}^2=\frac{1}{4}\sum_{k=0}^3i^k\norm{\ket{\psi}+(-i)^k\ket{\phi}}^2=\overline{\braket{\psi}{\phi}},
	\end{align}
	which implies that $\braket{c\psi}{\phi}=\braket{c\psi}{c^2\phi}=\overline{\braket{\psi}{c\phi}}$. Thus, for $\ket{\psi},\ket{\phi}\in\mc{H}_+$, $\braket{\psi}{\phi}=\overline{\braket{\psi}{\phi}}$, so $\mc{H}_+$ is a real inner product space. So, there exists an orthonormal basis $\set*{\ket{n}}{n\in\Gamma}\subseteq\mc{H}_+$; this basis is countable iff $\mc{H}_+$ is separable, which is iff $\mc{H}$ is. Using the decomposition, $\mc{H}$ is the closure of $\spn_\C\set*{\ket{n}}{n\in\Gamma}$. Thus, we may expand any $\ket{\psi}\in\mc{H}$ as $\ket{\psi}=\sum_{n\in\Gamma}\psi_n\ket{n}$, so
	\begin{align}
		c\ket{\psi}=\sum_{n\in\Gamma}\overline{\psi}_nc\ket{n}=\sum_{n\in\Gamma}\overline{\psi}_n\ket{n},
	\end{align}
	giving that any complex conjugate may be expressed as the complex conjugate with respect to a basis.

	\begin{definition}\hphantom{}
		\begin{itemize}
			\item A \emph{damping sequence} is a sequence $(\Delta_n)$ in $\mc{T}_2(\mc{H})$ such that $\norm{\Delta_n}\leq 1$, and $\Delta_n\ket{\psi}\rightarrow\ket{\psi}$ and $\Delta_n^\dag\ket{\psi}\rightarrow\ket{\psi}$ for all $\ket{\psi}\in\mc{H}$.

			\item The \emph{pure tensor norm} on $\mc{H}\otimes\mc{H}$ is $\norm{\ket{\Psi}}_\otimes=\sup\set[\big]{\abs*{(\bra{\psi}\otimes\bra{\phi})\ket{\Psi}}}{\ket{\psi},\ket{\phi}\in\mc{H};\norm{\ket{\psi}},\norm{\ket{\phi}}\leq 1}$.

			\item Let $c:\mc{H}\rightarrow\mc{H}$ be a complex conjugate. An \emph{approximate maximally entangled state} is a sequence $(\ket{\Psi_n})$ in $\mc{H}\otimes\mc{H}$ such that $\norm{\ket{\Psi_n}}=1$, and $\frac{1}{\norm{\ket{\Psi_n}}_\otimes}(\bra{\psi}\otimes\Id)\ket{\Psi_n}\rightarrow c\ket{\psi}$ and $\frac{1}{\norm{\ket{\Psi_n}}_\otimes}(\Id\otimes\bra{\psi})\ket{\Psi_n}\rightarrow c\ket{\psi}$ for all $\ket{\psi}\in\mc{H}$.
		\end{itemize}
	\end{definition}

	It is direct to see that every damping sequence generates an approximate maximally entangled state, and vice versa.

	\begin{lemma}\hphantom{}
	\begin{itemize}
		\item Let $(\Delta_n)$ be a damping sequence. Writing the singular-value decompositions $\Delta_n=\sum_{i=1}^\infty s_{n,i}\ketbra{\phi_{n,i}}{\chi_{n,i}}$, the sequence $(\ket{\Psi_n})$ defined as $\ket{\Psi_n}=\frac{1}{\norm{\Delta_n}_2}\sum_{i=1}^\infty s_{n,i}\ket{\phi_{n,i}}\otimes c\ket{\chi_{n,i}}$ is an approximate maximally entangled state.

		\item Let $(\ket{\Psi_n})$ be an approximate maximally entangled state. Writing the Schmidt decomposition $\ket{\Psi_n}=\sum_{i=1}^\infty\sqrt{p_{n,i}}\ket{\phi_{n,i}}\otimes\ket{\chi_{n,i}}$, the sequence $(\Delta_n)$ defined as $\Delta_n=\frac{1}{\norm{\ket{\Psi_n}}_\otimes}\sum_{i=1}^\infty\sqrt{p_{n,i}}\ketbra{\phi_{n,i}}{c\chi_{n,i}}$ is a damping sequence.
	\end{itemize}
	\end{lemma}

	\begin{proof}
		For the first point, note first that $\braket{\Psi_n}=\frac{1}{\Tr(\Delta_n^\dag\Delta_n)}\sum_{i=1}^\infty s_{n,i}^2=1$. Next, for any $\ket{\psi}\in\mc{H}$,
		\begin{align}
			\parens*{\bra{\psi}\otimes\Id}\ket{\Psi_n}=\frac{1}{\norm{\Delta_n}_2}\sum_{i=1}^\infty s_{n,i}\braket{\psi}{\phi_{n,i}}c\ket{\chi_{n,i}}=\frac{1}{\norm{\Delta_n}_2}c\sum_{i=1}^\infty s_{n,i}\ket{\chi_{n,i}}\!\!\braket{\phi_{n,i}}{\psi}=c\frac{\Delta_n^\dag}{\norm{\Delta_n}_2}\ket{\psi}.
		\end{align}
		Since $\norm*{\ket{\Psi_n}}_\otimes=\frac{\norm{\Delta_n}}{\norm{\Delta_n}_2}$ and $\norm{\Delta_n}\rightarrow 1$, we have that $\frac{1}{\norm{\ket{\Psi_n}}_\otimes}\parens*{\bra{\psi}\otimes\Id}\ket{\Psi_n}=c\frac{\Delta_n^\dag}{\norm{\Delta_n}}\ket{\psi}\rightarrow c\ket{\psi}$. In the same way,
		\begin{align}
			\parens*{\Id\otimes\bra{\psi}}\ket{\Psi_n}=\frac{1}{\norm{\Delta_n}_2}\sum_{i=1}^\infty s_{n,i}\braket{\psi}{c\chi_{n,i}}\ket{\phi_{n,i}}=\frac{1}{\norm{\Delta_n}_2}\sum_{i=1}^\infty s_{n,i}\ket{\phi_{n,i}}\braket{\chi_{n,i}}{c\psi}=\frac{\Delta_n}{\norm{\Delta_n}_2}c\ket{\psi},
		\end{align}
		and hence $\frac{1}{\norm{\ket{\Psi_n}}_\otimes}\parens*{\bra{\psi}\otimes\Id}\ket{\Psi_n}=\frac{\Delta_n}{\norm{\Delta_n}}c\ket{\psi}\rightarrow\ket{\psi}$.

		For the second point, we have that $\norm{\Delta_n}=\frac{\norm{\ket{\Psi_n}}_\otimes}{\norm{\ket{\Psi_n}}_\otimes}=1$ and also $\norm{\Delta_n}_2^2=\frac{1}{\norm{\ket{\Psi_n}}_\otimes^2}\sum_{i=1}^\infty p_{n,i}=\frac{1}{\norm{\ket{\Psi_n}}_\otimes^2}<\infty$, so it is well-defined. Finally, for any $\ket{\psi}\in\mc{H}$,
		\begin{align}
			\Delta_n\ket{\psi}=\frac{1}{\norm{\ket{\Psi_n}}_\otimes}\sum_{i=1}^\infty\sqrt{p_{n,i}}\ket{\phi_{n,i}}\!\!\braket{c\psi}{\chi_{n,i}}=\frac{1}{\norm{\ket{\Psi_n}}_\otimes}(\Id\otimes\bra{c\psi})\ket{\Psi_n}\rightarrow c^2\ket{\psi}=\ket{\psi}
		\end{align}
		and
		\begin{align}
			\Delta_n^\dag\ket{\psi}=\frac{1}{\norm{\ket{\Psi_n}}_\otimes}c\sum_{i=1}^\infty\sqrt{p_{n,i}}\ket{\chi_{n,i}}\!\!\braket{\psi}{\phi_{n,i}}=c\frac{1}{\norm{\ket{\Psi_n}}_\otimes}(\bra{\psi}\otimes\Id)\ket{\Psi_n}\rightarrow c^2\ket{\psi}=\ket{\psi}.
		\end{align}
	\end{proof}

	Finally, we use damping sequences to construct states from operator-valued measures.

	\begin{definition}
		Let $(X,\scr{S},\mu)$ be a measure space and $P:\scr{S}\rightarrow\mc{B}(\mc{H})$ be an operator-valued measure. A damping sequence $(\Delta_n)$ in $\mc{B}(\mc{H})$ \emph{$\mu$-damps $P$} if $\norm{\Delta_n^\dag P\Delta_n}_1\ll \mu$, \emph{i.e.}\ if $\mu(E)=0$ then $\Delta_n^\dag P(E)\Delta_n=0$.
	\end{definition}

	Note first that this is well-defined. As $\Delta_n\in\mc{T}_2(\mc{H})$, $P(E)\Delta_n\in\mc{T}_2(\mc{H})$, and therefore $\Delta_n^\dag P(E)\Delta_n\in\mc{T}_1(\mc{H})$, so the trace norm is defined on this operator. Also, $(\Delta_n^\dag P(E)\Delta_n)$ converges weakly to $P(E)$, so the sequence can be seen to approximate the operator-valued measure. This is because, for any $\ket{\phi},\ket{\psi}\in\mc{H}$,
	\begin{align}
	\begin{split}
		\abs[\big]{\braket{\psi}{\Delta_n^\dag P(E)\Delta_n-P(E)}{\phi}}&\leq\abs[\big]{\braket{\psi}{\Delta_n^\dag P(E)\Delta_n-\Delta_n^\dag P(E)}{\phi}}+\abs[\big]{\braket{\psi}{\Delta_n^\dag P(E)-P(E)}{\phi}}\\
		&\leq\norm[\big]{P(E)^\dag\Delta_n\ket{\psi}}\norm[\big]{(\Delta_n-\Id)\ket{\phi}}+\norm[\big]{\ket{\psi}}\norm[\big]{(\Delta_n^\dag-\Id)P(E)\ket{\phi}}\\
		&\rightarrow 0.
	\end{split}
	\end{align}

	The definition of $\mu$-damping is meant to be a way to formalise the idea of damping unnormalizable states. That is, for Casimir operator $C$, we have a candidate damping sequence $(e^{-\frac{C}{2n^2}})$, where $\frac{1}{n}$ represents the damping strength, which we refer to as \emph{Casimir damping}. In order to extract a state-valued function $X\rightarrow\mc{D}(\mc{H})$ from the damped measure $\Delta_n^\dag P\Delta_n$, we may make use of the Radon-Nikodym theorem.

	\begin{lemma} \label{lemma:radon-nikodym}
			Let $P:\scr{S}\rightarrow\mc{B}(\mc{H})$ be a POVM measure, and let $(\Delta_n)$ be a sequence in $\mc{B}(\mc{H})$ that $\mu$-damps $P$, for $\mu:\scr{S}\rightarrow[0,\infty]$ a $\sigma$-finite measure. Then, for each $n$, there exists a measurable function $\rho_n:X\rightarrow\mc{D}(\mc{H})$ and an integrable function $\pi_n:X\rightarrow[0,\infty)$ such that, for all $E\in\scr{S}$,
			\begin{align}
				\Delta_n^\dag P(E)\Delta_n=\int_E\rho_n\pi_nd\mu
			\end{align}
	\end{lemma}

	\begin{proof}
		As $\norm{\Delta_n^\dag P\Delta_n}=\Tr(\Delta_n^\dag P\Delta_n)\ll \mu$, this implies that, in particular, $\braket{\psi}{\Delta_n^\dag P\Delta_n}{\phi}\ll \mu$ for all $\ket{\psi},\ket{\phi}\in\mc{H}$. Thus, by the Radon-Nikodym theorem, there exist integrable functions $f_{\psi,\phi}:X\rightarrow\C$ and $g:X\rightarrow[0,\infty)$ such that $\braket{\psi}{\Delta_n^\dag P(E)\Delta_n}{\phi}=\int f_{\psi,\phi}d\mu$ and $\Tr(\Delta_n^\dag P\Delta_n)=\int gd\mu$. Then, it is direct to see that, for almost every $x\in X$, $(\ket{\psi},\ket{\phi})\mapsto f_{\psi,\phi}(x)$ is linear in $\ket{\phi}$ and antilinear in $\ket{\psi}$. Then, there exists a function $F:X\rightarrow\mc{L}(\mc{H})$ such that $\braket{\psi}{F(x)}{\phi}=f_{\psi,\phi}(x)$ almost everywhere. Next, as $\norm{F(x)}_1=g(x)=\Tr(F(x))$ almost everywhere, we may assume that $F$ is bounded, trace class, and positive. Finally, we take $\rho_n(x)=\frac{F(x)}{\Tr(F(x))}$ and $\pi_n(x)=\Tr(F(x))$.
	\end{proof}

%	\setcounter{tocdepth}{1}
%	\tableofcontents

	\bibliography{monogamy}

\end{document}